%% file: coherent_expectation_paper012222.tex
\documentclass[tightenlines,superscriptaddress,eqsecnum,floats,nofootinbib,showpacs]{revtex4}

\usepackage{amsmath,amssymb,amsfonts,amsthm,mathtools,stmaryrd,mathrsfs}
\usepackage{slashed}
\usepackage{graphicx}
\usepackage{mathtools}
\usepackage{xcolor}

\usepackage{titlesec}
\titleformat{\paragraph}
{\normalfont\normalsize\centering\itshape}{\theparagraph}{1em}{}
\titlespacing*{\paragraph}
{0pt}{4.5ex plus 1ex minus .2ex}{4.5ex plus .2ex}

\usepackage{hyperref} % for HyperTeX cross-referencing

\newtheorem{thm}{Theorem}[section]

\newtheorem{pro}{Proposition}
\newtheorem{lmm}{Lemma}[section]

\newcommand{\makeSymbol}[1]{\mathord{\vcenter{\hbox{#1}}}}
\newcommand{\Pl}{\ell_{\mathrm{p}}} % Planck length
\newcommand{\sgn}{\mathrm{sgn}} % sign
\newcommand{\tr}{\mathrm{tr}} % trace
\newcommand{\dd}{{\rm d}}% derivative
\newcommand{\muh}{{\mu_{ \rm h}}} 
\newcommand{\g}{\mathbf{g} } 
\newcommand{\wdtau}{\mathrm{WDt}} 
\newcommand{\wdtaus}{\mathrm{WD}} 

\DeclarePairedDelimiter\floor{\lfloor}{\rfloor}

\begin{document}
\title{ First-Order Quantum Correction in Coherent State Expectation Value of Loop-Quantum-Gravity Hamiltonian}
\author{Cong Zhang\footnote{czhang(AT)fuw.edu.pl}}
\affiliation{Faculty of Physics, University of Warsaw, Pasteura 5, 02-093 Warsaw, Poland}
\author{Shicong Song\footnote{ssong2019(AT)fau.edu}}
\affiliation{Department of Physics, Florida Atlantic University, 777 Glades Road, Boca Raton, FL 33431-0991, USA}
\author{\ Muxin Han\footnote{hanm(At)fau.edu}}  
\affiliation{Department of Physics, Florida Atlantic University, 777 Glades Road, Boca Raton, FL 33431-0991, USA}
\affiliation{Institut f\"ur Quantengravitation, Universit\"at Erlangen-N\"urnberg, Staudtstr. 7/B2, 91058 Erlangen, Germany}
\begin{abstract}
Given the non-graph-changing Hamiltonian $\widehat{H[N]}$ in Loop Quantum Gravity (LQG), $\langle\widehat{H[N]}\rangle$, the coherent state expectation value  of $\widehat{H[N]}$, admits an semiclassical expansion in $\ell^2_{\rm p}$. 
In this paper, we
explicitly compute the expansion of $\langle\widehat{H[N]}\rangle$ to the linear order in $\ell^2_{\rm p}$ on the cubic graph with respect to the coherent state peaked at the homogeneous and isotropic data of cosmology. 
In our computation, 
a powerful algorithm, supported by rigorous proofs and several theorems,  is developed to overcome the complexity in the computation of $\langle \widehat{H[N]} \rangle$.
Particularly, some key innovations in our algorithm
substantially reduce the complexity in computing the Lorentzian part of $\langle\widehat{H[N]}\rangle$.
Moreover, with the algorithm developed in the present work, we can compute the expectation value of arbitrary monomial of holonomies and fluxes on one edge up to arbitrary order of $\Pl^2$. 
Finally, some quantum correction effects resulted from $\langle\widehat{H[N]}\rangle$ in cosmology are discussed at the end of this paper.
\end{abstract}

\maketitle
\tableofcontents

\section{Introduction}\label{sec:intro}

 Loop Quantum Gravity (LQG) is an approach toward the background independent and nonperturbative quantum gravity theory in four and higher dimensions \cite{thiemann2007modern,rovelli2014covariant,han2007fundamental,ashtekar2004back,makinen2020dynamics}. Several recent progresses have been made by the active research of the quantum dynamics of LQG \cite{Giesel:2006uj,Giesel:2006uk,giesel2007algebraic,Domagala:2010bm,lewandowski2011dynamics,alesci2015hamiltonian,assanioussi2015new,yang2015new,alesci2018quantum,Dapor:2017rwv,han2020effective}. Particularly, tremendous progresses have been made in both canonical and covariant LQG on the semiclassical limit and the consistency with respect to classical gravity e.g. \cite{thiemann2001gauge,thiemann2001gaugeII,thiemann2001gaugeIII,Giesel:2006uk,giesel2007algebraic,Dapor:2017gdk,Dapor:2017rwv,Liegener:2019jhj,Han:2020chr,Conrady:2008mk,semiclassical,HZ,Bianchi:2006uf,Han:2018fmu}. However, regarding the full theory of LQG dynamics, less progress has been made on its quantum corrections (see e.g. \cite{Liegener:2020dcg,Assanioussi:2017tql,Alesci:2013kpa,Borissov:1997ji} for some results, and \cite{Han:2020fil,Dona:2019dkf} for some results in the covariant approach). As a candidate of quantum gravity theory, it is important that LQG should shed light on quantum corrections to the classical theory of gravity.  

The present paper focuses on the canonical aspects of LQG. Due to the non-polynomial Hamiltonian constraint operator $\widehat{H[N]}=\widehat{H_E[N]}+(1+\beta^2)\widehat{H_L[N]}$,
there has been persistent confusion that the quantum dynamics of LQG might not be computable analytically  \cite{Nicolai:2005mc}. A previous work \cite{giesel2007algebraic} partially resolves this confusion, where it has been schematically shown that the coherent state expectation value of the Hamiltonian/master constraint are computable order-by-order by the semiclassical expansion in $\hbar$. It is remarkable that the proposed scheme in \cite{giesel2007algebraic} can also be applied to a wide class of non-polynomial operators used in the study of LQG dynamics. Although this scheme was proposed as early as when \cite{giesel2007algebraic} firstly published in 2006, the expectation value of $\widehat{H[N]}$ has only been computed at its classical limit, i.e. the 0-th order (in $\hbar$). However, due to the complexity of the operator, especially the the Lorentzian part of $\widehat{H[N]}$ (denoted as $\widehat{H_L[N]}$), the $O(\hbar)$ quantum correction has not been studied in the literature. 

The goal of the present work is to fill this gap by providing an explicit computation of the $O(\hbar)$ quantum correction in $\langle\widehat{H[N]}\rangle$ with respect to a certain coherent state. In this paper, in order to compute the quantum correction in $\langle\widehat{H[N]}\rangle$, a powerful algorithm is developed to overcome the complexity of $\widehat{H[N]}$ that is the non-graph-changing Hamiltonian on a cubic lattice $\gamma$. We explicitly expand $\langle\widehat{H[N]}\rangle$ to linear order in $\hbar$ by applying the algorithm, with respect to the coherent state that is peaked at the homogeneous and isotropic data of cosmology. Namely we explicitly compute $H_0$ and $H_1$ in
\begin{equation}\label{hamexp0}
\langle\widehat{H[1]}\rangle=H_0+\ell_{\rm p}^2 H_1+O(\ell_{\rm p}^4),\quad \ell_{\rm p}^2=\hbar \kappa
\end{equation} 
where $\kappa=8\pi G_{\rm Newton}$ and the lapse function $N=1$. $H_0$, representing the 0-th order, reproduces the cosmological effective Hamiltonian in the $\mu_0$-scheme \cite{yang2009alternative,Dapor:2017rwv,han2020effective,Kaminski:2020wbg}. Whereas, $H_1$ gives the first order quantum correction, which is presented by our result in this work. The explicit expression of $H_1$ is given in Section \ref{sec:results}. It is worth noting that the coherent state used for computing $\langle\widehat{H[N]}\rangle$ is not SU(2) gauge invariant (the motivation is stated below). 

This work is closely related to the reduced phase space formulation of LQG (see e.g. \cite{giesel2010algebraicIV,Giesel:2012rb,Domagala:2010bm}. 
In this formulation, some matter fields that are known as the clock fields is coupled to gravity.
These matter fields are regarded as material reference frames used to transform gravity variables to gauge invariant Dirac observables.
This procedure resolves the Diffeomorphism constraint and Hamiltonian constraint at the classical level resulting in the reduced phase space $\mathscr{P}_{red}$ of Dirac observables. 
The dynamics of the gravity-clock system is described by the material-time evolution generated by the physical Hamiltonian ${\bf H}$ on the reduced phase space $\mathscr{P}_{red}$.
As an interesting model, Gaussian dust is chosen to be our clock fields  \cite{Kuchar:1990vy,Giesel:2012rb}. 
Then the resulting reduced phase space, $\mathscr{P}_{red}$, is identical to the pure-gravity unconstrained phase space. This identification defines the pure-gravity Hamiltonian constraint with unit lapse (i.e. $N=1$) $H[1]$ on the resulting reduced phase space $\mathscr{P}_{red}$, which indictaes that the physical Hamiltonian ${\bf H}$ equals to $H[1]$ for the case when gravity is coupled to Gaussian dust. In this model, the quantization of $\mathscr{P}_{red}$ is the same as quantizing the pure-gravity unconstrained phase space, which leads to the physical Hilbert space $\mathcal{H}$ that is identical to the kinematical Hilbert space in the usual LQG. $\mathcal{H}$ is unconstrained because it is from the quantization of $\mathscr{P}_{red}$. The physical Hamiltonian operator is obtained by $\widehat{\bf H}=\frac{1}{2}(\widehat{H[1]}+\widehat{H[1]}^\dagger)$ with the usual LQG quantization of $\widehat{H[N]}$ \cite{Thiemann:2020cuq,thiemann1998quantum,Giesel:2006uj,Lewandowski:2014hza}. Therefore from the perspective of reduced-phase-space LQG, our work computes the $\langle \widehat{\bf H} \rangle$ with respect to the coherent state peaked at cosmological data on the graph $\gamma$, which is given by the real part of $\langle\widehat{H[1]}\rangle$ \eqref{hamexp0}.

Recent works have been focused on building models of LQG on a single graph $\gamma$ \cite{Dapor:2017gdk,Dapor:2017rwv,zhang2019bouncing,Liegener:2019jhj,Han:2019feb,Dapor:2020jvc,Han:2020iwk,Liegener:2020dcg}. Particularly, the quantum dynamics in the reduced-phase-space LQG framework is formulated on a cubic lattice $\gamma$ with a path integral \cite{han2020effective,Han:2020chr}
\begin{equation}
A_{[g],\left[g^{\prime}\right]}=\int \mathrm{d} h  [\mathrm{d} g]\, \nu[g]\, e^{S[g, h] / \ell_{\rm p}^2},\label{Agg0}
\end{equation}
which is the canonical-LQG analog of the spinfoam formulation. $A_{[g],\left[g^{\prime}\right]}$ is regarded as the transition amplitude of $\widehat{\bf H}$ between the initial and final SU(2) gauge invariant coherent states, denoted as  $|[g]\rangle,|\left[g^{\prime}\right]\rangle$ respectively. The integration variables contains both trajectories of $g\in\mathscr{P}_{red,\gamma}$ and SU(2) gauge transformations $h$ on $\gamma$, with $\nu[g]$ being a measure factor. Due to a feature of the path integral \eqref{Agg0}, the SU(2) gauge invariant amplitude $A_{[g],\left[g^{\prime}\right]}$ is expressed as an integral of SU(2) gauge non-invariant variables $g$ and $h$ \footnote{Here is a brief explanation of the reason why the path integral is in terms of gauge non-invariant coherent states (see \cite{han2020effective} for details): The transition amplitude between gauge invariant coherent states is $A_{[g],[g']}=\left\langle\Psi_{[g]}^{t}|U(T)| \Psi_{\left[g^{\prime}\right]}^{t}\right\rangle$ where $U(T)=\exp \left(-\frac{i}{\hbar} T \widehat{\mathbf{H}}\right)$. The gauge invariant coherent state is the group average of the gauge non-invariant state: $|\Psi_{[g]}^{t}\rangle=\int \mathrm{d} h|\psi_{g^{h}}^{t}\rangle$ where $h$ is the SU(2) gauge transformation. Since $U(T)$ is gauge invariant, we have $A_{[g],\left[g^{\prime}\right]}=\int \mathrm{d} h\langle\psi_{g}^{t}|U(\tau)| \psi_{g^{\prime h}}^{t}\rangle$ where the integrand can be written as a coherent state path integral with the standard method. }. The action $S[g,h]$ is linear to $\langle \widehat{\bf H} \rangle$ at SU(2) gauge non-invariant coherent states when the trajectories of $g$ are continuous in time. In contrast to the usual path integrals in quantum field theories, $S[g,h]$ contains the $O(\hbar)$ correction from $\langle \widehat{\bf H} \rangle$. Our work precisely computes this $O(\hbar)$ correction in $S[g,h]$ of cosmological dynamics.  

A general study of the equation of motion provided by the semiclassical limit of $A_{[g],\left[g^{\prime}\right]}$ is presented in \cite{Han:2020chr}. The application of it in cosmology is presented in \cite{han2020effective,Han:2020iwk}. The cosmological dynamics in the limit of $\hbar\to0$ gives the $\mu_0$-scheme effective cosmological dynamics which reduces to the classical FRLW cosmology at low energy density. Next, it is equally important to discover the $O(\hbar)$ correction of the effective cosmological dynamics. The effective dynamics with $O(\hbar)$ correction can be obtained by the \emph{quantum effective action} \cite{peskin1995introduction}, denoted as $\Gamma$, from the path integral defined in \eqref{Agg0}. Perturbatively, the $O(\hbar)$ correction in $\Gamma$ for cosmology contains 3 contributions: (1) $O(\hbar)$ correction in $S[g,h]$ which is computed in this work, (2) $O(\hbar)$ correction in $\log \nu[g]$ where $\nu[g]$ has been given explicitly in \cite{han2020effective}, and (3) $O(\hbar)$ correction in $\frac{1}{2}\log\det(\mathfrak{H})$ where the ``1-loop determinant'' $\det(\mathfrak{H})$ is the determinant of the Hessian matrix $\mathfrak{H}$ of $S[g,h]$. The $g$-$g$ matrix elements in $\mathfrak{H}$ has been computed in \cite{Han:2020iwk}. A continuous study of $\frac{1}{2}\log\det(\mathfrak{H})$ is postponed for future work. Therefore,in terms of the quantum correction in the effective cosmological dynamics, the present work computes an significant part in the $O(\hbar)$ correction of the quantum effective action $\Gamma$.    

After introducing motivations above, let us summarize several key steps in the computation of the present work: First of all, an important complication in $\widehat{H[N]}$ is the volume operator $\hat{V}_v=\sqrt{|\hat{Q}_v|}$ which contains the square-root and absolute-value, indicating that the $\widehat{H[N]}$ is non-polynomial. When $\langle\widehat{H[N]}\rangle$ is studied with respect to the coherent state, this issue is overcome by substituting $\hat{V}_v$ with the semiclassical expansion \cite{giesel2007algebraic}
\begin{equation}
\begin{aligned}
\hat V_{GT}^{(v)}=&\langle \hat Q_v\rangle^{2q}\left[1+\sum_{n=1}^{2k+1}(-1)^{n+1}\frac{q(1-q)\cdots(n-1-q)}{n!}\left(\frac{\hat Q_v^2}{\langle \hat Q_v\rangle^2}-1\right)^n\right]+O(\hbar^{k+1})
\end{aligned}
\end{equation}
where $\hat{Q}_v$ is formulated as a polynomial of flux operators and $q=1/4$. Truncating $\hat V_{GT}^{(v)}$ with a finite $k$ and substituting it back into $\widehat{H[N]}$ allows us to express $\langle\widehat{H[N]}\rangle$ by a expectation value of a polynomial operator. 

The resulting polynomial sums over a huge number of terms ($\sim 10^{19}$), each of which is a monomial of holonomy and fluxe operators. Computing expectation values of all terms would lead to a large computational complexity. The major complexity is encoded in the Lorentzian part of $\widehat{H[N]}$, denoted as $\widehat{H_L[N]}$. Several key methods are used to reduce the number of computations:

\begin{itemize}

\item The expectation value of every monomial term can be factorized into expectation values of holonomy-flux monomials with respect to different edges. Only certain types of expectation values of monomials on a single edge shall be computed. We further reduce the number of types by using the commutation relations, and several general formulae are derived for the expectation values of the resulting types (see Section \ref{Expectation values of operators on one edge}).

\item We develop a power-counting argument in order to specifically locate each power of $\hbar$, expression in $O(\hbar)$ represents the leading order behavior of each expectation value of the monomial operator (see Section \ref{se:leading}). Since we are only interested in expanding $\langle\widehat{H[N]}\rangle$ to the its linear order in $\hbar$, a substantial amount of expectation values of monomials can be neglected due to the fact that they are only contributing to higher order in $\hbar$.

\item When the coherent states are peaked at homogeneous and isotropic data. A large amount of symmetries that identify different terms are realized, which can be used to reduce the computational complexity.(see Secton \ref{sec:cosmodel}).   

\end{itemize}

Our method exponentially reduces the computational complexity. In particular, it is useful in computing the expectation value of Lorentzian part in $\widehat{H[N]}$. 

{In Section \ref{Introducing the algorithm}, 
In order to present the reduction methodology more concretely,
an example that contains $3^{3m-1}$ ($m$ can be large) monomials is demonstrated. By applying our method, only 5 monomials' expectation values need to be computed.}

{The purpose of the present paper is to give detailed derivations for the results presented in \cite{appeared}.} Computations in this paper are carried out by using Mathematica on the High Performance Computation server with two 48-Core Processors (AMD EPYC 7642). One can find the Mathematica codes at \cite{github}.

The explicit resulting expression of $O(\hbar)$ quantum correction in $\langle\widehat{H[N]}\rangle$ is summarized in Section \ref{sec:results}.
In order to demonstrate the physical significance of our results and effects from the $O(\hbar)$ correction to the classical limit of $\Re\langle\widehat{H[1]}\rangle$, the proposal in \cite{Dapor:2017rwv} is adopted: We view $\Re\langle\widehat{H[1]}\rangle$ in \eqref{hamexp0} as the effective Hamiltonian on the 2-dimensional phase space, denoted as $\mathscr{P}_{cos}$, of homogeneous and isotropic cosmology. $\Re\langle\widehat{H[1]}\rangle$ generates the Hamiltonian time evolution on the 2-dimensional phase space $\mathscr{P}_{cos}$. Time evolution of the homogeneous spatial volume is plotted, and is compared with the evolution generated by $\langle\widehat{H[1]}\rangle$ at the limit of $\hbar\to0$. The comparison demonstrates the effects on $\langle\widehat{H[1]}\rangle$, which is generated from the $O(\hbar)$ correction contribution  (see Section \ref{sec:results} for details). We emphasize that the proposal that we adopt for the cosmological evolution is not as rigorous as the path integral formula \eqref{Agg0}. Nevertheless, we have argued that the $O(\hbar)$ correction in $\langle\widehat{H[1]}\rangle$ only contributes partially to the quantum correction in $\Gamma$ which ultimately determines the quantum effect in the dynamics. The cosmological dynamics studied in Section \ref{sec:results} only aims for displaying the effect of the $O(\hbar)$ correction in $\langle\widehat{H[1]}\rangle$, and is not a rigorous prediction from the principle of LQG. \\

The structure of the present paper the followings. Section \ref{sec:one} reviews the theory of LQG on a cubic lattice, including the Hamiltonian and the coherent state. Section \ref{Expectation values of operators on one edge}, we demonstrate the computations of the expectation value of operators defined at a single edge.
Section \ref{se:leading}, we develop a power-counting argument in order to reduce the computational complexity. Section \ref{sec:cosmodel} discusses $\langle\widehat{H[1]}\rangle$ with respect to the coherent states peaked at homogeneous and isotropic data, and the symmetries which reduce the computational complexity. Section \ref{sec:results} presents the explicit results of the quantum correction in $\langle\widehat{H[1]}\rangle$. Section \ref{sec:conclusion}, we conclude and discuss a few outlooks of the present work.

\section{Preliminaries}\label{sec:one}

\subsection{Quantization and Hamiltonian}

Classically general relativity can be formulated with the Ashtekar-Barbero variables $(A_a^i,E^a_i)$ consisting of SU$(2)$ connection $A^i_a$ and canonically conjugate densitized triad field $E_i^a$ defined on the spatial manifold $\Sigma$ \cite{barbero1995real}. 
 We denote the coordinate on $\Sigma$ by $(x,y,z)$ .  Let $\gamma\subset \Sigma$ be a finite cubic lattice whose edges are parallel to the axes of the coordinates. The sets of edges and vertices in $\gamma$ are denoted by $E(\gamma)$ and $V(\gamma)$ respectively. Taking advantage of $\gamma$, we define holonomies along the edges of $\gamma$, 
\begin{equation}\label{eq:holonomy}
h_e(A)=\mathcal P\exp\int_e A=1+\sum_{n=1}^\infty \int_0^1\dd t_n\int_0^{t_n}\dd t_{n-1}\cdots\int_0^{t_2}\dd t_1 A(t_1)\cdots A(t_n),\ \forall e\in E(\gamma),
\end{equation}
and gauge covariant fluxes \cite{Thiemann:2000bv} on the 2-faces $S_e$ in the dual lattices $\gamma^*$,
\begin{equation}\label{eq:flux}
\begin{aligned}
p^i_s(e):=&-\frac{2}{\beta a^2}\tr\left[\tau^i\int_{S_e} \varepsilon_{abc} h(\rho^s_e(\sigma) ) E^c(\sigma) h(\rho^s_e(\sigma)^{-1})\right],
\end{aligned}
\end{equation}
where $S_e\in\gamma^*$ is the 2-face, $\rho^s(\sigma): [0,1]\to \Sigma$ is a path connecting  the source point $s_e\in e$ to $\sigma\in S_e$ such that $\rho_e^s(\sigma):[0,1/2]\to e$ and $\rho_e^s(\sigma):[1/2,1]\to S_e$. $a$ is a length unit (e.g. $a=1$mm) to make $p_s(e)$ dimensionless. Alternatively, one can choose the target point $t_e\in e$ rather than $s_e$ to define
\begin{equation}
\begin{aligned}
p^i_t(e):=&\frac{2}{\beta a^2}\tr\left[\tau^i\int_{S_e} \varepsilon_{abc} h(\rho^t_e(\sigma) ) E^c(\sigma) h(\rho^t_e(\sigma)^{-1})\right].
\end{aligned}
\end{equation}
where $\rho^t(\sigma): [0,1]\to \Sigma$ is a path connecting  the target $t_e\in e$ to $\sigma\in S_e$ such that $\rho_e^t(\sigma):[0,1/2]\to e$ and $\rho_e^t(\sigma):[1/2,1]\to S_e$. Given $(A_a^i,E_i^a)$, Eqs. \eqref{eq:holonomy} and \eqref{eq:flux} lead to a map
from the $E(\gamma)$ to SL($2,\mathbb C$),
\begin{equation}\label{eq:defineg}
\g: e\mapsto g_e=e^{i p^k_s(e)\tau_k}h_e.
\end{equation}
Because of the relation between $p_s$ and $p_t$ 
\begin{equation}
p_s^k(e^{-1})\tau_k=p_t^k(e)\tau_k=-h_e^{-1} p_s^k(e)\tau_k h_e
\end{equation}
we obtain that 
\begin{equation}
\begin{aligned}
g_{e^{-1}}=g_e^{-1}. 
\end{aligned}
\end{equation}
Thus the map $\g:E(\gamma)\to$SL($2,\mathbb C$) generate a homomorphism from the groupoid of the graph $\gamma$ to SL$(2,\mathbb C)$. The LQG phase space based on $\gamma$ is SL$(2,\mathbb C)^{|E(\gamma)|}$ and consists of all such homomorphisms \cite{Thiemann:2000bv}. Given a SU(2)-valued scalar field $\mathcal G:\Sigma\to$SU(2) on $\Sigma$, $\mathcal G$ defines a gauge transformation on $\g$, taking $\g$ to $\mathcal G\triangleright\g$ with
\begin{equation}
    (\mathcal G \triangleright\g)(e)=\mathcal G(s_e)\g(e)\mathcal G(t_e)^{-1},\ \forall e\in E(\gamma).
\end{equation}

The quantization of this classical lattice theory gives us LQG based on the graph $\gamma$. The Hilbert space $\mathcal H_\gamma$ consists of the square integrable functions of the holonomies. Given two functions $\psi_i: \{h_e\}_{e\in E(\gamma)}\to \mathbb C$, the inner produce is
\begin{equation}
\langle \psi_1|\psi_2\rangle=\int_{\mathrm{SU}(2)^{|E(\gamma)|}} \dd\muh\,\overline{\psi_1(\{h_e\}_{e\in E(\gamma)})}\psi_2(\{h_e\}_{e\in E(\gamma)})
\end{equation}
where $|E(\gamma)|$ denote the number of elements (i.e. cardinality) of $E(\gamma)$ and $\muh$ is the Haar measure. $\mathcal H_\gamma$ is the kinematical Hilbert space of the canonical LQG with the operator-constraint formalism. Moreover, $\mathcal H_\gamma$ modulo gauge transformations  represents the physical Hilbert space of the   reduced-phase-space LQG , where any gauge invariant function of $h_e(A)$ and $p_{s,t}^i(e)$ are Dirac observables, realized from the deparametrization by coupling to clock fields \cite{giesel2010algebraicIV}.

On $\mathcal H_\gamma$, $p_s^i(e)$ and $p_t^i(e)$ are quantized as the right- and left-invariant vector field, namely
\begin{equation}\label{eq:ptps}
\begin{aligned}
(\hat p_s^i(e)\psi)(h_{e'},\cdots,h_e,\cdots,h_{e''})&=i t\left.\frac{\dd}{\dd\epsilon}\right|_{\epsilon=0}\psi(h_{e'},\cdots,e^{\epsilon\tau^i}h_e,\cdots,h_{e''})\\
(\hat p_t^i(e)\psi)(h_{e'},\cdots,h_e,\cdots,h_{e''})&=-i t\left.\frac{\dd}{\dd\epsilon}\right|_{\epsilon=0}\psi(h_{e'},\cdots,h_ee^{\epsilon\tau^i},\cdots,h_{e''})
\end{aligned}
\end{equation}
where $t=\kappa\hbar/a^2=:\Pl^2/a^2$ (if $a=1$mm, $t\simeq6.56\times 10^{-63}$) and $\tau^j=(-i/2)\sigma^j$ with $\sigma^j$ being the Pauli matrix. Another kind of basic operators are the multiplication operators $D^\iota_{ab}(h_e)$ which is defined as 
\begin{equation*}
(D^\iota_{ab}(h_e)\psi)(A)=D^\iota_{ab}(h_e(A))\psi(A)
\end{equation*}
where $D^\iota(h_e(A))$ is the value of the Wigner-D matrix at $h_e(A)\in $SU(2). In this paper, $D^\iota(x)$ denotes the Wigner-D matrix only if $x$ is some specific SU(2) element. Moreover, $h_e$, when it appears alone as an operator, denote the matrix-valued multiplication operator $D^{1/2}(h_e)$ for simplicity. With this convention, the commutators between the basic operators read
\begin{equation}\label{eq:commutators}
\begin{aligned}
[D^\iota(h_{e}),D^\iota(h_{e'})]&=0=[\hat p_s^i(e),p_t^j(e')]\\
[\hat p_s^i(e),\hat p_s^j(e')]&=-it\delta_{ee'}\epsilon_{ijk}\hat p_s^k(e),\\
[\hat p_t^i(e),\hat p_t^j(e')]&=-it\delta_{ee'}\epsilon_{ijk}\hat p_t^k(e),\\
[\hat p_s^i(e),D^\iota(h_{e'})]&=it\delta_{ee'}D'{}^\iota(\tau^i)D^\iota(h_e),\\
[\hat p_t^i(e),D^\iota(h_{e'})]&=-it\delta_{ee'}D^\iota(h_e)D'{}^\iota(\tau^i).
\end{aligned}
\end{equation}
where $D'{}^\iota(\tau^i)$ is the corresponding representation matrix of $\tau^i$.

It is useful to introduce the flux operators with respect to the spherical basis. We define
\begin{equation}\label{eq:sphericalbasis}
\hat p_v^{\pm 1}(e):=\mp\frac{1}{\sqrt{2}}\left(\hat p_v^x(e)\pm i \hat p_v^y(e)\right),\ \hat p_v^0(e)=\hat p_v^z(e)
\end{equation}
with $v=s,t$. In the following context, $\alpha,\beta,\cdots=0,\pm 1$ is used to denote the indices in the spherical basis, and $i,j,k\cdots=1,2,3$, the indices in the Cartesian basis.

Taking advantage of the basic operators, one can define operators representing geometric observables such as areas and volumes \cite{rovelli1995discreteness,ashtekar1997quantum,ashtekar1997quantumII}. The volume operator plays an important role in the present work. Let $\mathcal R\subset \Sigma$ be a region in $\Sigma$. The volume operator of $\mathcal R$ is defined by 
\begin{equation}
\hat V_{\mathcal R}:=\sum_{v\in V(\gamma)\cap \mathcal R}\hat V_v=\sum_{v\in V(\gamma) \mathcal \cap R}\sqrt{|\hat Q_v|}
\end{equation}
where
\begin{equation}\label{eq:Qinxyz}
\hat Q_v=(\beta a^2)^3\varepsilon_{ijk} \frac{\hat p_s^i(e_x^+)-\hat p_t^i(e_x^-)}{2}\frac{\hat p_s^j(e_y^+)-\hat p_s^j(e_y^-)}{2}\frac{\hat p_s^k(e_z^+)-\hat p_s^k(e_z^-)}{2}
\end{equation}
where $e_i^\pm$ with $i=x,\ y,\ z$ are the edges along the $i$th axis such that $v$ is the source point of $e_i^+$ and the target point of $e^-_i$. The total volume is denoted by $\hat V=\sum_{v\in V(\gamma)}\hat V_v$. In terms of the flux operators with respect to the spherical basis \eqref{eq:sphericalbasis}, the operator $\hat Q_v$ defined in Eq. \eqref{eq:Qinxyz} becomes
\begin{equation}\label{eq:qdefinition}
\begin{aligned}
&\hat Q_v=-i(\beta a^2)^3\varepsilon_{\alpha\beta\gamma} \frac{\hat p_s^\alpha(e_x^+)-\hat p_t^\alpha(e_x^-)}{2}\frac{\hat p_s^\beta(e_y^+)-\hat p_s^\beta(e_y^-)}{2}\frac{\hat p_s^\gamma(e_z^+)-\hat p_s^\gamma(e_z^-)}{2}
\end{aligned}
\end{equation}
where $\varepsilon_{\alpha\beta\gamma}$ is defined by $\varepsilon_{-1,0,1}=1$.

In the operator-constraint formalism, the dynamics of LQG is encoded in the Hamiltonian constraint, which can be written as 
\begin{equation}
\widehat{H[N]}=\widehat{H_E[N]}+(1+\beta^2)\widehat{H_L[N]}
\end{equation}
where $\widehat{H_E[N]}$ is called the Euclidean part and $\widehat{H_L[N]}$ is the Lorentzian part. $N$ is the smeared function. $\widehat{H[N]}$ is constructed by using the Thiemann's trick  \cite{thiemann1998quantum,Giesel:2006uj}. The operator corresponding to the Euclidean part is
\begin{equation}
\begin{aligned}
&\widehat{H_E[N]}=\frac{1}{i\beta a^2 t}\sum_{v\in V(\gamma)} N(v)\sum_{e_I,e_J,e_K \text{ at } v}\epsilon^{IJK}\tr(h_{\alpha_{IJ}}h_{e_K}[\hat V_v,h_{e_K}^{-1}])
\end{aligned}
\end{equation}
where  $e_I$, $e_J$ and $e_K$ are oriented to be outgoing from $v$, $\epsilon^{IJK}=\sgn[\det (e_I\wedge e_J\wedge e_K)]$, $\alpha_{IJ}$ is the minimal loop around a plaquette consisting of $e_I$ and $e_J$, where it goes out via $e_I$ and comes back through $e_J$, taking $v$ as its end point.
With the same notion, the Lorentzian part reads
\begin{equation}
\begin{aligned}
\widehat{H_L[N]}=\frac{-1}{2i\beta^7 a^{10} t^5}\sum_v N(v)\sum_{e_I,e_J,e_K \text{ at } v}\varepsilon^{IJK}\tr( [h_{e_I},[\hat V,\hat H_E]]h_{e_I}^{-1} [h_{e_J},[\hat V,\hat H_E]]h_{e_J}^{-1}[h_{e_K},\hat V_v]h_{e_K}^{-1}).
\end{aligned}
\end{equation}

In the reduced-phase-space LQG where the diffeomorphism and Hamiltonian constraints are solved classically, the quantum dynamics is governed by the physical Hamiltonian $\widehat{\bf H}$. When we deparametrize gravity by coupling to the Gaussian dust \cite{Giesel:2012rb,Kuchar:1990vy}, the classical physical Hamiltonian ${\bf H}$ is formally the same as the Hamiltonian constraint with unit lapse, except all quantities in ${\bf H}$ are understood as Dirac observables. The quantization gives the Hamiltonian operator
\begin{eqnarray}
\widehat{\bf H}=\frac{1}{2}\left(\widehat{H[1]}+\widehat{H[1]}^\dagger\right)
\end{eqnarray}
${\bf \widehat{H}}$ is defined on $\mathcal{H}_\gamma$, which can be understood from a similar perspective of quantizing Dirac observables. Note that here we consider the non-graph-changing version of the Hamiltonian (constraint). If ${H[N]}$ is understood as constraint, the discretization and quantization on $\gamma$ cause the constraint anomaly. However in the reduced-phase-space LQG, the constraint anomaly is absent, because ${H[N]}$ is not a  constraint anymore.\footnote{Sometimes, ${H[N]}$ relates to conserved charges, then $\widehat{H[N]}$ on $\gamma$ may break the classical symmetry.}. The self-adjoint extension of $\hat{\mathbf{H}}$ exists \cite{Thiemann:2020cuq,private}, so we choose the extension and define the self-adjoint Hamiltonian which is still denoted by $\hat{\mathbf{H}}$.

\subsection{Coherent states}\label{sec:two}

Choosing a canonical orientation for each edge $e\in E(\gamma)$, the classical phase space based on the graph $\gamma$ is 
\begin{equation}
\Gamma_\gamma\cong [{\rm SL}(2,\mathbb C)]^{ |E(\gamma)|}.
\end{equation}
The complexifier coherent state $\Psi_\g$ is \cite{thiemann2001gauge} 
\begin{equation}\label{eq:coherentstate}
\Psi_\g=\bigotimes_{e\in E(\gamma)}\psi^{ t}_{g_e},\qquad \psi^t_{g_e}(h_e)=\sum_j d_je^{-\frac{t}{2}j(j+1)}\chi_j(g_e h_e^{-1})
\end{equation}
where $\psi^{ t}_{g_e}$ is the SU(2) coherent state at the edge $e$. The character $\chi_j(g_eh_e^{-1})$ is the trace of the $j$-representation of $g_eh_e^{-1}$. The property $\chi_j(g_e h_e^{-1})=\chi_j(g_e^{-1}h_e)$ leads to the useful relation $$\psi_{g_e}(h_e)=\psi_{g_{e^{-1}}}(h_{e^{-1}}).$$

Given $g\in$SL($2,\mathbb C$), it can be decomposed as
\begin{equation}\label{eq:decomposition}
g=e^{ip_k\tau^k}u=n^s e^{i(\eta+i\xi)\tau_3}(n^t)^{-1},
\end{equation}
where $\eta=-\sqrt{\vec p\cdot\vec p}$ and $n^s,n^t\in$SU(2), as well as $\xi\in\mathbb R$, are given by
\begin{equation}
\begin{aligned}
n^s \tau^3 (n^s)^{-1}&=-\frac{p_k}{\sqrt{\vec p\cdot \vec p}}\, \tau_k,\\
n^se^{-\xi\tau_3} (n^t)^{-1}&=u. 
\end{aligned}
\end{equation}
Although $n^s$ and $n^t$ are not uniquely defined by this equation, each of them relates to a unique vector through the equation, with $v=s,t$,
\begin{equation}
n^v\tau_3 (n^v)^{-1}=\vec n^v\cdot \vec\tau.
\end{equation}
It is shown in \cite{thiemann2001gaugeIII} and is revisited shortly that
\begin{equation}
\frac{\langle \psi_{g_e}^t|\vec{\hat p}_s(e)|\psi^t_{g_e}\rangle}{\|\psi^t_{g_e}\|^2}=-\eta_e\vec n_e^s+O(t),\ \frac{\langle \psi_{g_e}^t|\vec{\hat p}_t(e)|\psi_{g_e}^t\rangle}{\|\psi^t_{g_e}\|^2}= \eta_e\vec n_e^t+O(t)
\end{equation}
{where $\eta_e$,  $\vec n_e^s$ and $\vec n_e^t$ are defined as the decomposition parameters of $g_e$ as in Eq. \eqref{eq:decomposition}}. 
This equation indicates that $\eta_e\vec n_e^t$ is the classical limit of the flux operator at $e$.
 
The following properties of the $\psi_g^t$  \cite{thiemann2001gauge,thiemann2001gaugeII,thiemann2001gaugeIII} are useful in our analysis. Firstly, the inner product of these states read
\begin{equation}\label{eq:innerproduct}
\langle\psi_{g_1}^t|\psi_{g_2}^t\rangle=\psi^{2t}_{g_1^\dagger g_2}(1)= \frac{2\sqrt{\pi } e^{t/4}}{t^{3/2}}\, \frac{\zeta\,e^{\frac{\zeta^2}{t}}}{\sinh(\zeta)}+O(t^\infty)
\end{equation}
where $\tr(g_1^\dagger g_2)=2\cosh(\zeta)$ and $\Im(\zeta)\in[0,\pi]$ with $\Im(\zeta)$ the imaginary part of $\zeta$ \footnote{Here we used the following result shown in \cite{thiemann2001gaugeIII}. For any complex number $z=R+iI$, there exist real numbers $s\in\mathbb R$ and $\phi\in[0,\pi]$  such that  $\cosh(s+i\phi)=z$. $s$ and $\phi$ are uniquely determined except in the case $I=0$ and $|R|>1$ in which case the $s$ is determined up to its sign.}. Consequently, the norm of the coherent state is
\begin{equation}\label{eq:norm}
\begin{aligned}
\langle 1\rangle_{g}:=\langle\psi^t_g|\psi^t_g\rangle= \frac{2\sqrt{\pi } e^{ t/4}}{t^{3/2}}\frac{ p e^{\frac{ p ^2}{t}}}{\sinh(p)}+O(t^\infty),
\end{aligned}
\end{equation}
where $p=\sqrt{\vec{p}\cdot\vec{p}}$. Secondly, $\psi_g^t$ satisfy the completeness condition
\begin{equation}\label{eq:completerelation}
\int \dd\nu_t(g)|\psi^t_g\rangle\langle\psi^t_g|=\mathbb I,
\end{equation}
where the measure $\dd\nu_t(g)$ is  
\begin{equation}\label{eq:measure}
\dd\nu_t(g)=\frac{2\sqrt{2} e^{-t/4}}{(2\pi t)^{3/2}}\frac{\sinh(p)}{p}e^{-\frac{p^2}{t}}\dd\mu_H(u)\dd ^3p=\frac{2}{\langle 1\rangle_g\pi t^3}\dd\mu_H(u)\dd ^3p.
\end{equation}

Let us complete this section with some discussions on the volume operator contained in the Hamiltonian operator $\widehat{H[N]}$.  
Because of the square root in the definition of the volume operator, matrix elements of these operators are difficult to compute analytically. However, as far as the coherent state expectation value is concerned, the volume operators $\hat{V}_v$ in $\widehat{H[N]}$ can be replaced by Giesel-Thiemann's volume \cite{giesel2007algebraic} $\hat V_{GT}^{(v)}$ which is a semiclassical expansion
\begin{equation}\label{eq:VGT}
\begin{aligned}
\hat V_{GT}^{(v)}=&\langle \hat Q_v\rangle^{2q}\left[1+\sum_{n=1}^{2k+1}(-1)^{n+1}\frac{q(1-q)\cdots(n-1-q)}{n!}\left(\frac{\hat Q_v^2}{\langle \hat Q_v\rangle^2}-1\right)^n\right]+O(t^{k+1})
\end{aligned}
\end{equation}
where $q=1/4$. By making use of $\hat V_{GT}^{(v)}$, firstly truncating $\hat V_{GT}^{(v)}$ at finite $k$ and replacing $\hat{V}_v$ by $\hat V_{GT}^{(v)}$, $\widehat{H[N]}$ can be expressed by a polynomial of holonomies and fluxes.
Up to higher order in $t$, it is now manageable to compute the expectation value of $\widehat{H[N]}$, through computing the expectation value of a polynomial of holonomies and fluxes.

\section{Expectation values of operators on one edge}\label{Expectation values of operators on one edge}
As becoming clear in a moment, computing the coherent state expectation value of $\widehat{H[N]}$ can be reduced to computing expectation values of operator monomials on individual edges. In this section, let us firstly focus on the expectation value of operators on one edge.

Given a monomial of holonomies and fluxes on an edge $e$, its expectation value with respect to the coherent state $\psi^t_{g_e}$ labelled by $g_e=n_e^s e^{i z_e\tau_3}(n_e^t)^{-1}$ relates to its expectation value with respect to $\psi^t_{e^{iz_e\tau_3}}\equiv \psi^t_{z_e}$, by a gauge transformation generated by $n_e^s$ and $n_e^t$ \cite{Dapor:2017gdk}:
\begin{equation}\label{eq:basicformula}
\begin{aligned}
&\langle \psi^t_{g_e}|P(\{\hat p_s^{\alpha_i}(e)\},\{\hat p_t^{\alpha_j}(e)\},\{D^{\iota_k}_{a_k b_k}(h_e)\})|\psi^t_{g_e}\rangle\\
=&\langle \psi^t_{z_e}|P(\{\hat p_s^{\beta_i}(e)D^1_{\beta_i\alpha_i}((n_e^s)^{-1})\},\{\hat p_t^{\beta_j}(e)D^1_{\beta_j\alpha_j}((n_e^t)^{-1})\},\{D^{\iota_k}_{a_k c_k}(n_e^s)D^{\iota_k}_{c_k d_k}(h_e)D^{\iota_k}_{d_k b_k}((n_e^t)^{-1})\})|\psi^t_{z_e}\rangle,
\end{aligned}
\end{equation}
where $P(x,y,z)$ represents any monomial of $x=\{x_1,x_2,\cdots,x_m\}$, $y=\{y_1,\cdots,y_n\}$ and $z=\{z_1,z_2,\cdots,z_k\}$. This feature implies that one can always do the calculation with respect to $\psi^t_{z_e}$, then restore the information of $n_e^s$ and $n_e^t$ afterwards. In the following context, we denote 
\begin{equation}\label{eq:tocompute}
\langle \psi^t_{z_e}|\hat F_e|\psi^t_{z_e}\rangle=:\langle \hat F_e\rangle_{z_e}. 
\end{equation}
Now let us consider the algorithm to compute \eqref{eq:tocompute} for a general monomial $\hat F_e$ of holonomies and fluxes  step by step. Based on the algorithm described in the following subsections, our  codes \cite{github} are designed. By the codes, one can compute  the expectation value of  arbitrary monomial $\hat F_e$ up to  arbitrary order.
\subsection{The algorithm }\label{sec:basicformula}

\subsubsection{The first step} \label{sec:commutator}
Given a monomial of holonomy and flux operators. To compute the expectation value of this monomial, we need at first to remove all the holonomies to the right with the basic commutation relations  \eqref{eq:commutators}. We use the following proposition to implement the procedure.  
\begin{pro}\label{pro:changeO}
Let $\hat O$ be defined as
\begin{equation}
\hat O=\hat O_1\hat O_2\cdots\hat O_m.
\end{equation}
Denote $\mathcal I:=\{2,\cdots,m\}$. Let $\mathcal I_k:=\{i_1,i_2,\cdots,i_k\}$ with $i_1<i_2<\cdots<i_k$ be a sublist of $\mathcal I$ which contains $k$ elements. 
Then, by the definition of commutator, it has
\begin{equation}\label{eq:changeO}
\begin{aligned}
\hat O=\hat O_2\cdots\hat O_m\hat O_1+\sum_{k=1}^{m-1}\sum_{\mathcal I_k} \left(\prod_{l \in\mathcal I-\mathcal I_k}\hat O_l\right)[[\cdots[[\hat O_1,\hat O_{i_1}],\hat O_{i_2}]\cdots],\hat O_{i_k}].
\end{aligned}
\end{equation}
\end{pro}
The proof is quite straightforward with using the relation $\hat A\hat B=\hat B\hat A+[\hat A,\hat B]$ iteratively.  In Eq. \eqref{eq:changeO}, the terms at $k$ carry $k$-fold commutator. Due to the factor $t$ in the right hand side of the commutation relation \eqref{eq:commutators},
the $k$-fold commutator produces a factor $t^k$ in the final results, which implies that the contributions of these terms to the expectation value of $\hat O$ are at least at $t^k$-order.
 
Now let us see how to use Proposition \ref{pro:changeO} to move the holonomies to the right precisely. Assume that $\hat O_1=D^\iota_{ab}(h_e)$ in Eq. \eqref{eq:changeO} and that all of the other operators are fluxes. Then according to Eq. \eqref{eq:changeO}, we need to calculate $$[\cdots,[[\cdots [[D^\iota_{ab}(h_e),\hat p_s^{\alpha_1}(e)],\hat p_s^{\alpha_2}(e)]\cdots \hat p_s^{\alpha_m}(e)],\hat p_t^{\alpha_1}(e)]\cdots,\hat p_t^{\alpha_n}(e)].$$
The result can be derived by
\begin{equation}\label{eq:mfoldcom1}
\begin{aligned}
&[\cdots [[D^\iota_{ab}(h_e),\hat p_s^{\alpha_1}(e)],\hat p_s^{\alpha_2}(e)]\cdots\hat p_s^{\alpha_m}(e)]=(-it)^m D^{'\iota}_{aa_1}(\tau^{\alpha_1})D^{\iota'}_{a_1a_2}(\tau^{\alpha_2})\cdots D^{'\iota}_{a_{m-1}a_m}(\tau^{\alpha_m}) D^\iota_{a_mb}(h_e)
\end{aligned}
\end{equation}
with
$a_k=a-\sum_{i=1}^k\alpha_i$,
and
\begin{equation}\label{eq:mfoldcom2}
\begin{aligned}
&[\cdots [[D^\iota_{ab}(h_e),\hat p_t^{\alpha_1}(e)],\hat p_t^{\alpha_2}(e)]\cdots \hat p_t^{\alpha_m}(e)]=(it)^m D^\iota_{ab_m}(h_e) D^{'\iota}_{b_mb_{m-1}}(\tau^{\alpha_m})D^{'\iota}_{b_{m-1}b_{m-2}}(\tau^{\alpha_{m-1}})\cdots D^{'\iota}_{b_1b}(\tau^{\alpha_1})
\end{aligned}
\end{equation}
with $b_k=\sum_{i=1}^k b_i+b$, where we used that
$D^\iota_{ab}(\tau^\alpha)\propto \delta_{a,b+\alpha}$. Taking advantage of Eqs. \eqref{eq:mfoldcom1} and \eqref{eq:mfoldcom2}, we have, for instance,
\begin{equation}\label{eq:resulto1}
\begin{aligned}
&D^\iota_{ab}(h_e)\Big(\prod_{i=1}^m\hat p_s^{\alpha_i}(e)\Big)\Big(\prod_{j=1}^n\hat p_t^{\beta_j}(e)\Big)\\
=&\Big(\prod_{i=1}^m\hat p_s^{\alpha_i}(e)\Big)\Big(\prod_{j=1}^n\hat p_t^{\beta_j}(e)\Big)D^\iota_{ab}(h_e)-it\sum_{k=1}^m \Big( \prod_{i\neq k}\hat p_s^{\alpha_i}(e)\Big)\Big(\prod_{j=1}^n\hat p_t^{\beta_j}(e)\Big)D^{'\iota}_{ac}(\tau^{\alpha_k})D^\iota_{cb}(h_e)\\
&+(-it)^2\sum_{k<l} \Big( \prod_{i\notin\{k,l\}}\hat p_s^{\alpha_i}(e)\Big)\Big(\prod_{j=1}^n\hat p_t^{\beta_j}(e)\Big)D^{'\iota}_{ac}(\tau^{\alpha_k})D^{'\iota}_{cd}(\tau^{\alpha_l})D^\iota_{db}(h_e)\\
&+it\sum_{k=1}^n\Big(\prod_{i=1}^m\hat p_s^{\alpha_i}(e)\Big) \Big( \prod_{j\neq k}\hat p_t^{\beta_j}(e)\Big)D^\iota_{ac}(h_e)D^{'\iota}_{cb}(\tau^{\beta_k})\\
&+(it)^2\sum_{k<l}\Big(\prod_{i=1}^m\hat p_s^{\alpha_i}(e)\Big) \Big( \prod_{j\notin\{k,l\}}\hat p_t^{\alpha_i}(e)\Big)D^\iota_{ac}(h_e)D^{'\iota}_{cd}(\tau^{\beta_l})D^{'\iota}_{db}(\tau^{\beta_k})\\
&-(it)^2\sum_{k,l}\Big( \prod_{i\neq k}\hat p_s^{\alpha_i}(e)\Big)\Big( \prod_{j\neq l}\hat p_t^{\beta_j}(e)\Big)D^{'\iota}_{ac}(\tau^{\alpha_k})D^\iota_{cd}(h_e)D^{'\iota}_{db}(\tau^{\beta_l})+O(t^3).
\end{aligned}
\end{equation}
If there are more than one holonomies contained in $\hat O$, one can use this procedure to permute them one by one. Finally, $\hat O$ is expressed as summation of terms taking the form
\begin{equation}\label{eq:generalform1}
\prod_{k=1}^m \hat p_s^{\alpha_i}(e)\prod_{k=1}^n \hat p_t^{\alpha_i}(e)\prod_{i=1}^lD^{\iota_i}_{a_ib_i}(h_e).
\end{equation} 
Then, one can merge the holonomies by applying the formula \eqref{eq:coupleholonomy}, we eventually simplify $\hat O$ to be a sum of operators of the form 
\begin{equation}\label{eq:finalformstep1}
\Big(\prod_{i=1}^m\hat p_s^{\alpha_i}(e)\Big)\Big(\prod_{j=1}^n\hat p_t^{\beta_j}(e)\Big)D^\iota_{ab}(h_e)
\end{equation}

\subsubsection{The second step}
The second step is to transform $\hat p_t^\beta(e)$ in Eq. \eqref{eq:finalformstep1} to $\hat p_s^\beta(e)$. To do this,  we employ the formula 
\begin{equation}\label{eq:pspth}
\begin{aligned}
&\langle \hat p_s^{\alpha_1}(e)\cdots \hat p_s^{\alpha_m}(e)\hat p_t^{\beta_1}(e)\cdots \hat p_t^{\beta_n}(e)D^{\iota}_{ab}(h_e)\rangle_{z_e}\\
=&(-1)^n e^{-(\beta_1+\cdots+\beta_n)\overline{z_e}} \langle \hat p_s^{\beta_n}(e)\cdots \hat p_s^{\beta_1}(e) \hat p_s^{\alpha_1}(e)\cdots \hat p_s^{\alpha_m}(e)D^{\iota}_{ab}(h_e)\rangle_{z_e}.
\end{aligned}
\end{equation}
The proof of  this formula is quit technical and is put in Appendix \ref{app:graph}. Because of this equation, we now only need to consider the expectation value of operators
\begin{equation}
\hat F^{\alpha_1\cdots \alpha_m}_{\iota a b}=\hat p_s^{\alpha_1}(e)\cdots \hat p_s^{\alpha_m}(e)D^{\iota}_{ab}(h_e).
\end{equation}

\subsubsection{The third step}
To compute  the expectation value of $\hat F^{\alpha_1\cdots \alpha_m}_{\iota a b}$, we need to consider the cases with $\iota=0$ and $\iota\neq  0$  separately. 
As shown in Appendix \ref{app:shabiref}, the expectation value for $ \hat F^{\alpha_1\cdots \alpha_m}_{000}\equiv  \hat F^{\alpha_1\cdots \alpha_m}$ reads 
\begin{equation}\label{eq:qmzw}
\begin{aligned}
&\langle\hat F^{\alpha_1\cdots \alpha_m}\rangle_{z_e}\\
=&\delta\left(\sum_{i}\alpha_i,0\right)t^m\prod_{i=1}^m\frac{1}{(1+|\alpha_i|)^{1/2}}e^{t/4}\int_{-\infty}^\infty\dd x x\prod_{k=1}^m\left(\frac{\alpha_k-1}{2}x-\frac{\partial_y}{2}+\sum_{i=1}^k\alpha_i-\frac{\alpha_k}{2}\right)\frac{e^{-\frac{t}{4}x^2+x\eta}}{2\sinh(y)}\Big|_{y\to\eta}+O(t^\infty).
\end{aligned}
\end{equation}
For the operator $\hat  F^{\alpha_1\cdots \alpha_m}_{\iota a b}$ with $\iota\neq 0$, the explicit results is presented by \eqref{eq:mergedqmh} and \eqref{eq:algebraicF1}. Then, as discussed in Appendix \ref{app:shabiref}, at least for $\iota\leq 20$, the results can be simplified as 
 \begin{equation}\label{eq:expectedgeneralzw}
  \begin{aligned}
  &\langle \hat F^{\alpha_1\cdots\alpha_m}_{\iota a b}\rangle_{z_e}=t^me^{bz_e}\sum_{\substack{0\leq d\leq \iota\\ d+\iota\in \mathbb Z}}\frac{2-\delta(d,0)}{2} e^{-\frac{t}{4} \left(2 d ^2-1\right)}\int\dd x\, e^{-\frac{1}{4} t \left(x^2-2 d x\right)}F_\iota(\frac{x-1}{2}-d,\frac{x-1}{2},\frac{\partial_\eta}{2})\frac{\sinh(x\eta)}{\sinh(\eta)}+O(t^{-\infty}).
  \end{aligned}
\end{equation}
where $F_\iota(\frac{x-1}{2}-d,\frac{x-1}{2},\frac{\partial_\eta}{2})$, also depending on the list $\{\alpha_i\}_{i=1}^m$, is given by 
\begin{equation}\label{eq:algebraicF10}
\begin{aligned}
F_\iota(\frac{x-1}{2}-d,\frac{x-1}{2},\frac{\partial_\eta}{2})=&\delta(\sum_{i=1}^m\alpha_i-a+b,0)\prod_{i=1}^m\frac{1}{(1+|\alpha_i|)^{1/2}}\left(\frac{1}{(\iota+a)!(\iota-a)!(\iota+b)!(\iota-b)!}\right)^{1/2}\\
&\frac{(-1)^{d-2a+b}x(x-2d)}{(x-d+\iota)_{2\iota+1}}\prod_{k=1}^m(\alpha_i \frac{x-1}{2}-\frac{\partial_\eta}{2}+\sum_{i=1}^k\alpha_i)\\
&\sum_{z=0}^{\iota+d}\frac{(-1)^{z}(\iota+d)_{z}(\frac{x-1}{2}+\frac{\partial_\eta}{2}+b-z+\iota)_{\iota+a}(\frac{x-1}{2}-\frac{\partial_\eta}{2}-d-b+z)_{\iota-a}}{z!}\\
&\sum_{z=0}^{\iota-d}\frac{(-1)^{z}(\iota-d)_z(\frac{x-1}{2}+\frac{\partial_\eta}{2}-d-z+\iota)_{\iota-b}(\frac{x-1}{2}-\frac{\partial_\eta}{2}+z)_{\iota+b}}{z!}.
\end{aligned}
\end{equation}
We would like to compare our results \eqref{eq:expectedgeneralzw} with the known results in \cite{liegener2020expectation}.  At first, our formula  \eqref{eq:qmzw} and \eqref{eq:expectedgeneralzw} generalize the known results in literature \cite{liegener2020expectation}, in the sense that our formula gives the results for arbitrary lists $\{\alpha_i\}_{i=1}^m$ of flux indices and triples $(\iota,a,b)$ with at least $\iota\leq 20$. Moreover, with our formula, one can get the expectation values to arbitrary order of $t$. However, in  \cite{liegener2020expectation}, the authors give only the results for the special cases where the list $\{\alpha_i\}_{i=1}^m$ contains at most either a single $-1$, or a single $1$, or a pair of $(-1,1)$. They are all the cases such that  the expectation values have non-vanishing $O(t^0)$ or $O(t)$-term. Other cases are also interesting when we study the higher-order correction, even though the higher-order  correction  is beyond the present work. Our formula reduce to these known results at the special cases. The current work only use these special cases, but our codes \cite{github} are designed based on the generalization formulas \eqref{eq:qmzw}, \eqref{eq:expectedgeneralzw} and Theorem \ref{thm:extendsummation}, since the generalized formulas have the potential in the generalization for computing higher-order correction.

Finally, let us complete this subsection by sketching the algorithm based on Eqs. \eqref{eq:qmzw}, \eqref{eq:expectedgeneralzw} and Theorem \ref{thm:extendsummation} to compute the expectation value of $\hat F^{\alpha_1\cdots\alpha_m}_{\iota a b}$. One can refer to \cite{github} for more details. According to Theorem \ref{thm:extendsummation}, we can simplify (at least for $\iota\leq 20$) the integrals in Eqs. \eqref{eq:qmzw} and \eqref{eq:expectedgeneralzw} to a linear combination of integrals taking the forms
$$I_1=\int_{-\infty}^\infty \dd x\, e^{-a x^2+bx } \frac{\sinh(x \eta)}{x}=\frac{\pi}{2}\left(\text{erfi}\left(\frac{b+\eta}{2\sqrt{a}}\right)-\text{erfi}\left(\frac{b-\eta}{2\sqrt{a}}\right)\right) $$
and
$$I_2=\int_{-\infty}^\infty\dd x\, e^{- a x^2+b x} \mathrm{pol}(x,\partial_z)\frac{e^{\pm \eta x}}{f(z)}\Bigg|_{z=\eta}$$
where $a>0$, $b\in\mathbb R$, $f$ is some function and $\mathrm{pol}(x,\partial_z)$ denotes a polynomial of $x$ and $\partial_z$. This is the first step of our algorithm, without considering the realization of their concrete form for now. Because $I_1$ as a function of $a,b$ is known, the next step of the algorithm is to compute $I_2$. To do this, we first expand $\mathrm{pol}(x,\partial_z)$ to write the integrand of $I_2$ as a linear combination  of $\left(\partial_z^n\frac{1}{f(z)}\right)x^m e^{-ax^2+b x\pm \eta x}.$
Then by substituting the results $\int\dd x\, x^n e^{-a x^2+bx\pm\eta x }$, $I_2$ can be computed easily. By this discussion, the only remaining problem is how to realize the concrete linear combination form of $I_1$ and $I_2$, which can be illustrated by the derivation of $\langle\hat F^{\alpha_1\cdots\alpha_m}_{\iota a b}\rangle_{z_e}$ for $\iota =1$ in Appendix \ref{app:gengjiashabiref}. For this case, the crucial step to simplify $F_1$ is to apply the formula \eqref{eq:key11} and  \eqref{eq:key12} inspired by the proof of Theorem \ref{thm:extendsummation}. Taking advantage of Eq. \eqref{eq:key11}, Eq.\eqref{eq:key12} and the trick \eqref{eq:trick11}, one can finally get \eqref{eq:f1} which is a linear combination of integrals taking the forms of $I_1$ and $I_2$.

\subsection{The cases when all flux indices vanish}
In our computation, we often use the operator
\begin{equation}\label{eq:vanishingfluxind}
\begin{aligned}
D^\iota_{a_1b_1}(h_e)[\hat p_s^0(e)]^{m_1}[\hat p_t^0(e)]^{n_1}\cdots D^\iota_{a_kb_k}(h_e) [\hat p_s^0(e)]^{m_k}[\hat p_t^0(e)]^{n_k}.
\end{aligned}
\end{equation}
To deal with this kind of operators, let us consider the operator $D^\iota_{ab}(h_e)(\hat p_s^0(e))^{m}(\hat p_t^0(e))^{n}$. By applying Proposition \ref{pro:changeO}, it can be simplified to 
\begin{equation}
\begin{aligned}
&D^\iota_{ab}(h_e)[\hat p_s^0(e)]^{m}[\hat p_t^0(e)]^{n}\\
=&[\hat p_s^0(e)]^{m}[\hat p_t^0(e)]^{n}D^\iota_{ab}(h_e)-a t m [\hat p_s^0(e)]^{m-1}[\hat p_t^0(e)]^{n}D^\iota_{ab}(h_e)+bt n[\hat p_s^0(e)]^{m}[\hat p_t^0(e)]^{n-1}D^\iota_{ab}(h_e)+O(t^2).
\end{aligned}
\end{equation}
Then, for the operator \eqref{eq:vanishingfluxind}, it has
\begin{equation}
\begin{aligned}
&D^\iota_{a_1b_1}(h_e)[\hat p_s^0(e)]^{m_1}[\hat p_t^0(e)]^{n_1}\cdots D^\iota_{a_kb_k}(h_e)[\hat p_s^0(e)]^{m_k}[\hat p_t^0(e)]^{n_k}\\
=&[\hat p_s^0(e)]^{\sum_{i=1}^{k}m_i}[\hat p_t^0(e)]^{\sum_{i=1}^{k}n_i}\prod_{i=1}^kD^\iota_{a_ib_i}(h_e)-t\left(\sum_{i=1}^k a_i\left[\sum_{l=i}^km_i\right][\hat p_s^0(e)]^{\left(\sum_{i=1}^{k}m_i\right)-1}[\hat p_t^0(e)]^{\sum_{i=1}^{k}n_i}\prod_{i=1}^kD^\iota_{a_ib_i}(h_e)\right)\\
&+t\left(\sum_{i=1}^k b_i\left[\sum_{l=i}^kn_i\right][\hat p_s^0(e)]^{\sum_{i=1}^{k}m_i}[\hat p_t^0(e)]^{\left(\sum_{i=1}^{k}n_i\right)-1}\prod_{i=1}^kD^\iota_{a_ib_i}(h_e)\right)+O(t^2)
\end{aligned}
\end{equation}
By \eqref{eq:pspt}, we finally have
\begin{equation}\label{eq:allvanish}
\begin{aligned}
&D^\iota_{a_1b_1}(h_e)[\hat p_s^0(e)]^{m_1}[\hat p_t^0(e)]^{n_1}\cdots D^\iota_{a_kb_k}(h_e)[\hat p_s^0(e)]^{m_k}[\hat p_t^0(e)]^{n_k}\\
=&(-1)^{\sum_{i=1}^k n_i}[\hat p_s^0(e)]^{\sum_{i=1}^{k}(m_i+n_i)}\prod_{i=1}^kD^\iota_{a_ib_i}(h_e)-t(-1)^{\sum_{i=1}^k n_i}\left(\sum_{i=1}^k a_i\left[\sum_{l=i}^km_l\right]+\sum_{i=1}^k b_i\left[\sum_{l=i}^kn_l\right]\right)\times\\
&[\hat p_s^0(e)]^{\sum_{i=1}^{k}(m_i+n_i)-1}\prod_{i=1}^kD^\iota_{a_ib_i}(h_e).
\end{aligned}
\end{equation}
Recalling the derivation of $\langle \hat F^{\alpha_1\cdots\alpha_m}_{\iota a b}\rangle_{z_e}$, we get
\begin{equation}
\begin{aligned}
\langle (\hat p_s^0(e))^mD^\iota_{ab}(h_e)\rangle_{z_e}=e^{b\eta}\left(- t\frac{\partial_\eta}{2}\right)^m e^{-b\eta} \langle D^\iota_{ab}(h_e)\rangle_{z_e}
\end{aligned}
\end{equation}
where the result of $\langle (\hat p_s^0(e))^m\rangle_{z_e}$ is given by setting $\iota=0=a=b$.
It can be verified that, $\langle D^\iota_{ab}(h_e)\rangle_{z_e}$ takes the form that 
\begin{equation}
\langle D^\iota_{ab}(h_e)\rangle_{z_e}=\langle 1\rangle_{z_e}(g_0+tg_1(\eta)+O(t^2))=f(t)\frac{e^{\frac{\eta^2}{t}}\eta}{\sinh(\eta)} (g_0+tg_1(\eta)+O(t^2))
\end{equation}
with some functions $g_0$, $g_1$ and $f$. Therefore, with Fa\`a di Bruno's formula, we can have that
\begin{equation}
\begin{aligned}
&\langle (\hat p_s^0(e))^mD^\iota_{ab}(h_e)\rangle_{z_e}=e^{b\eta}\left(-t\frac{\partial_\eta}{2}\right)^m e^{-b\eta} \langle D^\iota_{ab}(h_e)\rangle_{z_e}\\
=&\langle 1\rangle_{z_e} \left(-\eta\right)^m[g_0+tg_1(\eta)]+\langle 1\rangle_{z_e} \frac{m(m+1)}{4}\left(-\eta\right)^{m-2} g_0 t\\
&+\langle 1\rangle_{z_e}\frac{m}{2} (-\eta)^{m-1} (\coth (\eta )+b)g_0 t+O(t^2)
\end{aligned}
\end{equation}
Based on these formula, we can propose a faster algorithm to deal with these cases.

\section{Power counting}\label{se:leading}
After introducing the derivations of expectation values of several characterized operators,
 we finally need to deal with a set of specific operators that takes the following form, $\sum_{\vec \alpha}\mathcal T^{\alpha_1\alpha_2\cdots\alpha_m}\hat O_{\alpha_1\alpha_2\cdots\alpha_m}$, where $\mathcal T$ is some numerical factors and $\hat O$ is some polynomial operators of holonomies and fluxes. In principle, we would need to compute the expectation values of $\hat O_{\alpha_1\cdots\alpha_m}$ for all indices $\vec\alpha=(\alpha_1,\cdots,\alpha_m)$. This computation can be preformed thanks to previous sections. However, the computational complexity comes from the huge amount of terms in the sum over $\vec \alpha$. Since  we are only interested in the expectation value up to $O(t)$, the complexity can be reduced by certain power-counting argument: we count the least power of $t$ contains in each $\langle \hat O_{\alpha_1\cdots\alpha_m}\rangle$ before explicit computation, then we omit those terms only contribute to higher order than $O(t)$ in $\langle \widehat{H[N]}\rangle$. It turns out that a large degree of complexity can be reduced in this manner. The following arguments in this section will be proven rigorously in Appendix \ref{sec:matrixelementhp}.

In this section, we will denote $\Psi_\g$ defined in \eqref{eq:coherentstate} by $|\Psi_{\vec g}\rangle$ with $\vec g=\{g_e\}_{e\in E(\gamma)}$, namely
\begin{equation}\label{eq:coherentstatepsig}
|\Psi_{\vec g}\rangle=\bigotimes_{e\in E(\gamma)}|\psi_{g_e} \rangle.
\end{equation}
Similarly, $|\Psi_{\vec g^{(i)}}\rangle$ denotes the coherent state that at the edge $e$ is  $|\psi_{g_e^{(i)}}\rangle$. Let $\hat O$ take the form of
\begin{equation}\label{eq:operatorO}
\hat O=\hat O_1\hat O_2\cdots\hat O_k
\end{equation}
with $\hat O_i$ being arbitrary polynomial of fluxes and holonomies.
Inserting the resolution of identity \eqref{eq:completerelation}, we have
\begin{equation}
\begin{aligned}\label{eq:insertidentity}
\langle\Psi_{\vec g}|\hat O|\Psi_{\vec g}\rangle=\int\prod_{m=1}^{k-1}\dd \nu(\vec g^{(m)}) \prod_{i=1}^k\langle\Psi_{\vec g^{(i-1)}}|\hat O_i|\Psi_{\vec g^{(i)}}\rangle
\end{aligned}
\end{equation}
where $|\Psi_{\vec g^{(0)}}\rangle=|\Psi_{\vec g^{(k)}}\rangle:=|\Psi_{\vec g}\rangle$ and the measure $\dd\nu(\vec g^{(m)})$ is
\begin{equation}\label{eq:entiremeasure}
\dd \nu(\vec g^{(m)})=\prod_{e\in E(\gamma)}\dd\nu(g_e^{(m)})
\end{equation}
with $\dd\nu(g_e^{(m)})$ defined in \eqref{eq:measure}. Eq. \eqref{eq:insertidentity} relates the expectation value of $\hat O$ to matrix elements of each individual $\hat O_i$. Thus we are motivated to study matrix elements of polynomial of holonomies and flux. One  can refer  to  Appendix \ref{sec:matrixelementhp}  for more details on this issue.   According to the analysis therein, the matrix elements of the fluxes and holomomies are of a form described below 
\begin{equation}
\langle\psi_{g_e}|\hat O_i|\psi_{g_e'}\rangle=\langle\psi_{g_e}|\psi_{g_e'}\rangle\left(E_0(g_e,g_e')+tE_1(g_e,g_e')+O(t^\infty)\right).
\end{equation}

%\subsection{The leading order of the expectation value}\label{se:leading}
Assigning to each edge $e$ a complex number $w_e=p_e-i\theta_e$, we have the coherent state 
\begin{equation}
|\Psi_{\vec w}\rangle:=\bigotimes_{e\in E(\gamma)}|\psi_{w_e}\rangle.
\end{equation}
For the operator $\hat O$ in Eq. \eqref{eq:operatorO}, we state the following result obtained firstly in \cite{giesel2007algebraic}

\begin{thm}\label{thm:leadingordergeneral}
Consider an operator $\hat O=\prod_{i=1}^k\hat O_i$. Assume that, for each operator $\hat O_i$, its matrix elements $\langle \Psi_{\vec g^{(1)}}|\hat O_i|\Psi_{\vec g^{(2)}}\rangle$ take the following form  
\begin{equation}\label{eq:matrixelementform}
\langle \Psi_{\vec g^{(1)}}|\hat O_i|\Psi_{\vec g^{(2)}}\rangle=\langle \Psi_{\vec g^{(1)}}|\Psi_{\vec g^{(2)}}\rangle\left(E_0^{(i)}(\vec g^{(1)},\vec g^{(2)})+t E_1^{(i)}(\vec g^{(1)},\vec g^{(2)})+O(t^\infty)\right).
\end{equation}
Let $N_0$ be the number of operators  $\hat O_m\in \{\hat O_i\}_{i=1}^k$ such that
\begin{equation}\label{eq:operatorom}
\frac{\langle \Psi_{\vec w}|\hat O_m|\Psi_{\vec w}\rangle}{\langle \Psi_{\vec w}|\Psi_{\vec w}\rangle}=O(t),
\end{equation}
where the $O(t^0)$ term vanishes on the RHS. 
Then the expectation value of $\hat O$ with respect to the coherent state $|\Psi_{\vec w}\rangle$ satisfies 
\begin{equation}
\frac{\langle \Psi_{\vec w}|\hat O|\Psi_{\vec w}\rangle}{\langle \Psi_{\vec w}|\Psi_{\vec w}\rangle}=O(t^n),\ \text{with }n\geq {\floor{\frac{N_0+1}{2}}}
\end{equation}
where $\floor{x}$ is the largest integer no larger than $x$. 
\end{thm}
A detailed proof of the above result is provided in Appendix \ref{sec:matrixelementhp}, including a careful stationary phase analysis, the computation of nondegenerate Hessian matrix, and power-counting.

Because of the vanishing leading-order term of  $\langle\frac{\hat Q}{\langle\hat Q\rangle}-1\rangle$, it can be regarded as operator $\hat O_m$ satisfying \eqref{eq:operatorom}.
Thus, Theorem \ref{thm:leadingordergeneral} is applied to count the power of $t$ for the term including $\left(\frac{\hat Q^2}{\langle\hat Q\rangle^2}-1 \right)^k$ in $\langle \widehat{H[N]}\rangle$. Moreover, in order to apply Theorem \ref{thm:leadingordergeneral}, matrix elements of $\hat O_i$ have to be computable. We have to factorize $\frac{\hat Q^2}{\langle\hat Q\rangle^2}-1 =\left(\frac{\hat Q}{\langle\hat Q\rangle}+1 \right)\left(\frac{\hat Q}{\langle\hat Q\rangle}-1 \right)$ because every matrix element of $\hat Q$ is a polynomial of matrix elements of the flux operators, while that of $\hat Q^2$ is not.

Because the expectation values of $\hat p_s^{\pm 1}(e)$, $\hat p_t^{\pm 1}(e)$ and $D^\iota_{ab}(h_e)~(a\neq b)$ with respect to $\Psi_{\vec \omega}$ vanish, each of them can also be considered as operator $\hat O_m$ in \eqref{eq:operatorom}. Therefore, this theorem can be applied to study the leading order of  monomial of holonomies and fluxes. Let us use {$\hat p^{\beta}(e)$} to denote either $\hat p_s^\beta(e)$ or $\hat p_t^\beta(e)$, and use $\mathcal M$ to denote the monomial of holonomies and fluxes. Let $N_\pm$ be the number of $\hat p^{\pm 1}(e)$ respectively and, $M_+$ (respectively $M_-$) be the number of $D^{\frac{1}{2}}_{-\frac{1}{2}\frac12}(h_e)$ (respectively $D^{\frac{1}{2}}_{\frac{1}{2},-\frac12}(h_e)$) in $\mathcal M$. According to our analysis above, the expectation value of $\mathcal M$ with respect to the coherent state $|\psi_{z_e}\rangle$ with $z_e\in \mathbb C$ is non-vanishing if
\begin{equation}
\sum_{i=1}^{m}\beta_i+\sum_{j=1}^k(b_j-a_j)=0. 
\end{equation}
 Hence, we have
 \begin{equation}
N_++M_+=N_-+M_-.
\end{equation}
Therefore, this theorem gives us that the leading order the expectation value $\langle \mathcal M\rangle_{z_e}$ is $O(t^{N_++M_+})$ or higher. We have more discussions on this case. Since the matrix elements of $\hat p^\alpha(e)$ and $D^{\frac 12}_{ab}(h_e)$ are computable, the results on the leading order of $\mathcal M$ can be calculated more concretely. The result is summarized as the following theorem.

\begin{thm}\label{thm:leadingordermultiplyPH}
Given $\mathcal M$ an arbitrary monomial of holonomies and fluxes. 
Let $\mathcal M'$ be the operator resulting from $\mathcal M$ by deleting all factors $\hat p^0(e)$ and $D^{\frac{1}{2}}_{aa}(h_e)$.  Denote the number of $\hat p_s^0(e)$ and $\hat p_t^0(e)$ in $\mathcal M$ as $N_{0,s}$ and $N_{0,t}$ respectively, and the number of $D^{\frac 1 2}_{\frac12 \frac12}(h_e)$ and $D^{\frac 1 2}_{-\frac12 -\frac12}(h_e)$ as $M_{0+}$ and $M_{0-}$ respectively. Then the leading order of $\langle \mathcal M\rangle_{z_e} $ is exactly $O(t^{M_++N_+})$ if and only if the leading order of $\langle \mathcal M'\rangle_{z_e} $ is exactly $O(t^{M_++N_+})$, {where $N_\pm$ be the number of $\hat p^{\pm 1}(e)$ respectively and, $M_+$ (respectively $M_-$) be the number of $D^{\frac{1}{2}}_{-\frac{1}{2}\frac12}(h_e)$ (respectively $D^{\frac{1}{2}}_{\frac{1}{2},-\frac12}(h_e)$) in $\mathcal M$.}
Moreover, 
%for the case when $\langle \mathcal M\rangle_{z_e} $ and $\langle \mathcal M'\rangle_{z_e} $ are exactly $O(t^{M_++N_+})$, 
it has
 \begin{equation}\label{eq:leadingequalmultileading}
 \langle \mathcal M \rangle_{z_e} \cong (\langle \hat p_s^0(e)\rangle_{z_e})^{N_{0,s}}(\langle \hat p_t^0(e)\rangle_{z_e})^{N_{0,t}}(\langle D^{\frac 12}_{\frac12\frac12}(h_e)\rangle_{z_e})^{M_{0,+}}(\langle D^{\frac 12}_{-\frac12-\frac12}(h_e)\rangle_{z_e})^{M_{0,-}}\langle \mathcal M'\rangle_{z_e}
 \end{equation}  
 where $\cong $ means the %leading-order terms, i.e. the 
 $O(t^{M_++N_+})$ terms of the left and right hand sides are equal to each other. 
 \end{thm} 
 The proof of this theorem is quite technical and, thus, presented in Appendix \ref{app:proofofthem44}

 \section{Cosmological expectation value}\label{sec:cosmodel}
We apply our computation of expectation values to coherent states labelled by homogeneous and isotropic data. The symmetry group of the homogeneous and isotropic cosmology is $\mathbb T \rtimes F$ where  $F$ is the isotropic subgroup and $\mathbb T$ is the translation subgroup. Denote the subgroup of $\mathbb T \rtimes F$ preserving $\gamma$ by  $S_\gamma$. 
A classical state $\g$ is said to be symmetric with respect to $S_\gamma$ if $s^*\g:=\g \circ s$ is identical with $\g$ up to a gauge transformation $s$ ($\forall s\in S_\gamma$). According to this definition, classically symmetric states $\g$ are of the form \cite{alesci2014quantum}
\begin{equation}\label{eq:coscoh}
\g:e\mapsto g_e=n_e e^{iz\tau_3} n_e^{-1}
\end{equation}
with $n_e\in$SU(2) satisfying
\begin{equation}
n_e\tau_3 n_e^{-1}=\vec n_e\cdot\vec\tau.
\end{equation}
 In the last equation, $\vec n_e$ is the unit vector pointing to direction of edge $e$. Then, for each $s=(t,f)\in \mathbb T \rtimes F$, it can be verified that
 \begin{equation}\label{eq:gaugeAndDiff}
 \g\circ s={\rm Ad}_f\circ \g
 \end{equation}
 where ${\rm Ad}_f \circ \g(e)=f \g(e) f^{-1}$ for all $e\in E(\gamma)$.

\subsection{Symmetries of the expectation value} \label{sec:symmetry1}
Given $\hat F_e$ as a polynomial of fluxes and holonomies on $e$.  For $s=(t,f)\in\mathbb T\rtimes F$, Eq. \eqref{eq:gaugeAndDiff} results in 
\begin{equation}\label{eq:DifftoOp}
\langle \psi_{g_{s(e)}}|\hat F_{s(e)}|\psi_{g_{s(e)}}\rangle=\langle \psi_{f g_{e}f^{-1}}|\hat F_{e}|\psi_{fg_{e}f^{-1}}\rangle=\langle \psi_{g_{e}}|(f\triangleright \hat F_{e})|\psi_{g_{e}}\rangle.
\end{equation}
where $f\triangleright \hat F_e$ denote the gauge transformed operator of $\hat F_e$ by $f$ and the last equality can be derived by using the similar procedure as to derive Eq. \eqref{eq:basicformula}. 

To expand the expectation value of $\hat H_E$ and $\hat H_L$ to order $O(t)$, one needs to replace the operator $\hat V_v$ by $\hat V_{GT}^{(v)}$ defined in \eqref{eq:VGT}. 
Then the Euclidean part $\widehat{H_E[N]}$ is rewritten in terms of (there is no summation over $I,J,K$ here)
\begin{equation}\label{eq:hen}
\hat{H}_E^{(n)}(v;e_I,e_J,e_K)=\frac{1}{i\beta a^2 t}\epsilon_{IJK}\tr(h_{\alpha_{IJ}}[h_{e_K},\hat Q_v^{2n}]h_{e_K}^{-1}),
\end{equation}
and the Lorentzian part, in terms of 
\begin{equation}\label{eq:hln}
\begin{aligned}
&\hat{H}_L^{(\vec k)}(v;v_1,v_2,v_3,v_4;e_I,e_J,e_K)\\
=&\frac{-1}{2i\beta^7 a^{10} t^5}\epsilon^{IJK}\tr( [h_{e_I},[\hat Q_{v_1}^{2k_1},\hat{H}_E^{(k_2)}(v_2)]]h_{e_I}^{-1} [h_{e_J},[\hat Q_{v_3}^{2k_3},\hat{H}_E^{(k_4)}(v_4)]]h_{e_J}^{-1}[h_{e_K},\hat Q_v^{2k_5}]h_{e_K}^{-1})
\end{aligned}
\end{equation}
with $\vec k=(k_1,k_2,k_3,k_4,k_5)$. Define 
\begin{equation}\label{eq:HEnv}
\hat H_E^{(n)}(v)=\sum_{e_I,e_J,e_K}\hat H_E^{(n)}(v;e_I,e_J,e_K)
\end{equation}
 and
 \begin{equation}\label{eq:HLnv}
\hat H_L^{(\vec k)}(v)=\sum_{v_1,v_2,v_3,v_4,e_I,e_J,e_K}\hat{H}_L^{(\vec k)}(v;v_1,v_2,v_3,v_4;e_I,e_J,e_K).
 \end{equation}
The Euclidean and Lorentzian parts, with the replacement $\hat V_v\to \hat V_{GT}^{(n)}$ truncated at a finite $n$, are linear combinations of $\hat H_E^{(n)}(v)$ and $\hat H_L^{(\vec k)}(v)$ with various $n$ and $\vec k$ respectively.

%%%%%%%%%%%%

\subsubsection{Symmetries of the Euclidean part}

According to Eq. \eqref{eq:DifftoOp} and the gauge invariance of $\hat{H}_E^{(n)}(v;e_I,e_J,e_K)$,
one realizes the following symmetry
\begin{equation}\label{eq:HEDiffInv}
\begin{aligned}
\langle \hat{H}_E^{(n)}(v;e_I,e_J,e_K)\rangle&=\langle \hat{H}_E^{(n)}(s(v);s(e_I),s(e_J),s(e_K))\rangle,\\
\end{aligned}
\end{equation}
where $\langle\cdot\rangle$ denotes the expectation value with respect to the cosmological coherent state given by \eqref{eq:coscoh}, and $s=(t,f)$ is a symmetry of the graph.

By this relation, Eq. \eqref{eq:HEnv} is simplified as
\begin{equation}\label{eq:HEnvp}
\hat H_E^{(n)}(v)=24(\hat H_E^{(n)}(v;e_x^+,e_y^+,e_z^+)+\hat H_E^{(n)}(v;e_x^+,e_y^+,e_z^-))
\end{equation}
where the prefactor 24 is deduced by the fact that there are totally 48 terms in the RHS of \eqref{eq:HEnv}.

Moreover, 
$[h_{e_z^\pm},\hat Q_v^{2n}]h_{e_z^\pm}^{-1}$ appearing in $\hat{H}_E^{(n)}(v;e_x^+,e_y^+,e_z^\pm)$ potentially relates $\hat{H}_E^{(n)}(v;e_x^+,e_y^+,e_z^+)$ with $\hat{H}_E^{(n)}(v;e_x^+,e_y^+,e_z^-)$.Concisely,
\begin{equation}\label{eq:ezpm}
\begin{aligned}
[h_{e^\pm_z},\hat Q_v^{2n}]h_{e^\pm_z}^{-1}=&\sum_{l=1}^{2n}\sum_{\mathcal P_l}\Big(\mp it\frac{(\beta a^2)^3}{8}\Big)^l\hat Q_v^{p_1}\epsilon_{\alpha_1\beta_1\gamma_1}\hat{X}^{\alpha_1}\hat{Y}^{\beta_1}\tau^{\gamma_1}\hat Q_v^{p_2}\epsilon_{\alpha_2\beta_2\gamma_2}\hat{X}^{\alpha_2}\hat{Y}^{\beta_2}\tau^{\gamma_2}
\cdots \hat Q_v^{p_l} \epsilon_{\alpha_l\beta_l\gamma_l}\hat{X}^{\alpha_l}\hat{Y}^{\beta_l}\tau^{\gamma_l}\hat Q_v^{p_{l+1}}
\end{aligned}
\end{equation}
where the edges $e_z^\pm$ are oriented so that $s(e_z^+)=v=s(e_z^-)$, $\hat X^\alpha=\hat p_s^\alpha(e')-\hat p_t^\alpha(e')$ and $\hat Y^\alpha=\hat p_s^\alpha(e'')-\hat p_t^\alpha(e'')$, and $\mathcal P=\{p_1,p_2,\cdots,p_{l+1}\}$ with $p_i\in\mathbb Z$, $p_i\geq 0$ and $\sum_{i=1}^{l+1}p_i=2n-l$.\footnote{A general equation can be obtained analogously if the holonomy $h_{e_z^\pm}$ is replaced by a holonomy along other edges. In the following context, Eq. \eqref{eq:ezpm} will be usually referred as this general equation. }
Substituting the last equation into the expression of $\hat H_E^{(n)}(v)$, one has that
\begin{equation}\label{eq:HEezpezm}
\hat{H}_E^{(n)}(v;e_x^+,e_y^+,e_z^+)+\hat{H}_E^{(n)}(v;e_x^+,e_y^+,e_z^-)=2\hat{\tilde{H}}_E^{(n)}(v;e_x^+,e_y^+,e_z^+)
\end{equation}
where $\hat{\tilde{H}}_E^{(n)}(v;e_x^+,e_y^+,e_z^+)$ is the operator $\hat{H}_E^{(n)}(v;e_x^+,e_y^+,e_z^+)$ with applying the following replacement 
\begin{equation}\label{eq:ezpmreplacement}
[h_{e_z^+},\hat Q_v^{2n}]h_{e_z^+}^{-1}\to \sum_{l \text{ is odd} }\sum_{\mathcal P_l}\Big(- it\frac{(\beta a^2)^3}{8}\Big)^l\hat Q_v^{p_1}\epsilon_{\alpha_1\beta_1\gamma_1}\hat{X}^{\alpha_1}\hat{Y}^{\beta_1}\tau^{\gamma_1}\hat Q_v^{p_2}\epsilon_{\alpha_2\beta_2\gamma_2}\hat{X}^{\alpha_2}\hat{Y}^{\beta_2}\tau^{\gamma_2}
\cdots \hat Q_v^{p_l} \epsilon_{\alpha_l\beta_l\gamma_l}\hat{X}^{\alpha_l}\hat{Y}^{\beta_l}\tau^{\gamma_l}\hat Q_v^{p_{l+1}}.
\end{equation} 

By Eq. \eqref{eq:HEnvp}, $\hat H_E^{(n)}(v)$ becomes $\hat H_E^{(n)}(v)=48 \hat{\tilde{H}}_E^{(n)}(v;e_x^+,e_y^+,e_z^+)$. Thus,
when we calculate the expectation value of the Euclidean part, it is only necessary to consider $\hat{\tilde{H}}_E^{(n)}(v;e_x^+,e_y^+,e_z^+)$ rather than $\hat{H}_E^{(n)}(v;e_x^+,e_y^+,e_z^\pm)$. 

Further, according to Eq. \eqref{eq:ezpm}, the Euclidean Hamiltonian is of the form $$\hat H_E^{(n)}(v;e_I,e_J,e_K)=\epsilon_{IJK}\tr(h_{\alpha_{IJ}}\tau^\alpha)\hat O_\alpha,$$
where $\hat O_\alpha$ is a polynomial of fluxes. 
Then the fact $\tr(h\tau^\alpha)=-\tr(h^{-1}\tau^\alpha)$ gives
\begin{equation}\label{eq:extrasymmtryofHE}
\hat H_E^{(n)}(v;e_I,e_J,e_K)=\hat H_E^{(n)}(v;e_J,e_I,e_K).
\end{equation}

In summary, originally there are totally 48 terms for every $\hat H_E^{(n)}(v)$ in Eq. \eqref{eq:HEnv}. However, thanks to the symmetries discussed in this section, we have
$\hat H_E^{(n)}(v)=48 \hat{\tilde{H}}_E^{(n)}(v;e_x^+,e_y^+,e_z^+)$, which means that only the expectation value of  $\hat{\tilde{H}}_E^{(n)}(v;e_x^+,e_y^+,e_z^+)$ is necessary to be computed.

%%%%%%%%%%%%%%%%%%%%%%%%%%
\subsubsection{Symmetries of the Lorentzian part} 
Considering a list of vertices and edges $(v;v_1,v_2,v_3,v_4;e_I,e_J,e_K)$ with $e_I,\ e_J$ and $e_K$ being outgoing from $v$, we have that $(e_I, e_J, e_K)$ is either left-handed or right-handed. Thus,
there exists a rotation $f$ which leaves $v$ invariant such that $(v;f(v_1),f(v_2),f(v_3),f(v_4);f(e_I),f(e_J),f(e_K))$ is either 
$$(v;f(v_1),f(v_2),f(v_3),f(v_4);e_x^+,e_y^+,e_z^+)$$ or 
$$(v;f(v_1),f(v_2),f(v_3),f(v_4);e_x^+,e_y^+,e_z^-).$$ Therefore, \eqref{eq:HLnv} is simplified to
\begin{equation}\label{eq:HLnvp}
\hat H_L^{(\vec k)}(v)=24\sum_{v_1,v_2,v_3,v_4}\left(\hat{H}_L^{(\vec k)}(v;v_1,v_2,v_3,v_4;e_x^+,e_y^+,e_z^+)+\hat{H}_L^{(\vec k)}(v;v_1,v_2,v_3,v_4;e_x^+,e_y^+,e_z^-)\right)
 \end{equation}
Moreover, since the term $[h_{e_z^+},\hat Q_v^{2n}]h_{e_z^+}^{-1}$ appears in $\hat H_L^{(\vec k)}$ too, Eq. \eqref{eq:ezpm} can be applied again to simplify Eq. \eqref{eq:HLnvp} to obtain the following expression
\begin{equation}\label{eq:HLnvpp}
\hat H_L^{(\vec k)}(v)=48\sum_{v_1,v_2,v_3,v_4}\hat{\tilde H}_L^{(\vec k)}(v;v_1,v_2,v_3,v_4;e_x^+,e_y^+,e_z^+),
 \end{equation}
 where this $\hat{\tilde H}_L^{(\vec k)}$ operator is given by $\hat{H}_L^{(\vec k)}(v;v_1,v_2,v_3,v_4;e_x^+,e_y^+,e_z^+)$ with the replacement \eqref{eq:ezpmreplacement}.  As a consequence, it is only necessary to compute the expectation value of 
 $$\hat{\tilde H}_L^{(\vec k)}(v;v_1,v_2,v_3,v_4;e_x^+,e_y^+,e_z^+)$$
 for different vertices $v_1,\, v_2,\, v_3$ and $v_4$.

The above discussion simplifies the computation of the Lorentzian part. 
However, more symmetries are required in order to reduce the computation time to an acceptable level.
For this purpose, let us firstly look at the term $[h_e[\hat Q_{v_1}^m,h_{e}^{-1}],\hat Q_{v_2}]$ which is from the commutator between the volume and the Euclidean part. We obtain the following proposition which can be proven by Eq. \eqref{eq:ezpm} directly. 
\begin{pro}\label{pro:simplifyQ}
Given an edge $e$ with the source $s(e)$ and the target $t(e)$, $
[h_e[\hat Q_{s(e)}^m,h_{e}^{-1}],\hat Q_{v}]= 0$ for all $v\neq s(e)$.
\end{pro}

With this proposition, we consider the commutator $ [\hat Q_{v_1}^{2k},\hat H_E^{(n)}(v_2;e_I,e_J,e_K)]$ which defines the operator $\hat K$ as
\begin{equation}\label{eq:Koperator}
\hat K=\frac{1}{it}[\hat V,\hat H_E].
\end{equation}
By definition, we have
\begin{equation}\label{eq:quantumComHV}
\begin{aligned}
 [\hat Q_{v_1}^{2k},\hat H_E^{(n)}(v_2;e_I,e_J,e_K)]
 =\frac{2}{i\beta a^2 t}(\hat K_1+\hat K_2)
\end{aligned}
\end{equation}
with
\begin{equation}
\begin{aligned}
K_1:=&\epsilon_{IJK}\tr([\hat \hat Q_{v_1}^{2k},h_{\alpha_{IJ}}][h_{e_K},\hat Q_{v_2}^{2n}]h_{e_K}^{-1})\\
\hat K_2:=&\epsilon_{IJK}\tr(h_{\alpha_{IJ}}\big[\hat Q_{v_1}^{2k},[h_{e_K},\hat Q_{v_2}^{2n}]h_{e_K}^{-1}\big]).
\end{aligned}
\end{equation}
The classical analogy of Eq. \eqref{eq:Koperator} is
\begin{equation}
K=\{V,H_E\}. 
\end{equation}
Substituting the expression of $H_E$, one has
\begin{equation}\label{eq:classicalK}
K=\{V,H_E\}=\int \dd^3 x\{V, F_{ab}^i(x)\}\frac{\epsilon_{ijk}E^a_i E^b_j}{\sqrt{\det(E)}}.
\end{equation}
According to Eq. \eqref{eq:classicalK}, only the Poisson bracket between volume $V$ and the curvature $F_{ab}^i$ is involved in the classical expression of $K$. In the quantum theory, $F_{ab}^i$ is quantized to a holonomy along some loop $\alpha_{IJ}$. Thus, comparing to Eq. \eqref{eq:quantumComHV}, the operator $\hat K_1$ corresponds to the RHS of Eq. \eqref{eq:classicalK}, while $\hat K_2$ gives an extra term in $\hat K$. According to Proposition \ref{pro:simplifyQ}, this extra term $\hat K_2$ vanishes unless $v_1=v_2=s(e_K)$ at which Eq. \eqref{eq:ezpm} can be applied to cancel the holonomies inside the commutators of $\hat K_2$. Then $\hat K_2$ is simplified to the following form
\begin{equation*}
\tr(h_{\alpha_{IJ}}[\hat Q_{v_1}^{2k},\text{polynomial of only fluxes}]).
\end{equation*}
Therefore, it is because of the non-commutativity between the flux operators that
the operator $\hat K_2$ appears in $\hat K$. Note that the existence of $\hat{K}_2$ does not affect the continuum limit of $\lim_{t\to0}\langle\widehat{H[N]}\rangle$ (the classical limit of $\langle\widehat{H[N]}\rangle$ reduces to the classical continuum expression of $H[N]$ when the sizes of lattice edges are neglected \cite{Han:2020chr}). %once we consider the continue limit, the fluxes become the fields $E^a_i(x)$ which are commutative with respect to each other. That is, the  non-commutativity feature of the fluxes is no longer kept once the continue limit is considered. This discussion gives an explanation on the non-compatibility caused by $\hat K_2$ between $K$ and its operator correspondence $\hat K$. Moreover, because the non-commutativity between fluxes is considered in $\hat K_2$, let us give more details on its simplification.

By Eq. \eqref{eq:ezpmreplacement}, $\hat K_2$ can be simplified as
\begin{equation}
\begin{aligned}
\hat K_2=&\epsilon_{IJK} \sum_{p_1+p_2=2n-1}(2k)\frac{-it (\beta a^2)^3}{8}\,\tr(h_{\alpha_{IJ}}\tau^\gamma)\hat Q_{v}^{p_1}\big[\hat Q_{v},\epsilon_{\alpha_1\beta_1\gamma_1}\hat X_I^{\alpha_1}\hat X_J^{\beta_1}\big]\hat Q_{v}^{2k-1+p_2}\\
&+\epsilon_{IJK}\sum_{p_1+p_2=2n-1}\frac{2k(2k-1)}{2}\frac{-it (\beta a^2)^3}{8}\,\tr(h_{\alpha_{IJ}}\tau^\gamma)\hat Q_{v}^{p_1}\big[\hat Q_{v},\big[\hat Q_{v},\epsilon_{\alpha_1\beta_1\gamma_1}\hat X_I^{\alpha_1}\hat X_J^{\beta_1}\big]\big]\hat Q_{v}^{2k-2+p_2}\\
&+O(t^4)
\end{aligned}
\end{equation}
where $\hat X_I^\alpha=\hat p_s^\alpha(e_I)-\hat p_t^\alpha(e_I)$ and 
the conclusion that $\hat K_2\neq0$ if $v_1= v_2\equiv v$ is used.
For the first term, we have up to $O(t^4)$
\begin{equation}
\begin{aligned}
\text{first term}
=&\epsilon_{IJK}(2n)(2k)\frac{-it (\beta a^2)^3}{8}\,\tr(h_{\alpha_{IJ}}\tau^\gamma)\hat Q_{v}^{2n}\big[\hat Q_{v},\epsilon_{\alpha_1\beta_1\gamma_1}\hat X_I^{\alpha_1}\hat X_J^{\beta_1}\big]\hat Q_{v}^{2k-2}\\
&-(2k)\frac{2n(2n+1)}{2}\frac{-it (\beta a^2)^3}{8}\,\tr(h_{\alpha_{IJ}}\tau^\gamma)\hat Q_{v}^{2n-1}\big[\hat Q_v,\big[\hat Q_{v},\epsilon_{\alpha_1\beta_1\gamma_1}\hat X_I^{\alpha_1}\hat X_J^{\beta_1}\big]\big]\hat Q_{v}^{2k-2}
\end{aligned}
\end{equation}
For the second term, up to $O(t^4)$ we have
\begin{equation}
\begin{aligned}
\text{second term}=\epsilon_{IJK}(2n)\frac{2k(2k-1)}{2}\frac{-it (\beta a^2)^3}{8}\,\tr(h_{\alpha_{IJ}}\tau^\gamma)\hat Q_{v}^{2n-1}\big[\hat Q_{v},\big[\hat Q_{v},\epsilon_{\alpha_1\beta_1\gamma_1}\hat X_I^{\alpha_1}\hat X_J^{\beta_1}\big]\big]\hat Q_{v}^{2k-2}
\end{aligned}
\end{equation}
Finally, $\hat K_2$ is
\begin{equation}\label{eq:K2}
\begin{aligned}
\hat K_2=&\epsilon_{IJK}(2n)(2k)\frac{-it (\beta a^2)^3}{8}\,\tr(h_{\alpha_{IJ}}\tau^\gamma)\hat Q_{v}^{2n}\big[\hat Q_{v},\epsilon_{\alpha_1\beta_1\gamma_1}\hat X_I^{\alpha_1}\hat X_J^{\beta_1}\big]\hat Q_{v}^{2k-2}\\
&+\epsilon_{IJK}\frac{(2k) (2n) (2k-2n-2)}{2}\frac{-it (\beta a^2)^3}{8}\,\tr(h_{\alpha_{IJ}}\tau^\gamma)\hat Q_{v}^{2n-1}\big[\hat Q_v,\big[\hat Q_{v},\epsilon_{\alpha_1\beta_1\gamma_1}\hat X_I^{\alpha_1}\hat X_J^{\beta_1}\big]\big]\hat Q_{v}^{2k-2}\\
&+O(t^4)
\end{aligned}
\end{equation}
Because of the commutators between fluxes operators,
\begin{equation}
\begin{aligned}
[\hat p^\alpha(e),\hat p^\beta(e)]=&t(-1)^{\gamma}\varepsilon_{-\gamma\alpha\beta}\hat p^\gamma(e)=:t C_{\alpha\beta\gamma}\hat p^\gamma(e)
\end{aligned}
\end{equation}
with $\varepsilon_{-1,0,1}=1$ and $p^\alpha(e)$ denoting $p_t^\alpha(e)$ or $p_s^\alpha(e)$, one obtains the following
\begin{equation}\label{eq:commutatorPssPt}
[\hat p_s^\alpha(e^+)+s_1\hat p_t^\alpha(e^-),\hat p_s^\beta(e^+)+s_2\hat p_t^\beta(e^-)]=tC_{\alpha\beta\gamma}(\hat p_s^\gamma(e^+)+s_1s_2\hat p_t^\gamma(e^-)),
\end{equation}
with $s_1,s_2=\pm 1$.
 Substituting Eq. \eqref{eq:commutatorPssPt} into Eq. \eqref{eq:K2}, we express $\hat K_2$, as well as $\hat K$, as a polynomial of $h_e$ and $p_s^\alpha(e^+)\pm p_t^\alpha(e^-)$.

Moreover, thanks to the above results, $\hat H_L^{(\vec k)}(v)$ in Eq. \eqref{eq:HLnvpp} can finally be simplified to be in terms of
\begin{equation}\label{eq:cruxofsimplification}
\frac{C}{t^2}\tr(h_{e_x^+}\, F_1 h_{e_x^+}^{-1}h_{e_y^+} F_2 h_{e_y^+}^{-1} G_1)
\end{equation}
where $C$ is some constant of order $t^0$ or higher,  $F_i$ with $i=1,2$ are some monomials of holonomies and ($p_s^\alpha(e^+)\pm p_t^\alpha(e^-)$) and $G_1$ is a monomial of  ($p_s^\alpha(e^+)- p_t^\alpha(e^-)$).

The results in Sec. \ref{se:leading} can be used to reduce the computational complexity too. To use these results, one needs to apply the basic commutation relations \eqref{eq:commutators} to simplify the Hamiltonian operator such that the operators after the simplification are written in terms of $C \hat P$ with $C$ being some constant of order $t^0$ or higher, and $\hat P$ being some monomial of holonomies and fluxes.

In order to achieve so, one needs to permute $h_{e_x^+}$ and $\hat F_1$, as well as $h_{e_y^+}$ and $\hat F_2$, in Eq. \eqref{eq:cruxofsimplification} with applying Eq. \eqref{eq:changeO}.
 Take the permutation of $h_{e_x^+}$ and $\hat F_1$ as an example:
Implementing the results of Eq. \eqref{eq:changeO}, one substitutes $\hat O_1$ by $h_{e_x^+}$, and $\hat O_{i_k}$ by $\hat p_s^{\alpha}(e_x^+)$ and/or $\hat p_t^{\alpha}(e_x^+)$. One of many these substitutions inevitably generates some special terms in which the commutators only contain $\hat p_s^{\alpha}(e_x^+)$. The computation of these commutators with \eqref{eq:commutators} will lead to results that are proportional to $h_{e_x^+}$. After substituting these permuted results into Eq. \eqref{eq:cruxofsimplification}, this $h_{e_x^+}$ eventually cancels with $h_{e_x^+}^{-1}$.
(similar to Eq. \eqref{eq:ezpm}). One can apply the same mechanism to permute $h_{e_y^+}$ and $\hat F_2$.
 
 Let us collect these special terms coming from permuting $h_{e_x^+}$ and $\hat F_1$ as well as permuting $h_{e_y^+}$ and $\hat F_2$. Denote the partial sum of these special terms in $\hat H_L^{(k)}$ by ${}^{\rm alt}\hat{H}_L^{(\vec k)}$. Because of the cancellation between holonomies and their inverses, ${}^{\rm alt}\hat{H}_L^{(\vec k)}$ no longer depends on $h_{e_x^+}$ and $h_{e_y^+}$ . 
 
 It turns out that ${}^{\rm alt}\hat{H}_L^{(\vec k)}$ possesses more symmetries which will be discussed shortly below. These special terms can be equivalently selected by considering only the non-commutativity between $h_{e_x^+}$ and $p_s^\alpha(e_x^+)$ but ignoring the non-commutativity between $h_{e_x^+}$ and $p_t^\alpha(e_x^+)$. That is
 \begin{equation}\label{eq:halt0}
\text{the special terms of }h_{e_x^+}\hat F_1 h_{e_x^+}^{-1} =h_{s_x^+}\hat F_1 h_{s_x^+}^{-1}
 \end{equation}
 where $s_x^+$ is the segment within $e_x^+$ and does not contain the target $t(e_x^+)$.
 Because of the aforementioned cancellation between the holonomies and their inverses, it is remarkable that the length of the segment does not cause any ambiguity {and the operator in Eq. \eqref{eq:halt0} does not change graph ever segment of edges in the holonomy is chosen}.
 Concretely, Eq. \eqref{eq:halt0} results in
\begin{equation}
\begin{aligned}
&{}^{\rm alt}\hat{H}_L^{(\vec k)}(v;v_1,v_2,v_3,v_4;e_I,e_J,e_K)\\
=&\frac{-1}{2i\beta^7 a^{10} t^5}\epsilon^{IJK}\tr\left( [h_{s_I},[\hat Q_{v_1}^{2k_1},\hat{H}_E^{(k_2)}(v_2)]]h_{s_I}^{-1} [h_{s_J},[\hat Q_{v_3}^{2k_3},\hat{H}_E^{(k_4)}(v_4)]]h_{s_J}^{-1}[h_{s_K},\hat Q_v^{2k_5}]h_{s_K}^{-1}\right).
\end{aligned}
\end{equation}
It is interesting that the RHS could be understood as that from an alternative definition of the Lorentzian part, 
\begin{equation}
\begin{aligned}
{}^{\rm alt}\widehat{H_L[N]}=\frac{-1}{2i\beta^7 a^{10} t^5}\sum_v N(v)\sum_{s_I,s_J,s_K \text{ at } v}\varepsilon^{IJK}\tr\left( [h_{s_I},[\hat V,\hat H_E]]h_{s_I}^{-1} [h_{s_J},[\hat V,\hat H_E]]h_{s_J}^{-1}[h_{s_K},\hat V_v]h_{s_K}^{-1}\right).\label{Halt}
\end{aligned}
\end{equation}
in which all edges $e_I,e_J,e_K$ are replaced by their corresponding segments $s_I,s_J,s_K$ with $s_I\subset e_I$. Indeed, ${}^{\rm alt}\widehat{H_L[N]}$ is obtained by an alternative regularization/quantization of the Hamiltonian, i.e. via the following replacement
$$\{K,\dot{e}^aA_a(x)\}\to \frac{-1}{2\kappa(i\hbar\beta)^2}[h_{s_e},[\hat{V},\hat{H}_E]] h_{s_e}^{-1},$$
where the holonomy along the segment $s_e\subset e$ instead of the entire edge $e$ is used. Here, $\dot e^a$ denote the vector tangent to $e$.

Collect the terms in $\hat H_L^{(k)}$ other than the special terms discussed above, and denote their sum by ${}^{\rm extra}\hat{H}_L^{(\vec k)}$, namely
\begin{eqnarray}
&&{}^{\rm extr}\hat{H}_L^{(\vec k)}(v;v_1,v_2,v_3,v_4;e_I,e_J,e_K)\nonumber\\
&=&\hat{H}_L^{(\vec k)}(v;v_1,v_2,v_3,v_4;e_I,e_J,e_K)-{}^{\rm alt}\hat{H}_L^{(\vec k)}(v;v_1,v_2,v_3,v_4;e_I,e_J,e_K). 
\label{altextra}
\end{eqnarray}
The operators ${}^{\rm alt}\hat{H}_L^{(\vec k)}$ and ${}^{\rm extr}\hat{H}_L^{(\vec k)}$ are dealt with separately in our algorithm.   

For ${}^{\rm alt}\hat{H}_L^{(\vec k)}$, the simplification procedures discussed above result in 
\begin{equation}\label{eq:cruxsimplificationsolved}
\frac{C}{t^2}h_{s_x^+}\, F_1 h_{s_x^+}^{-1}h_{s_y^+} F_2 h_{s_y^+}^{-1} G_1,
\end{equation}
instead of Eq. \eqref{eq:cruxofsimplification}. 
Since $[h_{s_I},\hat p_t^\alpha(e_I)]=0$, we can simplify these terms with
\begin{equation}\label{eq:hphm1}
\begin{aligned}
h_{s^\pm}\prod_{i=1}^m (\sigma^+_i\hat p^{\alpha_i}_s(e^+)+\sigma^-_i \hat p^{\alpha_i}_t(e^-))h_{s^\pm}^{-1}=\sum_{\mathcal I}(i t)^{|\mathcal I|}\prod_{i\notin \mathcal I}(\sigma_i^+\hat p^{\alpha_i}_s(e^+)+\sigma_i^- \hat p^{\alpha_i}_t(e^-))\prod_{j\in\mathcal I}\sigma_j^\pm \tau^{\alpha_j}
\end{aligned}
\end{equation}
where $\mathcal I$ is a subsets of $\{1,2,\cdots,m\}$ with its length denoted by $|\mathcal I|$, $s^+$ and $s^-$ are segments of $e^+$ and $(e^-)^{-1}$ respectively with $e^\pm$ oriented such that $e^+$ and $(e^-)^{-1}$ are both outgoing,
$\sigma_i^+=1$ and  $\sigma_i^-=\pm1$. The summation over $e_I$, $e_J$ and $e_K$ in Eq. \eqref{eq:HLnv} motivates us to compute the following 
\begin{equation}
h_{s^+}\prod_{i=1}^m (\sigma^+_i\hat p^{\alpha_i}_s(e^+)+\sigma^-_i \hat p^{\alpha_i}_t(e^-))h_{s^+}^{-1}-h_{s^-}\prod_{i=1}^m (\sigma^+_i\hat p^{\alpha_i}_s(e^+)+\sigma^-_i \hat p^{\alpha_i}_t(e^-))h_{s^-}^{-1}.
\end{equation} 
Then, one can apply the aforementioned replacement
\begin{equation}\label{eq:hphm1replacement}
h_{s^+}\prod_{i=1}^m (\sigma^+_i\hat p^{\alpha_i}_s(e^+)+\sigma^-_i \hat p^{\alpha_i}_t(e^-))h_{s^+}^{-1}\to \sum_{\mathcal J}(i t)^{|\mathcal J|}\prod_{i\notin \mathcal J}(\sigma_i^+\hat p^{\alpha_i}_s(e^+)+\sigma_i^- \hat p^{\alpha_i}_t(e^-))\prod_{j\in\mathcal J} \tau^{\alpha_j}
\end{equation}
with $\mathcal J\subset\{1,2,\cdots,m\}$ such that 
\begin{equation}
\prod_{j\in \mathcal J}\sigma_j^-=-1.
\end{equation}

Substitute Eq. \eqref{eq:hphm1replacement} into \eqref{eq:cruxsimplificationsolved}, one cancels the prefactor $1/t^2$, simplifying \eqref{eq:cruxsimplificationsolved} to be in the form of $C\hat P$ with some constant $C$ of order $t^0$ or higher. Hence the results in Sec. \ref{se:leading} can be applied. 

Further, ${}^{\rm alt}\hat{H}_L^{(\vec k)}$ brings the following symmetries. Consider a $\pi$-rotation which transforms either $e_x^+$ to $e_x^-$ or $e_y^+$ to $e_y^-$. Denote it by $\mathfrak s$. Moreover, to indicate the dependence of $F_i$ on vertices and edges, we will rewrite $F_i$ in Eq. \eqref{eq:cruxsimplificationsolved} as $F_i(v,e)$. 
Then, {since $\mathfrak{s}(s_k^+)$ with $k=x,y$ is either $s_k^+$ or $s_k^-$ by definition of $\mathfrak{s}$, }
 Eq. \eqref{eq:hphm1} tells us 
\begin{equation}
\begin{aligned}
&h_{\mathfrak{s}(s_x^+)}F_1(\mathfrak s(v),\mathfrak s(e))h_{\mathfrak{s}(s_x^+)}^{-1}h_{\mathfrak{s}(s_y^+)}F_2(\mathfrak s(v),\mathfrak s(e))h_{\mathfrak{s}(s_y^+)}^{-1}G_1(\mathfrak s(v),\mathfrak s(e))\\
=&h_{s_x^+}F_1(\mathfrak s(v),\mathfrak s(e))h_{s_x^+}^{-1}h_{s_y^+}F_2(\mathfrak s(v),\mathfrak s(e))h_{s_y^+}^{-1}G_1(\mathfrak s(v),\mathfrak s(e)).
\end{aligned}
\end{equation}
Furthermore, with recalling Eq. \eqref{eq:DifftoOp}, we obtain the following equation
\begin{equation}\label{eq:symmetryHalt}
\begin{aligned}
&\langle h_{s_x^+}F_1(v,e)h_{s_x^+}^{-1}h_{s_y^+}F_2(v,e)h_{s_y^+}^{-1}G_1(v,e)\rangle\\
=&\langle h_{s_x^+}F_1(\mathfrak s(v),\mathfrak s(e))h_{s_x^+}^{-1}h_{s_y^+}F_2(\mathfrak s(v),\mathfrak s(e))h_{s_y^+}^{-1}G_1(\mathfrak s(v),\mathfrak s(e))\rangle,
\end{aligned}
\end{equation}
 which reduces the number of contributing vertices and edges in the computation of ${}^{\rm alt}\hat{H}_L^{(\vec k)}$.

For the operator ${}^{\rm extr}\hat{H}_L^{(\vec k)}$, Eq. \eqref{eq:hphm1} can no longer be applied. 
Therefore, the symmetry implied by \eqref{eq:symmetryHalt} is not manifested. To reduce the complexity of the computation, the following strategy is proposed. Consider a rotation, denoted by $\mathfrak r$, about the axis $(1/\sqrt{2},1/\sqrt{2},0)$ for $\pi$ radians which exchanges the $x$- and $y$-axes, and flips the $z$-axis. We obtain
\begin{equation}\label{eq:trotatedHL}
\begin{aligned}
&\langle \hat{\tilde H}_L^{(\vec k)}(v;v_1,v_2,v_3,v_4;e_x^+,e_y^+,e_z^+)\rangle=\langle\hat{\tilde H}_L^{(\vec k)}(\mathfrak r(v);\mathfrak r(v_1),\mathfrak r(v_2),\mathfrak r(v_3),\mathfrak r(v_4);\mathfrak{r}(e_x^+),\mathfrak{r}(e_y^+),\mathfrak{r}(e_z^+))\rangle\\
=&\langle\hat{\tilde H}_L^{(\vec k)}(v;\mathfrak r(v_1),\mathfrak r(v_2),\mathfrak r(v_3),\mathfrak r(v_4);e_y^+,e_x^+,e_z^-)\rangle=\langle\hat{\tilde H}_L^{(\vec k)}(v;\mathfrak r(v_1),\mathfrak r(v_2),\mathfrak r(v_3),\mathfrak r(v_4);e_y^+,e_x^+,e_z^+)\rangle
\end{aligned}
\end{equation}
in which we use the definition of $\hat{\tilde H}_L^{(\vec k)}$ in  \eqref{eq:HLnvpp}, assume $\mathfrak r(v)=v$ without loss of generality. {Moreover, the first equality of the last equation is obtained by using Eq. \eqref{eq:DifftoOp}, the second one, by the definition of $\mathfrak r$ and, the last one, by Eq. \eqref{eq:hphm1}.
}
Then  consider the operator 
\begin{equation}
\begin{aligned}
&\hat F(v;v_1,v_2,v_3,v_4;e_x^+,e_y^+,e_z^+)\\
:=&\hat{H}_L^{(\vec k)}(v;v_1,v_2,v_3,v_4;e_x^+,e_y^+,e_z^+)+\hat{H}_L^{(\vec k)}(v;\mathfrak r(v_3),\mathfrak r(v_4),\mathfrak r(v_1),\mathfrak r(v_2);e_x^+,e_y^+,e_z^+)
\end{aligned}
\end{equation}
 According to \eqref{eq:trotatedHL}, $\hat F$ is 
 \begin{equation}
\begin{aligned}
&\hat F(v;v_1,v_2,v_3,v_4;e_x^+,e_y^+,e_z^+)\\
=&-\tr( [h_{e_y^+},[\hat Q_{\mathfrak r(v_1)}^{2k_1},\hat{H}_E^{(k_2)}(\mathfrak r(v_2))]]h_{e_y^+}^{-1} [h_{e_x^+},[\hat Q_{\mathfrak r(v_3)}^{2k_3},\hat{H}_E^{(k_4)}(\mathfrak r(v_4))]]h_{e_x^+}^{-1}[h_{e_z^+},\hat Q_v^{2k_5}]h_{e_z^+}^{-1})\\
&+\tr( [h_{e_x^+},[\hat Q_{\mathfrak r(v_3)}^{2k_1},\hat{H}_E^{(k_2)}(\mathfrak r(v_4))]]h_{e_x^+}^{-1} [h_{e_y^+},[\hat Q_{\mathfrak r(v_1)}^{2k_3},\hat{H}_E^{(k_4)}(\mathfrak r(v_2))]]h_{e_y^+}^{-1}[h_{e_z^+},\hat Q_v^{2k_5}]h_{e_z^+}^{-1})
\end{aligned}
\end{equation}
 up to an overall factor.
For clarity, we will denote 
\begin{equation}
\begin{aligned}
\hat X\equiv &[h_{e_x^+},[\hat Q_{\mathfrak{r}(v_3)}^{2k_3},\hat{H}_E^{(k_4)}(\mathfrak r(v_4))]]h_{e_x^+}^{-1}\\
\hat Y\equiv  &[h_{e_y^+},[\hat Q_{\mathfrak r(v_1)}^{2k_1},\hat{H}_E^{(k_2)}(\mathfrak r(v_2))]]h_{e_y^+}^{-1}\\
\hat Z\equiv &[h_{e_z^+},\hat Q_v^{2k_5}]h_{e_z^+}^{-1}.
\end{aligned}
\end{equation}
Because of the holonomies therein, they are all operator-valued matrices whose entries, thus, will be denoted as $\hat X_{\tilde a\tilde b}$, $\hat Y_{\tilde a\tilde b}$ and $\hat Z_{\tilde a\tilde b}$ respectively. Then, one has
\begin{equation}\label{FXYZ}
\hat F=\hat X_{\tilde a\tilde b}\hat Y_{\tilde b\tilde c}\hat Z_{\tilde c\tilde a}-\hat Y_{\tilde a\tilde b}\hat X_{\tilde b\tilde c}\hat Z_{\tilde c\tilde a}=\hat Y_{\tilde b\tilde c}\hat X_{\tilde a\tilde b}\hat Z_{\tilde c\tilde a}-\hat Y_{\tilde a\tilde b}\hat X_{\tilde b\tilde c}\hat Z_{\tilde c\tilde a}+[\hat X_{\tilde a\tilde b},\hat Y_{\tilde b\tilde c}]\hat Z_{\tilde c\tilde a}.
\end{equation}
According to Eq. \eqref{eq:basicformula}, we first compute the expectation values of the operators with respect to coherent states labeled by $e^{iz_e\tau_3}$, and then gauge transform the results correspondingly. Hence, applying this procedure to the the first two terms of Eq. \eqref{FXYZ}, one finally simplifies the subtraction of their expectation values as 
\begin{equation}\label{eq:FXYZp}
\begin{aligned}
\langle \hat Y_{\tilde b\tilde c}\hat X_{\tilde a\tilde b}\hat Z_{\tilde c\tilde a}-\hat Y_{\tilde a\tilde b}\hat X_{\tilde b\tilde c}\hat Z_{\tilde c\tilde a}\rangle_{\{g_e\}}=(\mathcal C_{abcdef}-\mathcal C'_{abcdef})\langle \hat{Y}_{ab}\hat{X}_{cd}\hat{Z}_{ef}\rangle_{\{z_e\}}
\end{aligned}
\end{equation}
where $\mathcal C_{abcdef}$ and $\mathcal C'_{abcdef}$ are two sets of constant coefficients produced by the gauge-transformation. 
Thanks to the the last equation, one has
\begin{equation}\label{eq:FXYZpp}
\langle\hat F\rangle_{\{g_e\}}=(\mathcal C_{abcdef}-\mathcal C'_{abcdef})\langle \hat{Y}_{ab}\hat{X}_{cd}\hat{Z}_{ef}\rangle_{\{z_e\}}+\langle[\hat X_{\tilde a\tilde b},\hat Y_{\tilde b\tilde c}]\hat Z_{\tilde c\tilde a}\rangle_{\{g_e\}}. 
\end{equation}
Regardless of the factors $\mathcal C_{abcdef}-\mathcal C'_{abcdef}$ which is easy to compute, according to Eq. \eqref{eq:FXYZp}, the expectation values of $\langle \hat{Y}_{ab}\hat{X}_{cd}\hat{Z}_{ef}\rangle_{\{z_e\}}$ and $\langle[\hat X_{\tilde a\tilde b},\hat Y_{\tilde b\tilde c}]\hat Z_{\tilde c\tilde a}\rangle_{\{g_e\}}$ are needed to compute to get $\langle \hat F\rangle_{\{g_e\}}$. While for the last term, $[\hat X_{\tilde a\tilde b},\hat Y_{\tilde b\tilde c}]\hat Z_{\tilde c\tilde a}$ itself is of $O(t)$. Thus we only need to compute the leading-order term of $\langle [\hat X_{\tilde a\tilde b},\hat Y_{\tilde b\tilde c}]\hat Z_{\tilde c\tilde a}\rangle$ by applying Theorem \ref{thm:leadingordermultiplyPH}. By definition of $\hat F$, without applying Eq. \eqref{eq:FXYZpp}, we need to compute $\langle \hat Y_{\tilde a\tilde b}\hat X_{\tilde b\tilde c}\hat Z_{\tilde c\tilde a}\rangle$ and $\langle \hat Y_{\tilde a\tilde b}\hat X_{\tilde b\tilde c}\hat Z_{\tilde c\tilde a}\rangle$  to get $\langle\hat F\rangle$. Comparing the computed terms before and after  applying Eq. \eqref{eq:FXYZpp}, it is concluded that the computational complexity is reduced due to Eq. \eqref{eq:FXYZpp}, because it is much easier to compute $\langle [\hat X_{\tilde a\tilde b},\hat Y_{\tilde b\tilde c}]\hat Z_{\tilde c\tilde a}\rangle$ by applying Theorem \ref{thm:leadingordermultiplyPH}.

To explicitly compute $\Delta\mathcal C\coloneqq \mathcal C_{abcdef}-\mathcal C'_{abcdef}$, let us use $O_i$ to denote expectation values of polynomials of fluxes and holonomies with respect to coherent states labeled by $e^{iz_e\tau_3}$. 
Then $\hat{H}_L^{(\vec k)}$ and the corresponding results of $\Delta\mathcal C$ take the following forms, which are discussed case by case.
\begin{itemize}
 \item[(1)] If $\hat{H}_L^{(\vec k)}$ is expressed in terms of the form
 \begin{itemize}
 \item[(1.1)] $(n_x O_1\tau^{\alpha}n_x^{-1})_{\tilde a\tilde b}(n_y O_2\tau^\beta n_y^{-1})_{\tilde b \tilde c}(n_z O_3 \tau^{\gamma}n_z^{-1})_{\tilde c\tilde a}$;
 \item[(1.2)] $(n_xO_1\tau^{\alpha}n_x^{-1})_{\tilde a\tilde b}(n_yO_2\tau^{\beta_1}\tau^{\beta_2} n_y^{-1})_{\tilde b \tilde c}(n_zO_3 \tau^{\gamma} n_z^{-1})_{\tilde c\tilde a}$;
 \item[(1.3)] $(n_xO_1\tau^{\alpha_1}\tau^{\alpha_2}n_x^{-1})_{\tilde a\tilde b}(n_yO_2\tau^\beta n_y^{-1})_{\tilde b \tilde c}(n_zO_3\tau^{\gamma} n_z^{-1})_{\tilde c\tilde a}$,
 \end{itemize}
then $\Delta\mathcal C$ is defined as the following,
 \begin{equation*}
 \begin{aligned}
\Delta\mathcal C=&\tr\left[(n_x^{-1} n_y)\left(\prod_{\beta=1}^{m_y}\tau^{\beta}\right)(n_y^{-1}n_z)\tau^\gamma (n_z^{-1}n_x)\left(\prod_{\alpha=1}^{m_x}\tau^\alpha\right)\right]\\
&-\tr\left[(n_y^{-1} n_x)\left(\prod_{\alpha=1}^{m_x}\tau^{\alpha}\right)(n_x^{-1}n_z)\tau^\gamma (n_z^{-1}n_y)\left(\prod_{\beta=1}^{m_y}\tau^\beta\right)\right]     
 \end{aligned}
 \end{equation*}
 where $(m_y,m_x)=(1,1)$, $(2,1)$ and $(1,2)$ for the cases (1.1), (1.2) and (1.3) respectively. 
 \item[(2)] If $\hat{H}_L^{(\vec k)}$ is expressed in terms of the form
 \begin{itemize}
    \item[(2.1)] $(n_xO_1h_{e_x^+}\tau^{\alpha}h_{e_x^+}^{-1}n_x^{-1})_{\tilde a\tilde b}(n_yO_2h_{e_y^+}\tau^{\beta}h_{e_y^+}^{-1}n_y^{-1})_{\tilde b \tilde c}(n_zO_3\tau^\gamma n_z^{-1})_{\tilde c\tilde a}$;
    \item[(2.2)] $(n_xO_1h_{e_x^+}\tau^{\alpha}h_{e_x^+}^{-1}n_x^{-1})_{\tilde a\tilde b}(n_yO_2h_{e_y^+}\tau^{\beta_1}\tau^{\beta_2} h_{e_y^+}^{-1}n_y^{-1})_{\tilde b \tilde c}(n_zO_3 \tau^\gamma n_z^{-1})_{\tilde c\tilde a}$;
  \item[(2.3)] $(n_xO_1h_{e_x^+}\tau^{\alpha_1}\tau^{\alpha_2}h_{e_x^+}^{-1}n_x^{-1})_{\tilde a\tilde b}(n_yO_2h_{e_y^+}\tau^{\beta}h_{e_y^+}^{-1}n_y^{-1})_{\tilde b \tilde c}(n_zO_3 \tau^\gamma n_z^{-1})_{\tilde c\tilde a}$,
 \end{itemize}
 then $\Delta\mathcal C$ is defined as the following,
 \begin{equation*}
 \begin{aligned}
 \Delta\mathcal C=\left(\prod_{\beta=1}^{m_y}\tau^{\beta}\right)_{b c}\,\left(\prod_{\alpha=1}^{m_x}\tau^{\alpha}\right)_{d a}\Big\{(n_x^{-1}n_y)_{ab}[(n_y^{-1}n_z)\tau^\gamma (n_z^{-1}n_x)]_{cd}-
(n_y^{-1}n_x)_{cd}[(n_x^{-1}n_z)\tau^\gamma (n_z^{-1}n_y)_{ab}\Big\}
 \end{aligned}
 \end{equation*}
 \item[(3)] If $\hat{H}_L^{(\vec k)}$ is expressed in terms of the form
  \begin{itemize}
  \item[(3.1)]  $(n_xO_1h_{e_x^+}\tau^{\alpha}h_{e_x^+}^{-1}n_x^{-1})_{\tilde a\tilde b}(n_yO_2\tau^{\beta_1}h_{e_y^+}\tau^{\beta_2} h_{e_y^+}^{-1}n_y^{-1})_{\tilde b \tilde c}(n_zO_3 \tau^\gamma n_z^{-1})_{\tilde c\tilde a}$,
  \end{itemize}
 then $\Delta\mathcal C$ is defined as the following,
  \begin{equation*}
 \begin{aligned}
 \Delta\mathcal C=(\tau^{\beta_2})_{bc}(\tau^{\alpha})_{da}\Big\{[n_x^{-1} n_y\tau^{\beta_1}]_{ab} [(n_y^{-1}n_z) \tau^\gamma (n_z^{-1}n_x)]_{cd}- [(n_x^{-1}n_z)\tau^\gamma n_z^{-1}n_y\tau^{\beta_1}]_{ab}[(n_y^{-1}n_x) ]_{cd}\Big\}
 \end{aligned}
 \end{equation*}
  \item[(4)] If $\hat{H}_L^{(\vec k)}$ is expressed in terms of the form
 \begin{itemize}
 \item[(4.1)] $(n_xO_1\tau^{\alpha_1}h_{e_x^+}\tau^{\alpha_2}h_{e_x^+}^{-1}n_x^{-1})_{\tilde a\tilde b}(n_yO_2h_{e_y^+}\tau^{\beta}h_{e_y^+}^{-1}n_y^{-1})_{\tilde b \tilde c}(n_zO_3 \tau^\gamma n_z^{-1})_{\tilde c\tilde a}$,
 \end{itemize}
 then $\Delta\mathcal C$ is defined as the following,
 \begin{equation*}
 \begin{aligned}
 \Delta\mathcal C=(\tau^{\alpha_2})_{da} (\tau^\beta)_{bc}\Big\{[(n_x^{-1}n_y)_{ab} [(n_y^{-1}n_z\tau^\gamma n_z^{-1}n_x\tau^{\alpha_1}]_{cd}-[(n_x^{-1})n_z)\tau^\gamma (n_z^{-1}n_y)]_{a b}[(n_y^{-1}n_x)\tau^{\alpha_1}]_{cd}\Big\}
 \end{aligned}
 \end{equation*}
 \item[(5)] If $\hat{H}_L^{(\vec k)}$ is expressed in terms of the form
 \begin{itemize}
 \item[(5.1)] $(n_xO_1h_{e_x^+}\tau^{\alpha}h_{e_x^+}^{-1}n_x^{-1})_{\tilde a\tilde b}(n_yO_2\tau^{\beta}n_y^{-1})_{\tilde b \tilde c}(n_zO_3 \tau^\gamma n_z^{-1})_{\tilde c\tilde a}$;
   \item[(5.2)] $(n_xO_1h_{e_x^+}\tau^{\alpha}h_{e_x^+}^{-1}n_x^{-1})_{\tilde a\tilde b}(n_yO_2\tau^{\beta_1}\tau^{\beta_2}n_y^{-1})_{\tilde b \tilde c}(n_zO_3\tau^\gamma n_z^{-1})_{\tilde c\tilde a}$;
   \item[(5.3)] $(n_xO_1h_{e_x^+}\tau^{\alpha_1}\tau^{\alpha_2}h_{e_x^+}^{-1}n_x^{-1})_{\tilde a\tilde b}(n_yO_2\tau^{\beta}n_y^{-1})_{\tilde b \tilde c}(n_zO_3\tau^\gamma n_z^{-1})_{\tilde c\tilde a}$,
    \end{itemize}
     then $\Delta\mathcal C$ is defined as the following,
     \begin{equation*}
     \begin{aligned}
     \Delta\mathcal C=\left(\prod_{\alpha=1}^{m_x}\tau^{\alpha}\right)_{ba}\Big\{[(n_x^{-1}n_y)\big(\prod_{\beta=1}^{m_y}\tau^{\beta}\big)(n_y^{-1}n_z)\tau^\gamma (n_z^{-1}n_x)]_{ab}-[(n_x^{-1}n_z)\tau^\gamma (n_z^{-1}n_y)(\prod_{\beta=1}^{m_y}\tau^{\beta})(n_y^{-1}n_x )]_{ab}\Big\}
     \end{aligned}
     \end{equation*}
     \item[(6)] If $\hat{H}_L^{(\vec k)}$ is expressed in terms of the form
\begin{itemize}
\item[(6.1)] $(n_xO_1\tau^{\alpha}n_x^{-1})_{\tilde a\tilde b}(n_yO_2h_{e_y^+}\tau^\beta h_{e_y^+}^{-1}n_y^{-1})_{\tilde b \tilde c}(n_zO_3\tau^\gamma n_z^{-1})_{\tilde c\tilde a}$;
  \item[(6.2)] $(n_xO_1\tau^{\alpha}n_x^{-1})_{\tilde a\tilde b}(n_yO_2h_{e_y^+}\tau^{\beta_1}\tau^{\beta_2} h_{e_y^+}^{-1}n_y^{-1})_{\tilde b \tilde c}(n_zO_3\tau^\gamma n_z^{-1})_{\tilde c\tilde a}$;
 \item[(6.3)] $(n_xO_1\tau^{\alpha_1}\tau^{\alpha_2}n_x^{-1})_{\tilde a\tilde b}(n_yO_2h_{e_y^+}\tau^\beta h_{e_y^+}^{-1} n_y^{-1})_{\tilde b \tilde c}(n_zO_3\tau^\gamma n_z^{-1})_{\tilde c\tilde a}$,
\end{itemize}
then $\Delta\mathcal C$ is defined as the following,
\begin{equation*}
\begin{aligned}
\Delta\mathcal C= \left(\prod_{\beta=1}^{m_y}\tau^\beta \right)_{ba}\Big\{[(n_y^{-1} n_z)\tau^{\gamma}(n_z^{-1}n_x)(\prod_{\alpha=1}^{m_x}\tau^{\alpha})(n_x^{-1}n_y)]_{ab}-[(n_y^{-1}n_x)(\prod_{\alpha=1}^{m_x}\tau^{\alpha})(n_x^{-1}n_z)\tau^\gamma (n_z^{-1}n_y)]_{ab}\Big\}
\end{aligned}
\end{equation*}
\item[(7)] If $\hat{H}_L^{(\vec k)}$ is expressed in terms of the form
\begin{itemize}
\item[(7.1)] $(n_xO_1\tau^{\alpha}n_x^{-1})_{\tilde a\tilde b}(n_yO_2\tau^{\beta_1}h_{e_y^+}\tau^{\beta_2} h_{e_y^+}^{-1}n_y^{-1})_{\tilde b \tilde c}(n_zO_3\tau^\gamma n_z^{-1})_{\tilde c\tilde a}$,
\end{itemize}
then $\Delta\mathcal C$ is defined as the following,
\begin{equation*}
\begin{aligned}
\Delta\mathcal C=(\tau^{\beta_2})_{ba}\left\{[(n_y^{-1}n_z)\tau^\gamma (n_z^{-1}n_x) \tau^{\alpha}(n_x^{-1}n_y) \tau^{\beta_1}]_{ab}-[(n_y^{-1}n_x)\tau^{\alpha}(n_x^{-1}n_z)\tau^\gamma (n_z^{-1}n_y)\tau^{\beta_1}]_{ab}\right\}
\end{aligned}
\end{equation*}
 \item[(8)] If $\hat{H}_L^{(\vec k)}$ is expressed in terms of the form
  \begin{itemize}
  \item[(8.1)] $(n_xO_1\tau^{\alpha_1}h_{e_x^+}\tau^{\alpha_2}h_{e_x^+}^{-1}n_x^{-1})_{\tilde a\tilde b}(n_yO_2\tau^{\beta}n_y^{-1})_{\tilde b \tilde c}(n_zO_3\tau^\gamma n_z^{-1})_{\tilde c\tilde a}$
  \end{itemize}
  the subtraction of the coefficients are of the form
  \begin{equation*}
  \begin{aligned}
  \Delta\mathcal C=(\tau^{\alpha_2})_{ba}\Big\{[(n_x^{-1}n_y)\tau^{\beta}(n_y^{-1}n_z)\tau^\gamma (n_z^{-1}n_x)\tau^{\alpha_1}]_{ab}-[(n_x^{-1}n_z)\tau^\gamma (n_z^{-1}n_y)\tau^{\beta}n_y^{-1}n_x\tau^{\alpha_1}]_{ab}\Big\}.
  \end{aligned}
  \end{equation*}
 \end{itemize}

In summary, let $N_1$ be the original number of terms in ${}^{\rm alt}\hat H_L^{(k)}$. The symmetries implied by Eqs. \eqref{eq:HEezpezm} and \eqref{eq:extrasymmtryofHE} reduces this number to $N_1/4^2$, where the power 2 is from the fact that ${}^{\rm alt}\hat H_L^{(k)}$ consists of two $\hat H_E$. Then Eq. \eqref{eq:HLnvpp} reduces this number further to $N_1/(4^2\times 48)$. Finally the symmetry \eqref{eq:symmetryHalt} reduces this number to $$\frac{N_1}{4^2\times 48\times 4}= \frac{N_1}{3072}.$$
Let $N_2$ be the original number of terms in ${}^{\rm extr}\hat H_L^{(k)}$. The symmetries implied by Eqs. \eqref{eq:HEezpezm} and \eqref{eq:extrasymmtryofHE} are used at first to reduce this number to $N_2/4^2$. Then Eq. \eqref{eq:HLnvpp} reduces it to $N_2/(4^2\times 48)$. Finally, Eq. \eqref{eq:FXYZpp} reduce this number further to be about\footnote{The word ``about" is because there exists cases with $\hat X=\hat Y$ in Eq. \eqref{FXYZ}.}
$$\frac{N_2}{4^2\times 48\times 2}= \frac{N_2}{1536}.$$

\subsection{Exhibit an explicit computation}  \label{Introducing the algorithm}

In order to demonstrate the idea of our algorithm, some simple examples are used in this section. All of the cases illustrated by these examples finally occur in our computation. {Particularly, some functions like $P(v,i,\mathcal I^{(i,m)})$, $\wdtau(\mathcal I^{(i,k)})$ and $\wdtaus(\mathcal I^{(i,k)})$ indeed exist in our codes \cite{github} as the same manner. 
}

Roughly speaking, the computation is divided into two steps. The first is to simplify the operator with applying the commutation relations \eqref{eq:commutators} and the second is to compute the expectation values of the simplified operator.

In the first step, the non-commutative multiplications between operator make a non-trivial simplification. Because of the operator $\hat Q$, we actually need to deal with 
$$\hat P^\alpha(v,i)=\hat p^{\alpha}_s(e^+_v)-\hat p^{\alpha}_t(e^-_v)$$
where $e^\pm_v$ are the edges along the $i$th direction satisfying $s(e^+_v)=v=t(e^-_v)$. In the computation, we mainly need to deal with the commutators between $\hat P^\alpha(v,i)$ and holonomy. Thus let us consider the example
\begin{equation}\label{eq:example1}
\begin{aligned}
&[\prod_{j=1}^m\hat P^{\alpha_j}(v,i),h_{e_v^+}]=it\sum_{k=1}^m\tau^{\alpha_k}h_{e_v^+}\prod_{j\neq k}\hat P^{\alpha_j}(v,i)+(it)^2\sum_{k<l}\tau^{\alpha_l}\tau^{\alpha_k}h_{e_v^+}\prod_{j\neq k,l}\hat P^{\alpha_j}(v,i)+O(t^3).
\end{aligned}
\end{equation}
Because these $\alpha_j$ appear in the computed operator as dummy indices, their specific values do not matter in the operator-simplification procedure. In other words, $\prod_{j=1}^m\hat P^{\alpha_j}(v,i)$ is treated like a tensor for which only the type does matter. Therefore, based on our algorithm, we define a function
\begin{equation}
P(v,i,\mathcal I^{(i,m)}):=\prod_{\alpha\in\mathcal I^{(i,m)}}\hat P^{\alpha}(v,i),
\end{equation}
where $\mathcal I^{(i,m)}$ denote an index set of length $m$. In the supperscript of $\mathcal I^{(i,m)}$, $i$ is used again for convenience so that the multiplication of $(\hat P^\alpha(v,1)\hat P^\beta(v,2)\hat P^\gamma(v,3))^m$ occuring in $\hat Q_v^m$ can be denoted as
$$
(\hat P^\alpha(v,1)\hat P^\beta(v,2)\hat P^\gamma(v,3))^m=P(v,1,\mathcal I^{(1,m)})P(v,2,\mathcal I^{(2,m)})P(v,3,\mathcal I^{(3,m)}).
$$
Similarly, let us define
\begin{equation}
\wdtau(\mathcal I^{(i,k)}):=\tau^{\alpha_k}\tau^{\alpha_{k-1}}\cdots\tau^{\alpha_1}
\end{equation}
where $\mathcal I^{(i,k)}=\{\alpha_1,\cdots,\alpha_k\}$. 
With these notions, Eq. \eqref{eq:example1} can be written as
\begin{equation}\label{eq:example10}
\begin{aligned}
[P(v,i,\mathcal {I}^{(i,m)}),h_{e_v^+}]=&(it)\sum_{k=1}^m \wdtau(\tilde{\mathcal I}_k^{(i,1)})h_{e_v^+}P(v,i,\mathcal I_k^{(i,m-1)})\\
&+(it)^2\sum_{k<l}\wdtau(\tilde{\mathcal I}_{k,l}^{(i,2)})h_{e_v^+}P(v,i,\mathcal I_{k,l}^{(i,m-2)})+O(t^3)
\end{aligned}
\end{equation}
where $\mathcal I_{k_1,k_2,\cdots,k_q}^{(i,m)}$ is the index set of $\mathcal I^{(i,m)}$ without the $k_1$th, $k_2$th, $\cdots$,$k_q$th indices 
and $\tilde{\mathcal I}_{k_1,k_2,\cdots,k_q}^{(i,m)}$ is the complement of $\mathcal I_{k_1,k_2,\cdots,k_q}^{(i,m)}$.

Further, let us implement an non-operator factor $\mathcal F(\mathcal I^{(i,m)})$ contracted with $P(v,i,\mathcal {I}^{(i,m)})$ to consider $\sum_{\mathcal I^{(i,m)}}\mathcal F(\mathcal I^{(i,m)})[P(v,i,\mathcal {I}^{(i,m)}),h_{e_v^+}]$ which is finally simplified as
\begin{equation}\label{eq:example1p}
\begin{aligned}
&\sum_{\mathcal I^{(i,m)}}\mathcal F(\mathcal I^{(i,m)})[P(v,i,\mathcal {I}^{(i,m)}),h_{e_v^+}]\\
=&(it)\sum_{\mathcal I^{(i,m-1)},\mathcal I^{(i,1)}}\wdtau(\mathcal I^{(i,1)})h_{e_v^+}P(v,i,\mathcal I^{(i,m-1)})\sum_{k=1}^m\mathcal F(\{\mathcal I^{(i,m-1)},\mathcal I^{(i,1)}\}_{(k)})\\
&+(it)^2\sum_{\mathcal I^{(i,m-2)},\mathcal I^{(i,2)}}\wdtau(\mathcal I^{(i,2)})h_{e_v^+}P(v,i,\mathcal I^{(i,m-2)})\sum_{k<l}\mathcal F(\{\mathcal I^{(i,m-2)},\mathcal I^{(i,2)}\}_{(k,l)})+O(t^3).
\end{aligned}
\end{equation}
Here $\{\mathcal I^{(i,m-n)},\mathcal I^{(i,n)}\}_{(k_1,k_2,\cdots, k_n)}$ is a joint set obtained by merging
$\mathcal I^{(i,m-n)},\mathcal I^{(i,n)}$ in such a way that the elements in $\mathcal I^{(i,n)}$ (respecting their original order) are distributed in $k_1$th,$\cdots$,$k_n$th position in $\{\mathcal I^{(i,m-n)},\mathcal I^{(i,n)}\}_{(k_1,k_2,\cdots, k_n)}$. 
%%%%%
Instead of deleting $n$ indices in all possible $\mathcal I^{(i,m)}$, one replace this deleting procedure by an inserting procedure, namely merging all possible sets $\mathcal I^{(i,m-n)}$ with length $m-n$ and $\mathcal I^{(i,n)}$ with length $n$ in the above way for all values of $k_1<k_2<\cdots<k_n$. Therefore, Eq. \eqref{eq:example10} can be simplified further as
\begin{equation}\label{eq:example11}
\begin{aligned}
&[h_{e_v^+},P(v,i,\mathcal {I}^{(i,m)})]\to \\
&it\wdtau(\mathcal I^{(i,1)})h_{e_v^+}P(v,i,\mathcal I^{(i,m-1)})+(it)^2\wdtau(\mathcal I^{(i,2)})h_{e_v^+}P(v,i,\mathcal I^{(i,m-2)})+O(t^3).
\end{aligned}
\end{equation}
Then, for each evaluation of the index sets $\mathcal{I}^{(i,m-n)}$ and $\mathcal{I}^{(i,n)}$, we joint them in the aforementioned way to consider $\{\mathcal I^{(i,m-n)},\mathcal I^{(i,n)}\}_{(k_1,k_2,\cdots, k_n)}$ for all values of $k_1<k_2<\cdots<k_n$.

Similarly for the commutator $[h_{(e_v^-)^{-1}},P(v,i,\mathcal {I}^{(i,m)})]$, we have
\begin{equation}\label{eq:example12}
\begin{aligned}
&[h_{(e_v^-)^{-1}},P(v,i,\mathcal {I}^{(i,m)})]\to \\
&-it\wdtau(\mathcal I^{(i,1)})h_{e_v^+}P(v,i,\mathcal I^{(i,m-1)})+(-it)^2\wdtau(\mathcal I^{(i,2)})h_{e_v^+}P(v,i,\mathcal I^{(i,m-2)})+O(t^3).
\end{aligned}
\end{equation}
For $[h_{e_v^-},P(v,i,\mathcal {I}^{(i,m)})]$ and $[h_{(e_v^+)^{-1}},P(v,i,\mathcal {I}^{(i,m)})]$, we need to define 
\begin{equation}
\wdtaus(\mathcal I^{(i,k)}):=\tau^{\alpha_1}\tau^{\alpha_{1}}\cdots\tau^{\alpha_k}
\end{equation}
where $\mathcal I^{(i,k)}=\{\alpha_1,\cdots,\alpha_k\}$. 
\begin{equation}\label{eq:example134}
\begin{aligned}
&[h_{e_v^-},P(v,i,\mathcal {I}^{(i,m)})]\to \\
&it\wdtaus(\mathcal I^{(i,1)})h_{e_v^+}P(v,i,\mathcal I^{(i,m-1)})+(it)^2\wdtaus(\mathcal I^{(i,2)})h_{e_v^+}P(v,i,\mathcal I^{(i,m-2)})+O(t^3),\\
&[h_{(e_v^+)^{-1}},P(v,i,\mathcal {I}^{(i,m)})]\to \\
&-it\wdtaus(\mathcal I^{(i,1)})h_{e_v^+}P(v,i,\mathcal I^{(i,m-1)})+(-it)^2\wdtaus(\mathcal I^{(i,2)})h_{e_v^+}P(v,i,\mathcal I^{(i,m-2)})+O(t^3).
\end{aligned}
\end{equation}

The second step is to compute the expectation value of the simplified operator. Let us still take $\sum_{\mathcal I^{(i,m)}}\mathcal F(\mathcal I^{(i,m)})[P(v,i,\mathcal {I}^{(i,m)}),h_{e_v^+}]$ as an example.
A subtlety here is that the operator $\sum_{\mathcal I^{(i,m)}}\mathcal F(\mathcal I^{(i,m)})[P(v,i,\mathcal {I}^{(i,m)}),h_{e_v^+}]$ involves two edges $e^\pm_v$ because of the operator $\hat P^{\alpha}(v,i)$. We consider its expectation value with respect to coherent state $|\psi_{g_{e_v^+}}\rangle\otimes |\psi_{g_{e_v^-}}\rangle$, where, without loss of generality, we set 
$$g_{e_v^+}=g_{e_v^-}=g=n e^{iz\tau_3}n^{-1}.$$ 

According to the above discussion, 
one can apply the result of Eq. \eqref{eq:example1p} to compute the expectation value of its LHS.
Let us take the first term for instance. We need to compute the expectation value of  
\begin{equation}\label{eq:expec1}
\begin{aligned}
\sum_{\mathcal I^{(i,m-1)},\mathcal I^{(i,1)},b}\wdtau(\mathcal I^{(i,1)},a,b)\left\langle D^{1/2}_{bc}(h_{e_v^+}) P(v,i,\mathcal I^{(i,m-1)}) \right\rangle_{g}\sum_{k=1}^m\mathcal F(\{\mathcal I^{(i,m-2)},\mathcal I^{(i,1)}\}_{(k)})
\end{aligned}
\end{equation}
Applying Eq. \eqref{eq:basicformula}, we finally obtain the following
\begin{equation}\label{eq:expec2}
\begin{aligned}
&\sum_{\mathcal I^{(i,m-1)},\mathcal I^{(i,1)},b}\wdtau(\mathcal I^{(i,1)},a,b)\left\langle D^{1/2}_{bc}(h_{e_v^+}) P(v,i,\mathcal I^{(i,m-1)}) \right\rangle_{g}\sum_{k=1}^m\mathcal F(\{\mathcal I^{(i,m-2)},\mathcal I^{(i,1)}\}_{(k)})\\
=&\sum_{\mathcal I^{(i,m-1)},\mathcal I^{(i,1)},b}\wdtau(\mathcal I^{(i,1)},a,b)D^{1/2}_{bb'}(n) \left\langle D^{1/2}_{b'c'}(h_{e_v^+}) P(v,i,\mathcal I^{(i,m-1)})\right\rangle_{z}\\
&D^{1/2}_{c'c}(n^{-1})\sum_{k=1}^m \widetilde{\mathcal F}(\{\mathcal I^{(i,m-2)},\mathcal I^{(i,1)}\}_{(k)})
\end{aligned}
\end{equation}
where $\widetilde{\mathcal F}$ is given by
\begin{equation*}
\widetilde{\mathcal F}(\{\alpha_1,\cdots,\alpha_k\}):=\sum_{\beta_1,\cdots,\beta_k}\mathcal F(\{\beta_1,\cdots,\beta_k\})D^1_{\alpha_1\beta_1}(n^{-1})\cdots D^1_{\alpha_k\beta_k}(n^{-1}).
\end{equation*}
To compute Eq. \eqref{eq:expec2}, the results of Theorem \ref{thm:leadingordermultiplyPH} can be applied. Thus, one obtains the possible $\mathcal I^{(i,m-1)}$ are:
\begin{itemize}
\item[(i)] If one only computes the leading-order term of the expectation value, then $\mathcal I^{(i,m-1)}$ is evaluated at $\underbrace{\{0,0,\cdots,0\}}_{m-1}$ 
\item[(ii)]
If one only computes the expectation value up to $O(t)$, then $\mathcal I^{(i,m-1)}$ is evaluated at $\underbrace{\{0,0,\cdots,0\}}_{m-1}$, $\underbrace{\{0,0,\cdots,0,\pm 1\}}_{m-1}$, $\underbrace{\{0,0,\cdots,0,1,-1\}}_{m-1}$ and $\underbrace{\{0,0,\cdots,0, -1,1\}}_{m-1}$.
\end{itemize}

With this discussion, Eq. \eqref{eq:expec2} can be computed by considering both the expectation-value part $$\left\langle D^{1/2}_{b'c'}(h_{e_v^+}) P(v,i,\mathcal I^{(i,m-1)})\right\rangle_{z}$$ and the non-operator-factor part $$\sum_{b}\wdtau(\mathcal I^{(i,1)},a,b)D^{1/2}_{bb'}(n)D^{1/2}_{c'c}(n^{-1})\sum_{k=1}^m \widetilde{\mathcal F}(\{\mathcal I^{(i,m-2)},\mathcal I^{(i,1)}\}_{(k)})$$
for each possibility of $\mathcal I^{(i,m-1)}$. 

For the expectation-value part, the values of $\left\langle D^{1/2}_{b'c'}(h_{e_v^+}) P(v,i,\mathcal I^{(i,m-1)})\right\rangle_{z}$ for Case (i) and the cases where $\mathcal I^{(i,m-1)}$ contains $\pm 1$ in Case (ii) are computed by Eq. \eqref{eq:leadingequalmultileading}. Furthermore, for the cases where $\mathcal I^{(i,m-1)}$ contains $\pm 1$ in Case (ii), the results are independent of the position of $\pm 1$ in $\mathcal I^{(i,m-1)}$. Therefore, only $1+2+2=5$ cases are considered in Case (ii) finally. Comparing with the original number of cases $3^{3m-1}$, one can obtain an advantage of our algorithm.  

For the non-operator-factor part, given a possible $\mathcal I^{(i,m-2)}$ and consider all possible $\mathcal I^{(i,1)}$s. Then all possible $\{\mathcal I^{(i,m-2)},\mathcal I^{(i,1)}\}_{(k)}$s can be reconstructed. It is noted that, for Case (ii), the positions of $\pm 1$ in $\mathcal I^{(i,m-2)}$ does matter to the value of $\widetilde{\mathcal F}$.  Thus we need to consider the permutations of indices in (with keeping the relative order between $\pm 1$) $\mathcal I^{(i,m-2)}$ in Case (ii) when reconstructing $\{\mathcal I^{(i,m-2)},\mathcal I^{(i,1)}\}_{(k)}$. With the results of these two parts, the results of Eq.\eqref{eq:expec1} is obtained correspondingly.

Finally, we complete this section with the discussion on the values of $n$ and $\vec k$ in $\hat H_E^{(n)}$ and $\hat H_L^{(\vec k)}$, {which is summarized as the following lemma
\begin{lmm}
To get the expectation values of $\hat H_E^{(n)}$ and $\hat H_L^{(\vec k)}$ up to order $O(t)$, it is sufficient to set $n$ and $\vec k=(k_1,k_2,k_3,k_4,k_5)$ respectively such that 
\begin{equation}
n\leq 3
\end{equation}
and 
\begin{equation}\label{eq:valuesofk}
\frac{|k_1+k_2-3|+(k_1+k_2-3)}{2}+\frac{|k_3+k_4-3|+(k_3+k_4-3)}{2}+k_5\leq 3.
\end{equation}
\end{lmm}
The proof of this lemma is sketched as follows. }
When one replaces $\hat{V}_v$ in $\widehat{H_E}$ by $\hat{V}_{G T}^{(v)}$, the term $\left(\hat Q^2/\langle\hat Q\rangle^2-1\right)^k$ in $\hat{V}_{G T}^{(v)}$ contributes as the following
\begin{equation}\label{eq:operatorqmqk}
\begin{aligned}
&\frac{1}{t}\tr(h_{\alpha_{IJ}}h_{s_K}[\left(\hat Q^2/\langle\hat Q\rangle^2-1\right)^k,h_{s_K}^{-1}])\\
=&\sum_{m=0}^{k-1}\tr(h_{\alpha_{IJ}}h_{s_K}\left(\hat Q^2/\langle\hat Q\rangle^2-1\right)^m\frac{[\hat Q^2/\langle\hat Q\rangle^2,h_{s_K}^{-1}]}{t}\left(\hat Q^2/\langle\hat Q\rangle^2-1\right)^{k-1-m})
\end{aligned}
\end{equation}
where the overall coefficient is neglected. Since $\hat Q^2/\langle\hat Q\rangle^2-1=\left(\hat Q/\langle\hat Q\rangle+1\right)\left(\hat Q/\langle\hat Q\rangle-1\right)$,
and $\langle \hat Q/\langle\hat Q\rangle-1 \rangle = O(t)$,
Theorem \ref{thm:leadingordergeneral} can be applied, which indicates that the leading-order expectation value of the operator \eqref{eq:operatorqmqk} is at least $O(t^{\floor{k/2}})$.
Thus, as far as  expanding the expectation value to $O(t)$ is concerning,  one can sufficiently choose $k\leq 3$, 
which leads to $n\leq 3$ in $\hat H_E^{(n)}$. A very similar discussion for $\hat H_L^{(\vec k)}$ can be done, which gives us Eq. \eqref{eq:valuesofk}.

\section{Quantum correction in the expectation value} \label{sec:results}
The resulting expectation value of the Hamiltonian with unit lapse $\widehat{H[1]}=\widehat{H_{E}}+\left(1+\beta^{2}\right) \widehat{H_{L}}$ at coherent states with cosmological data ($\eta<0$ in our convention) is shown as follows\footnote{Here, without loss of generality, the graph $\gamma$ is assumed to be a cubic lattice in $\mathbb T^3$ and has only a single vertex. }
\begin{eqnarray}
\langle \widehat{H_E}\rangle&=&6 a \sqrt{-\beta  \eta } \sin ^2(\xi )-\frac{3}{4} a t \sqrt{-\frac{\beta }{\eta ^3}} \sin ^2\left(\frac{\xi }{2}\right) \cos\left(\frac{\xi }{2}\right) \Bigg\{\cos\left(\frac{\xi }{2}\right) \Big[8 \eta ^2+8 \eta  (4 \cosh (\eta )-3) \text{csch}(\eta )-9\Big]\nonumber\\
&&-\,12 i \eta  \sin \left(\frac{\xi }{2}\right)\Bigg\}+O(t^2),\label{euclidhamres}\\
\langle \widehat{H_L}\rangle&=&-\frac{6 a \sqrt{-\beta  \eta } \sin ^2(\xi ) \cos ^2(\xi )}{\beta ^2}-\frac{3 a t}{262144 (-\beta  \eta )^{3/2}} \Bigg\{2 \left(3-220 \eta ^2\right) \cos (6 \xi )\nonumber\\
&&+\,4 i \eta  (4838 \sin (\xi )-6284 \sin (2 \xi )+4685 \sin (3 \xi )-5222 \sin (4 \xi )-105 \sin (5 \xi ))\nonumber\\
&& +\,2 (-3611+8 \eta  (492 \eta +11 i)) \cos (\xi )-2 (-789+4 \eta  (305 \eta +18 i)) \cos (2 \xi )\nonumber\\
&& +\,(4413-4 \eta  (928 \eta +49 i)) \cos (3 \xi )+8 (-1978+\eta  (4192 \eta -7 i)) \cos (4 \xi )\nonumber\\
&& +\,(-7+4 (-272 \eta +5 i) \eta ) \cos (5 \xi )-4 \eta  \coth (\eta ) \Big[536 \cos (\xi )+1731 \cos (2 \xi )\nonumber\\
&&+\,1524 \cos (3 \xi )-40548 \cos (4 \xi )+116 \cos (5 \xi )+117 \cos (6 \xi )+37292\Big]\nonumber\\
&&+\,8 \eta  \text{csch}(\eta ) \Big[130 \cos (\xi )+918 \cos (2 \xi )+801 \cos (3 \xi )-18618 \cos (4 \xi )\nonumber\\
&&+\,125 \cos (5 \xi )+58 \cos (6 \xi )+16362\Big]+8 (1436+\eta  (-4056 \eta +25 i))\Bigg\}+O(t^2). 
\end{eqnarray}
According to \eqref{altextra}, $\langle \widehat{H_L}\rangle=\langle{}^{\rm extr}\widehat{H_L}\rangle+\langle{}^{\rm alt}\widehat{H_L}\rangle$, in which $\langle{}^{\rm alt}\widehat{H_L}\rangle$ is, 
\begin{equation}
\begin{aligned}
\langle{}^{\rm alt}\widehat{H_L}\rangle=&-\frac{3 a \sqrt{-\beta  \eta } \sin ^2(2 \xi )}{8 \beta ^2}-\frac{3 a t}{32768 (-\beta  \eta )^{3/2}}\Bigg\{4 \left(104 \eta ^2-79\right) \cos (\xi )+\left(68-160 \eta ^2\right) \cos (2 \xi )\\
&+4 \left(55-104 \eta ^2\right) \cos (3 \xi )+\left(944 \eta ^2-501\right) \cos (4 \xi )-848 \eta ^2\\
&-2 \eta  \Big[\coth (\eta ) (-(262 \cos (\xi )-6 (46 \cos (2 \xi )+65 \cos (3 \xi )-380 \cos (4 \xi )+366)))\\
&-\text{csch}(\eta ) (-292 \cos (\xi )+268 \cos (2 \xi )+420 \cos (3 \xi )-2035 \cos (4 \xi )+1799)\\
&-64 i (6 \sin (\xi )-10 \sin (2 \xi )+6 \sin (3 \xi )-7 \sin (4 \xi ))\Big]+241\Bigg\}+O(t^2),
\end{aligned}
\end{equation}
and $\langle{}^{\rm extr}\widehat{H_L}\rangle$ is
\begin{equation}
\begin{aligned}
\langle{}^{\rm extr}\widehat{H_L}\rangle=&-\frac{9 a \sqrt{-\beta  \eta } \sin ^2(2 \xi )}{8 \beta ^2}+\frac{3 a t }{262144 (-\beta  \eta) ^{3/2}} \Bigg\{-2 \left(580 \eta ^2+72 i \eta -517\right) \cos (2 \xi )\\
&+2 \left(3-220 \eta ^2\right) \cos (6 \xi )+4 i \eta  \Big[3302 \sin (\xi )-3724 \sin (2 \xi )+3149 \sin (3 \xi )\\
&-35 (98 \sin (4 \xi )+3 \sin (5 \xi ))\Big]+2 (-2347+8 \eta  (284 \eta +11 i)) \cos (\xi )\\
&+(2653-4 \eta  (96 \eta +49 i)) \cos (3 \xi )+56 (-211+\eta  (464 \eta -i)) \cos (4 \xi )\\
&+(-7+4 (-272 \eta +5 i) \eta ) \cos (5 \xi )-4 \eta  \coth (\eta ) \Big[1584 \cos (\xi )+627 \cos (2 \xi )\\
&-36 \cos (3 \xi )-31428 \cos (4 \xi )+116 \cos (5 \xi )+117 \cos (6 \xi )+28508\Big]\\
&+8 \eta  \text{csch}(\eta ) \Big[714 \cos (\xi )+382 \cos (2 \xi )-39 \cos (3 \xi )\\
&-14548 \cos (4 \xi )+125 \cos (5 \xi )+58 \cos (6 \xi )+12764\Big]+8 \eta  (-3208 \eta +25 i)+9560\Bigg\}+O(t^2).
\end{aligned}
\end{equation}

When expressing $\langle \widehat{H[1]}\rangle$ to be $\langle \widehat{H[1]}\rangle=H_0+t H_1+O(t^2)$, we notice that the $O(t)$-term $H_1$ contains $\eta$ in its denominator and, thus, is divergent if $\eta\to0$. 
This feature is implied by the fact that if $\langle\hat{Q}_v\rangle\to0$, $\hat{V}_{GT}^{(v)}$ is divergent. This is because when if $\eta\to0$, $\langle \hat{Q}_v\rangle\to0$ since $\langle \hat{Q}_v\rangle\sim |\eta|^{3}$.

Hence, the expansion of $\langle \widehat{H[1]}\rangle$ becomes invalid when $\eta$ is too small. More precisely, expressing $\langle \widehat{H[1]}\rangle$ as $\langle \widehat{H[1]}\rangle=\sqrt{|\eta|}\left[\mathfrak{f}_0+({t}/{\eta^2})\, \mathfrak{f}_1+O(t^2)\right]$, we get that $\mathfrak{f}_0$ is independent of $\eta$, and $\mathfrak{f}_1$ is regular at $\eta\to 0$. Thus, it is concluded that our expansion is valid when $ \eta^{2}\gg t $. This aspect is important for a new improvement of cosmological effective dynamics derived from the full LQG \cite{Han:2021cwb}. The expansion is valid for large $|\eta|$, because $\mathfrak{f}_1$ is regular at $|\eta|\to\infty$.

Consider the reduced-phase-space LQG of gravity coupled to Gaussian dust. Then, the relational evolution with respect to the dust time $T$ will be generated by the physical Hamiltonian $\widehat{\bf H}=\frac{1}{2}\left(\widehat{H[1]}+\widehat{H[1}]^{\dagger}\right)$. Its coherent state expectation value reads  
\begin{equation}
\langle\widehat{\bf H}\rangle=\Re(\langle \widehat{H[1]}\rangle)\equiv H_0+t \widetilde{H}_1+O(t^2),\quad  \widetilde{H}_1=\Re (H_1).\label{effhamcos}
\end{equation}

In order to demonstrate the physical application and effects from the $O(\hbar)$ correction we adopt the proposal in \cite{Dapor:2017rwv} as follows. Firstly, we view $\langle\widehat{\bf H}\rangle$ as the effective Hamiltonian on the 2-dimensional phase space $\mathscr{P}_{cos}$ of homogeneous and isotropic cosmology. Then, one can verify that $\eta =-\frac{\mu^2 P_0}{\beta a^2}$ and $\xi=\mu\beta K_0$ where $\mu$ is the coordinate length of $e\in E(\gamma)$, $P_0$ is the square of the scale factor and $K_0$ is the extrinsic curvature. Thus, the Poisson bracket between $\xi$ and $\eta$ reads 
$\{\eta,\xi\}=\frac{\kappa}{3 a^2 }$. With this Poisson bracket, the Hamiltonian time evolution on $\mathscr{P}_{cos}$ generated by $\langle\widehat{\bf H}\rangle$ is computable. The numerical result is shown in Fig. \ref{fig:dynamics}, which respectively depicts the dynamics of the spatial volume governed by $H_0$, $\langle\widehat{\bf H}\rangle=H_0+t \widetilde{H}_1$ and the classical FLRW Hamiltonian $H_{cl}=-6a\beta^{-3/2}\sqrt{-\eta}\,\xi^2$.

\begin{figure}[h]
\centering
\includegraphics[width=0.7\textwidth]{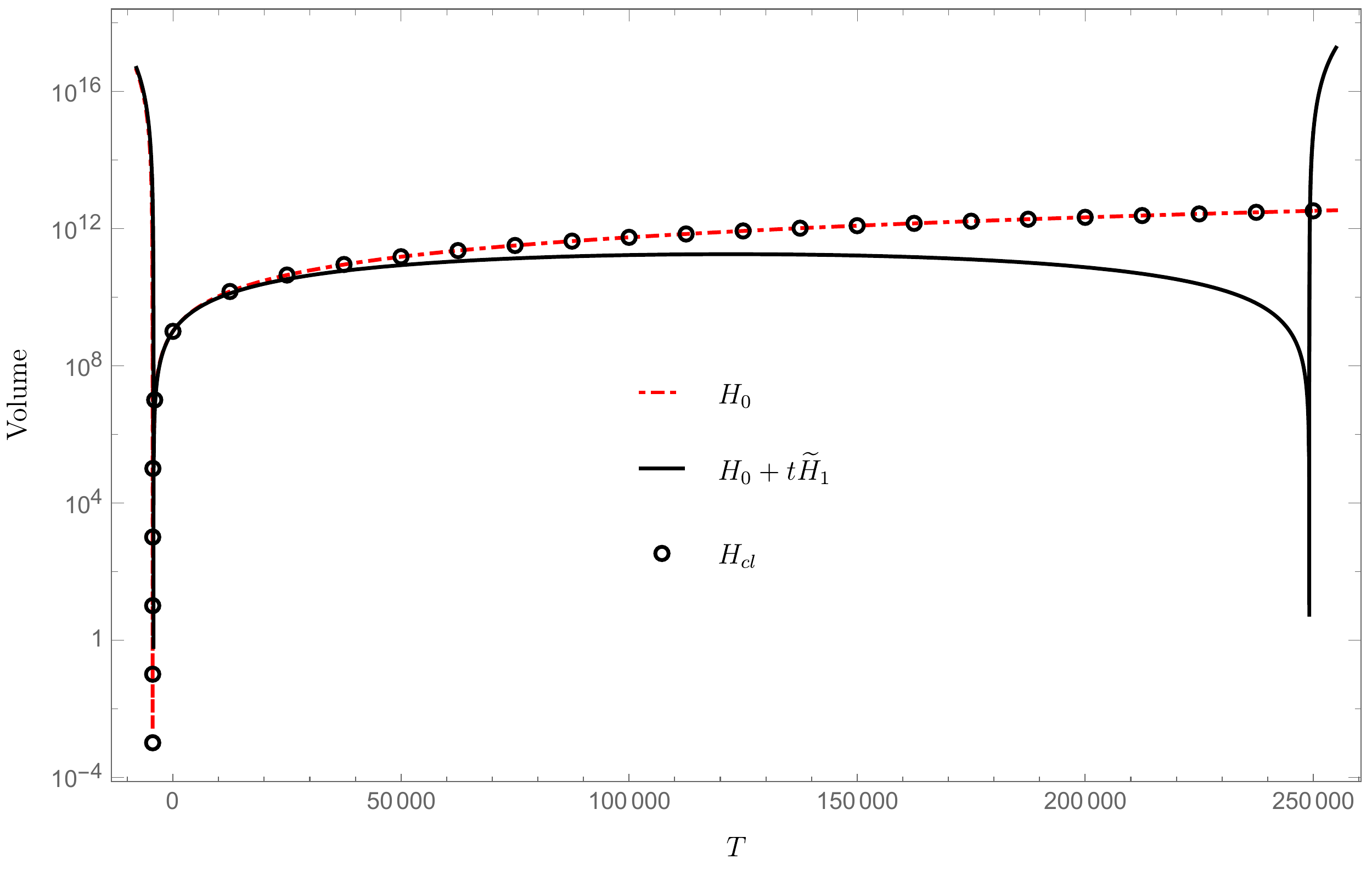}
\caption{Evolution of the spatial volume generated by $\langle\widehat{\bf H}\rangle=H_0+t \widetilde{H}_1$ in \eqref{effhamcos} (black curve),  $H_0$ (red dashed curve), and $H_{cl}$ (black circles).
The coincided initial data of $\eta,d\eta/dT$ are chosen to be at $T=0$ for these 3 cases. The parameters are set to be $a=1$, $t=10^{-5}$, and $\beta=0.2375$. }
\label{fig:dynamics}
\end{figure}

\begin{figure}[h]
    \centering
    \includegraphics[width=0.7\textwidth]{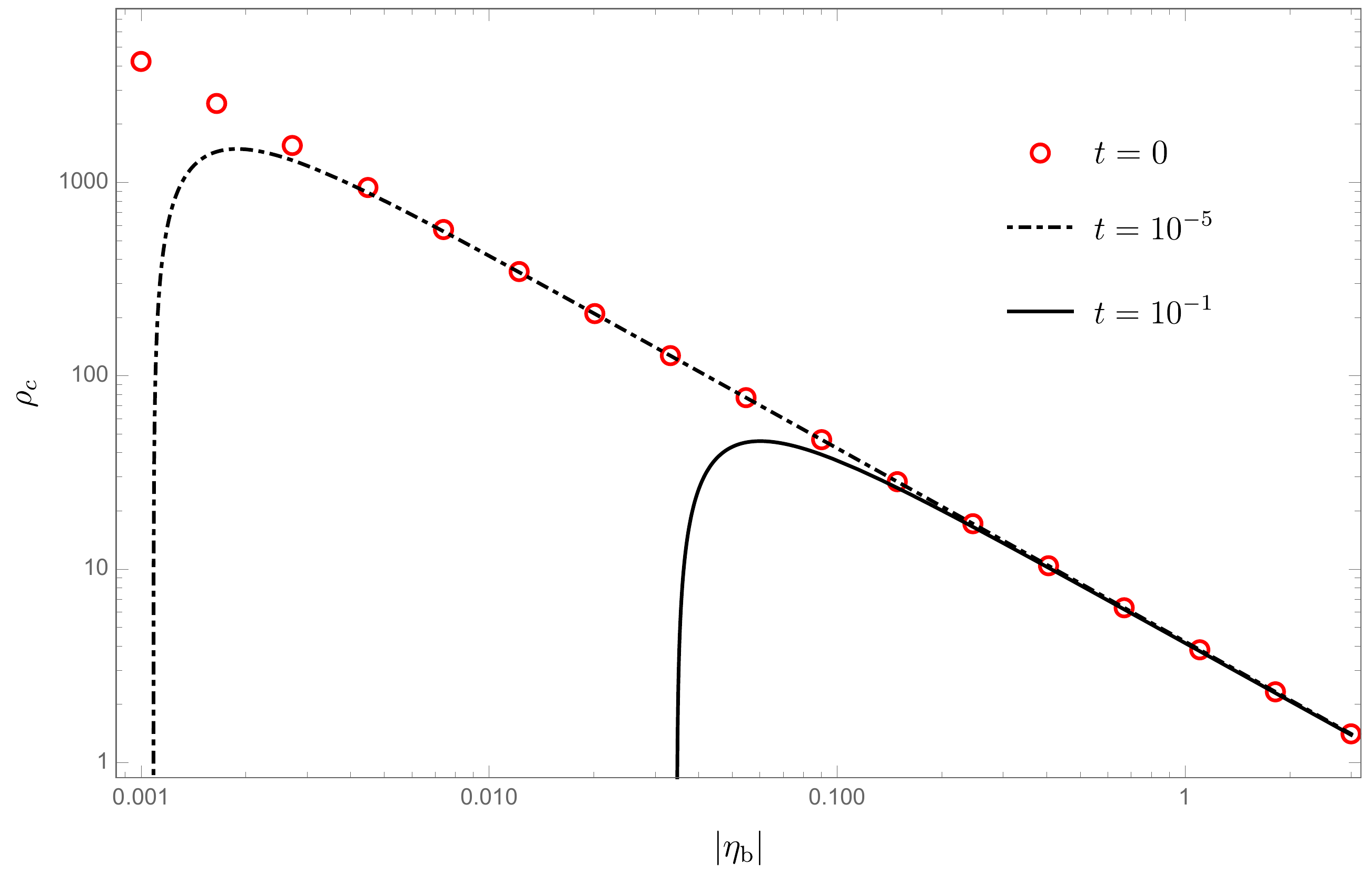}
    \caption{
    Plots of the critical density $\rho_c$ of the dynamics governed by $H_0$ (red circle) and,  $H_0+t\widetilde{H}_1$ with  $t=10^{-5}$ (dash-dotted curve) $t=0.01$ (solid curve). The parameters are set as $a=1$ and $\beta=0.2375$.}
    \label{fig:rhoc}
\end{figure}

In the example shown in Fig. \ref{fig:dynamics}, a relatively small $t=10^{-5}$ is chosen to display the effects of the next-to-leading-order term. The coincided initial data of $\eta,d\eta/dT$ are chosen to be at $T=0$ for all of the three cases. Since $\eta|_{T=0}$ gives a large spatial volume, $T=0$ is in the low-energy-density regime. 
As shown in Fig.\ref{fig:dynamics}, toward $T<0$, the evolution with respect to $H_0+t \widetilde{H}_1$ is similar as that corresponding to $H_0$, where the latter one gives the $\mu_0$-scheme effective dynamics. 
Both of the effective Hamiltonian $H_0+t \widetilde{H}_1$ and $H_0$ resolve the big-bang singularity by ``a bounce". Moreover, the dynamics of $H_0$ and $H_{cl}$ at late time $T>0$ are compatible. Further more, it is shown in Fig. \ref{fig:dynamics} that $t \widetilde{H}_1$ in $H_0+t \widetilde{H}_1$ behaves like an additional matter distribution with negative 
energy density of $O(t)$, which causes the universe to re-collapse and have another bounce at very late time. It should be noted that the time of re-collapse is extremely late because of the tiny value of $t$. It can be verified numerically that the time when the recollapse happens is correlated to $t$, namely, the recollapse happens earlier for a larger $t$ but happens later for a smaller $t$. For vanishing $t$, it needs infinitely long time before the recollapse happens, i.e. there is no recollapse. Finally, the collapse occurs in the FLRW phase with $\xi\ll 1$. Thus expanding $\widetilde{H}_1$ at $\xi=0$, we have
\begin{equation}
\widetilde{H}_1\cong \frac{3 a (-\beta  \eta )^{3/2} \left(1280 \eta ^2-3072 \eta  \coth (\eta )-1792 \eta  \text{csch}(\eta )-5568\right)}{262144 \beta^3 \eta ^3}+O\left(\xi ^2\right).
\end{equation}
That is, when $\xi$ is vanishing, $\widetilde{H}_1$ contains terms which do not vanish. It is these terms that lead to the recollapse. These terms are all contributed by the Lorentzian part, and is the consequence from the Thiemann's regularization of the Lorentzian Hamiltonian. 

For the critical density $\rho_c$ plotted in Fig.\ref{fig:rhoc}, it also receives correction from $t \widetilde{H}_1$ at the $T<0$ bounce. 
As is known, the critical density from the $\mu_0$-scheme dynamics governed by $H_0$ reads $\rho_{c}=\frac{3}{2a^2\kappa \beta^{3}\left(\beta^{2}+1\right)  |\eta_b|}$ with $|\eta_b|$ being the value of $|\eta|$ at the bounce. {When $\widetilde{H}_1$ term is considered, the dependence of $\rho_c$ on $\eta_b$ is only known numerically. Some numerical results are presented in  Fig. \ref{fig:rhoc}.}
As shown there, instead of blowing up for small $|\eta_b|$, the corrected $\rho_c$ from $\langle\widehat{\bf H}\rangle=H_0+t \widetilde{H}_1$ is bounded from above for small values of $|\eta_b|$. 
In an optimistic viewpoint, this correction of $\rho_c$ might hint that the correction from higher-order term in $t$ could potentially flatten the dependence
of $\eta_b$ in $\rho_c$. This flattened behavior of $\rho_c$ is also supported by the current model of $\bar{\mu}$-scheme effective dynamics with complete quantum corrections, which is considered as an important feature of $\bar{\mu}$-scheme Loop Quantum Cosmology. However, by recalling the fact that our expansion in $t$ requires $\eta^2\gg t$, one realizes that the small $|\eta_b|$ regime, where the correction of $\rho_c$ becomes significant, mostly violates this requirement (see fig.\ref{fig:rhoc}). Hence the quantum dynamics near the bounce is still an open problem from this point of view.

It should be emphasized that the proposal \cite{Dapor:2017rwv} adopted here for studying the effective dynamics is not as rigorous as the path integral formula \eqref{Agg0}. As argued in Section \ref{sec:intro}, the $O(t)$ correction in $\langle\widehat{H[1]}\rangle$ is only a partial contribution in the quantum effective action that ultimately determining the quantum effect in the dynamics. Hence, the cosmological dynamics plotted in Fig.\ref{fig:dynamics} only shows $O(t)$ correction in $\langle\widehat{H[1]}\rangle$ from one of the three $O(t)$ contributions in $\Gamma$, and is not yet a rigorous prediction from the principle of LQG. We have only focused on non-gauge-invariant coherent state and neglected the group averaging that impose the gauge invariance in the path integral \eqref{Agg0}. {Moreover, this paper only focus on the effects of the quantum correction of linear order in $t$, and the effects from higher order corrections are still unclear and beyond the scope of this paper.}

\section{Conclusion and outlook}\label{sec:conclusion}

In this paper, we developed an algorithm to overcome the complexity of computing the expectation value of LQG Hamiltonian operator $\widehat{H[N]}$. With this algorithm, the $O(\hbar)$ correction in the expectation value $\langle\widehat{H[N]}\rangle$ at the coherent state peaked at the homogeneous and isotropic data of cosmology is computed. In the current work, there are several perspectives which should be addressed in the future:

The first one is to complete the computation of the quantum effective action $\Gamma$ mentioned in Section \ref{sec:intro}. 
After completing of this current work, the only missing ingredient in the $O(\hbar)$ terms of $\Gamma$ is the ``1-loop determinant'' $\det(\mathfrak{H})$. Therefore, a research to be carried out immediately is to compute this correction of $\det(\mathfrak{H})$ at the homogeneous and isotropic background. 
Once we obtain all of the $O(\hbar)$ contribution to $\Gamma$, the variation of $\Gamma$ should give the quantum corrected effective equations which will demonstrate the quantum correction to the cosmological model implied by LQG.    

The next step of generalizing our computation is to study the expectation values of $\langle\widehat{H[N]}\rangle$ with respect to the coherent states peaked at cosmological perturbations. The semiclassical limits of the expectation value and the cosmological perturbation theory from the path integral \eqref{Agg0} have been studied in \cite{Han:2020iwk}. 
Thus, it is interesting to study the $O(\ell_{\rm p})$ correction to the cosmological perturbation theory.

Finally, the computation of the quantum correction in $\langle\widehat{H[N]}\rangle$ should also be extended to the model of gravity coupled to standard matter fields. The contributions of matter fields to $\widehat{H[N]}$ have been studied in \cite{Thiemann:1997rt,Sahlmann:2002qj}. Since the matter parts in $\widehat{H[N]}$ is much simpler than the Lorentzian part in $\widehat{H[N]}$, the computation of their expectation values should not be hard. Study of the matter contributions and their quantum corrections is a project currently undergoing \cite{inpreparation}. 

Our work expands $\langle\widehat{H[N]}\rangle$ to the order $t^1$ and neglects the higher order terms. It is important to understand higher order contributions, or namely the reminder of the expansion, in order to rigorously estimate the regime in the phase space where the expansion is good (see e.g. the paragraph above (\ref{effhamcos}) for the discussion of the phase space regime such that $t\widetilde{H}_1$ is small).

\section*{Acknowledgements}

M.H. acknowledges Andrea Dapor, Klaus Liegener, and Hongguang Liu for various discussions motivating this work. M.H. receives support from the National Science Foundation through grant PHY-1912278. C. Z. acknowledges the support by the Polish Narodowe Centrum Nauki, Grant No. 2018/30/Q/ST2/00811, and NSFC with Grants No. 11961131013, No. 11875006.

\input{appendix}

\bibliography{reference}
\bibliographystyle{jhep}

\end{document}

%% file: appendix.tex
\appendix
\section{SL($2,\mathbb C$) and SU(2) groups}\label{app:sl2csu2}
Let $\sigma_i$ with $i=1,2,3$ be the Pauli matrices and $\tau_k:=-i\sigma_k/2$. Define 
\begin{equation}
\vec\theta=(\theta \sin(\psi)\cos(\phi),\theta \sin(\psi)\sin(\phi),\theta\cos(\psi)). 
\end{equation}
$h\in$SU(2) can be coordinatized as
\begin{equation}
h=e^{\vec\theta\cdot\vec\tau}=\left(
\begin{array}{cc}
 \cos \left(\frac{\theta }{2}\right)-i \sin \left(\frac{\theta }{2}\right) \cos (\psi ) &-i \sin \left(\frac{\theta }{2}\right) \sin (\psi ) e^{-i\phi} \\
 -i\sin \left(\frac{\theta }{2}\right) \sin (\psi )e^{i\phi} & \cos \left(\frac{\theta }{2}\right)+i \sin \left(\frac{\theta }{2}\right) \cos (\psi ) \\
\end{array}
\right)
\end{equation}
with $\theta,\ \phi\in(0,2\pi)$ and $\psi\in(0,\pi)$. The Haar measure then is 
\begin{equation}
\dd\mu_H(\vec\theta)=\frac{1}{4 \pi ^2 }\sin ^2(\frac{\theta}{2})\sin(\psi)\dd\theta\dd\psi\dd\phi.
\end{equation}
Moreover, by defining 
\begin{equation}
\vec p=(p \sin(\alpha)\cos(\beta),p \sin(\alpha)\sin(\beta),p\cos(\alpha))
\end{equation}
with $p>0$, $\alpha\in (0,\pi)$ and $\beta\in(0,2\pi)$,
$g\in \mathrm{SL}(2,\mathbb C)$ can be parameterized as
\begin{equation}
g(\vec p,\vec\theta)=e^{i\vec p\cdot\vec \tau}e^{\vec\theta\cdot\vec\tau}=:e^{i\vec p\cdot\vec \tau} h(\vec\theta).
\end{equation}
For each $p>0$, there exists $u_{\vec p}^\pm \in\mathrm{SU}(2) $ such that
\begin{equation}\label{eq:up}
(u_{\vec p}^\pm)^{-1}\, \vec p\cdot \vec\tau\, (u_{\vec p}^\pm)=\pm p\tau_3.
\end{equation}
Note that $u_{\vec p}^\pm$ is determined by Eq. \eqref{eq:up} up to a right transformation by $e^{\alpha\tau_3}$. Namely $u_{\vec p}^\pm e^{\alpha\tau_3}$ for all $\alpha\in \mathbb R$ are solution to Eq. \eqref{eq:up} provided $u_{\vec p}^\pm$ does. Moreover, $u_{\vec p}^\pm$ has the relation 
\begin{equation}
u_{\vec p}^\pm =u_{-\vec p}^\mp.
\end{equation}
Let us denote $\eta\equiv \pm p$ and $u\equiv u_{\vec p}^\pm$ for convenience. Then
\begin{equation}\label{eq:eipsu2}
g(\vec p,\vec\theta)=u\, e^{i\eta\tau_3}\, u^{-1} h(\vec\theta).
\end{equation}
For $u^{-1} h(\vec\theta)\in$SU(2), decomposing it as 
\begin{equation}\label{eq:uh}
u^{-1} h(\vec\theta)= e^{-\xi\tau_3} n^{-1},
\end{equation}
one get
\begin{equation}
g(\vec p,\vec\theta)=u e^{i(\eta+i\xi)\tau_3}n^{-1}. 
\end{equation}
Note that Eq. \eqref{eq:uh} determines $n$ up to a right transformation by $e^{\alpha\tau_3}$ as that for $u_{\vec p}^\pm$. $n$ satisfies 
\begin{equation}
n(\eta\tau_3) n^{-1}=  h(\vec\theta) (\vec p\cdot\vec\tau) h(\vec\theta)^{-1}.
\end{equation}

The Wigner 3-$j$ symbol $\left(
\begin{array}{ccc}
j_1&j_2&j_3\\
m_1&m_2&m_3
\end{array}
\right)$ is an SU(2)-invariant tensor, namely 
\begin{equation}\label{eq:SU2invariant}
D^{j_1}_{n_1m_1}(h)D^{j_2}_{n_2m_2}(h)D^{j_3}_{n_3m_3}(h)\left(
\begin{array}{ccc}
j_1&j_2&j_3\\
m_1&m_2&m_3
\end{array}
\right)=\left(
\begin{array}{ccc}
j_1&j_2&j_3\\
n_1&n_2&n_3
\end{array}
\right),\ \forall h\in \mathrm{SU}(2).
\end{equation}
Define $\tau_\alpha$ with $\alpha=-1,0,1$ as
\begin{equation}
\tau_{\pm 1}=\mp\frac{\tau_1\pm i \tau_2}{\sqrt{2}},\ \tau_0=\tau_3.
\end{equation}
We obtain the $j$-representation of $\tau_\alpha$ in terms of the 3$j$ symbol, according to the Wigner–Eckart theorem, as
\begin{equation}\label{eq:tauj}
D^{'j}_{mm'}(\tau_\alpha)=iw_j\epsilon^j_{nm}\left(
\begin{array}{ccc}
j&j&1\\
n&m'&\alpha
\end{array}
\right)
\end{equation}
where $w_j=\sqrt{j(j+1)(2j+1)}$ and $\epsilon^j_{nm}=(-1)^{j+m}\delta(n,-m)$ is the 2-$j$ symbol. The 2-$j$ symbol is also SU(2) invariant. By 3-$j$ symbol and $\epsilon^j_{nm}$, any SU(2)-intertwiner can be constructed as 
\begin{equation}
\iota(k_1k_2\cdots k_{n-1})_{m_1m_2m_3\cdots m_n}=\left(
\begin{array}{ccc}
j_1&j_2&k_1\\
m_1&m_2&l_1
\end{array}
\right)\epsilon^{k_1}_{l_1l_1'}
\left(
\begin{array}{ccc}
k_1&j_3&k_2\\
l_1'&m_3&l_2
\end{array}
\right)\epsilon^{k_2}_{l_2l_2'}\cdots\epsilon^{k_{n-2}}_{l_{n-2}l_{n-2}'}\left(
\begin{array}{ccc}
k_{n-2}&j_{n-1}&j_n\\
l_{n-2}'&m_{n-1}&m_n
\end{array}
\right).
\end{equation}
Moreover, the Clebsch-Gordan coefficients relates to $3j$-symbol as
\begin{equation}\label{eq:CG3j}
\begin{aligned}
\langle j_1m_1j_2m_2|JM\rangle=&(-1)^{2j_2}
\sqrt{2J+1}\,\epsilon^J_{MN}\,\left(
\begin{array}{ccc}
J&j_2&j_1\\
N&m_2&m_1
\end{array}
\right)\\
=&(-1)^{j_1-j_2-J}
\sqrt{2J+1}\left(
\begin{array}{ccc}
j_1&j_2&J\\
m_1&m_2&N
\end{array}
\right)\epsilon^J_{NM}
\end{aligned}
\end{equation}

\section{the Clebsch-Gordan coefficients with negative parameters}\label{app:negativeCG}
Given $j_1$, $m_1$, $j_2$, $m_2$ and $m=m_1+m_2$, the Clebsch-Gordan coefficients $\langle j_1m_1j_2m_2|pm\rangle\equiv \left[
\begin{array}{ccc}
j_1&j_2&p\\
m_1&m_2&m
\end{array}
\right]$ for various $p$ satisfy the difference equation \cite{natura1988nature} %\cite{nature of the symmetry group of the 6j symbol}
\begin{equation}\label{eq:recurrenceCG}
A(p+1)\left[
\begin{array}{ccc}
j_1&j_2&p+1\\
m_1&m_2&m
\end{array}
\right]+A(p)
\left[
\begin{array}{ccc}
j_1&j_2&p-1\\
m_1&m_2&m
\end{array}
\right]+(A_0(p)-m_1+m_2)\left[
\begin{array}{ccc}
j_1&j_2&p\\
m_1&m_2&m
\end{array}
\right]=0
\end{equation}
where $\max(|m|,|j_1-j_2|)\leq p\leq j_1+j_2$ and
\begin{equation}
\begin{aligned}
A(p)=&\frac{1}{p}\sqrt{-\frac{(p^2-\xi_1^2)(p^2-\xi_2^2)(p^2-\xi_3^2)}{4p^2-1}},\\
 A_0(p)=&\frac{\xi_1\xi_2\xi_3}{p(p+1)}
\end{aligned}
\end{equation}
with 
\begin{equation*}
\xi_1=j_1-j_2,\ \xi_2=j_1+j_2+1,\ \xi_3=m. 
\end{equation*}
With the initial data
\begin{equation}\label{eq:initialCG}
\begin{aligned}
&\left[
\begin{array}{ccc}
j_1&j_2&j_1+j_2\\
m_1&m_2&m
\end{array}
\right] =\sqrt{\frac{\left(2 j_1\right)! \left(2 j_2\right)! \left(j_1+j_2-m\right)! \left(j_1+j_2+m\right)!}{\left(2 j_1+2 j_2\right)! \left(j_1-m_1\right)! \left(j_1+m_1\right)! \left(j_2-m_2\right)! \left(j_2+m_2\right)!}}
\end{aligned}
\end{equation}
the Clebsch-Gordan coefficients for other values of $p$ are computable with Eq. \eqref{eq:recurrenceCG}. 

In order to extend the Clebsch-Gordan coefficients to negative parameters, we define a function 
\begin{equation}
\mathcal C_0(x):=\sqrt{\frac{\left(2 j_2\right)!\, \Gamma(2x+1) \Gamma(x+j_2-m+1)\Gamma(x+j_2+m+1)}{\left(j_2-m_2\right)! \left(j_2+m_2\right)!\Gamma(2x+2j_2+1)\Gamma(x-m_1+1)\Gamma(x+m_1+1)  }},
\end{equation} 
with which 
\begin{equation}
\left[
\begin{array}{ccc}
j_1&j_2&j_1+j_2\\
m_1&m_2&m
\end{array}
\right] =\mathcal C_0(j_1).
\end{equation}
It is remarkable that $\mathcal C_0(x)$ is well-defined not only for positive $x$ such that
\begin{equation}
 x-m_1\in \mathbb Z,\ x\geq |m_1|
\end{equation}
but also for negative $x$ satisfying 
\begin{equation}
 x-m_1\in\mathbb N,\ x\leq\min(-|m_1|,-j_2-|m|-1)
\end{equation}
where those gamma functions with negative integers is understood as
\begin{equation}
\frac{\Gamma(-m_1)\cdots \Gamma(-m_k)}{\Gamma(-n_1)\cdots \Gamma(-n_k)}=\lim_{z\to 0}\frac{\Gamma(z-m_1)\cdots\Gamma(z-m_k)}{\Gamma(z-n_1)\cdots \Gamma(z-n_k)}=(-1)^{n_1+\cdots+n_k-m_1-\cdots-m_k}\frac{n_1!\cdots n_k!}{m_1!\cdots m_k!}. 
\end{equation}
By definition, the Clebsch-Gordan coefficients $\left[
\begin{array}{ccc}
j_1&j_2&j_1+j_2-\iota\\
m_1&m_2&m
\end{array}
\right]$ is obtained by applying the recurrence relation \eqref{eq:recurrenceCG}
successively for $\iota$ steps with the initial data $\mathcal C_0(j_1,j_1+j_2)$.
Then, we defined
$\left[
\begin{array}{ccc}
-j_1&j_2&-j_1+j_2-\iota\\
m_1&m_2&m
\end{array}
\right]$, the Clebsch-Gordan coefficients with negative parameters, as the result by applying  the recurrence relation 
\begin{equation}\label{eq:recurrencemj}
\begin{aligned}
\left[
\begin{array}{ccc}
-j_1&j_2&-q-1\\
m_1&m_2&m
\end{array}
\right]=-\frac{1}{\tilde A(-q)}\left(
\tilde A(-q+1)\left[
\begin{array}{ccc}
-j_1&j_2&-q+1\\
m_1&m_2&m
\end{array}
\right]+(\tilde A_0(-q)-m_1+m_2)\left[
\begin{array}{ccc}
-j_1&j_2&-q\\
m_1&m_2&m
\end{array}
\right]
\right)
\end{aligned}
\end{equation}
with the initial data $\mathcal C_0(-j_1)$, where 
\begin{equation}
\tilde A(q)=A(q)\big|_{j_1\to -j_1},\ \tilde A_0(q)=A_0(q)\big|_{j_1\to -j_1}.
\end{equation} 
This definition extended the Clebsch-Gordan coefficients to negative parameters. It guarantees that,
\begin{equation}
\begin{aligned}
\left[
\begin{array}{ccc}
-j_1&j_2&-j_1+j_2-\iota\\
m_1&m_2&m
\end{array}
\right]=\left.\left[
\begin{array}{ccc}
j_1&j_2&j_1+j_2-\iota\\
m_1&m_2&m
\end{array}
\right]\right|_{j_1\to-j_1}. 
\end{aligned}
\end{equation}  

 By definition, it has
\begin{equation}\label{eq:c0andthreej}
\begin{aligned}
\mathcal C_0(-j_1)&=\sqrt{\frac{\left(2 j_1-2 j_2-1\right)! \left(2 j_2\right)! \left(j_1-m_1-1\right)! \left(j_1+m_1-1\right)!}{\left(2 j_1-1\right)! \left(j_1-j_2-m-1\right)! \left(j_1-j_2+m-1\right)! \left(j_2-m_2\right)! \left(j_2+m_2\right)!}}\\
=&(-1)^{j_2+{m_2}}\left[
\begin{array}{ccc}
j_1-1&j_2&j_1-j_2-1\\
m_1&m_2&m
\end{array}
\right].
\end{aligned}
\end{equation}
Moreover, $\left[
\begin{array}{ccc}
j_1-1&j_2&j_1-j_2+\iota-1\\
m_1&m_2&m
\end{array}
\right]$ can also be obtained by applying successively the recurrence relation 
\begin{equation}\label{eq:recurrencejm1}
\begin{aligned}
\left[
\begin{array}{ccc}
j_1&j_2&q\\
m_1&m_2&m
\end{array}
\right]=-\frac{1}{B(q)}\left(
B(q-1)\left[
\begin{array}{ccc}
j_1&j_2&q-2\\
m_1&m_2&m
\end{array}
\right]+(B_0(q-1)-m_1+m_2)\left[
\begin{array}{ccc}
j_1&j_2&q-1\\
m_1&m_2&m
\end{array}
\right]
\right)
\end{aligned}
\end{equation}
with the initial data $\left[
\begin{array}{ccc}
j_1&j_2&j_1-j_2-1\\
m_1&m_2&m
\end{array}
\right]$, where 
\begin{equation}
B(q)=A(q)\big|_{j_1\to j_1-1},\ B_0(q):=A_0(q)\Big|_{j_1\to j_1-1}.
\end{equation}
 Furthermore, it can be verified that
\begin{equation}
B(q)=-\tilde A(-q),\ B_0(q-1)=\tilde A_0(-q).
\end{equation}
Therefore, according to Eqs. \eqref{eq:c0andthreej}, \eqref{eq:recurrencemj} and \eqref{eq:recurrencejm1}, we finally have
\begin{equation}
\left[
\begin{array}{ccc}
-j_1&j_2&-j_1+j_2-\iota\\
m_1&m_2&m
\end{array}
\right]=(-1)^{j_2+m_2-\iota}\left[
\begin{array}{ccc}
j_1-1&j_2&j_1-j_2+\iota-1\\
m_1&m_2&m
\end{array}
\right],
\end{equation}
namely
\begin{equation}\label{eq:negativeCG}
\left[
\begin{array}{ccc}
-j_1&j_2&-j_1+\Delta\\
m_1&m_2&m
\end{array}
\right]=(-1)^{\Delta+m_2}\left[
\begin{array}{ccc}
j_1-1&j_2&j_1-1-\Delta\\
m_1&m_2&m
\end{array}
\right].
\end{equation}

\section{Proof of \eqref{eq:pspth}}\label{app:graph}
One can refer to \cite{yang2017graphical,zhang2018towards,makinen2019introduction} for more details on this method. 
The 2-$j$ symbol is graphically represented as
\begin{equation}
\epsilon^j_{mn}=(-1)^{j+n}\delta(m,-n)=\makeSymbol{\includegraphics[width=0.1\textwidth]{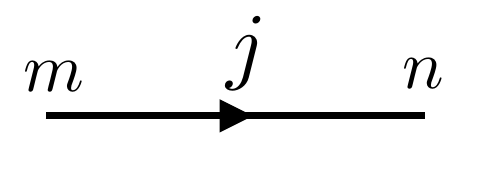}}.
\end{equation}
The $3j$-symbols is graphically represented as
\begin{equation}
\left(
\begin{array}{ccc}
j_1&j_2&j_3\\
m_1&m_2&m_3
\end{array}
\right)=\makeSymbol{\includegraphics[width=0.15\textwidth]{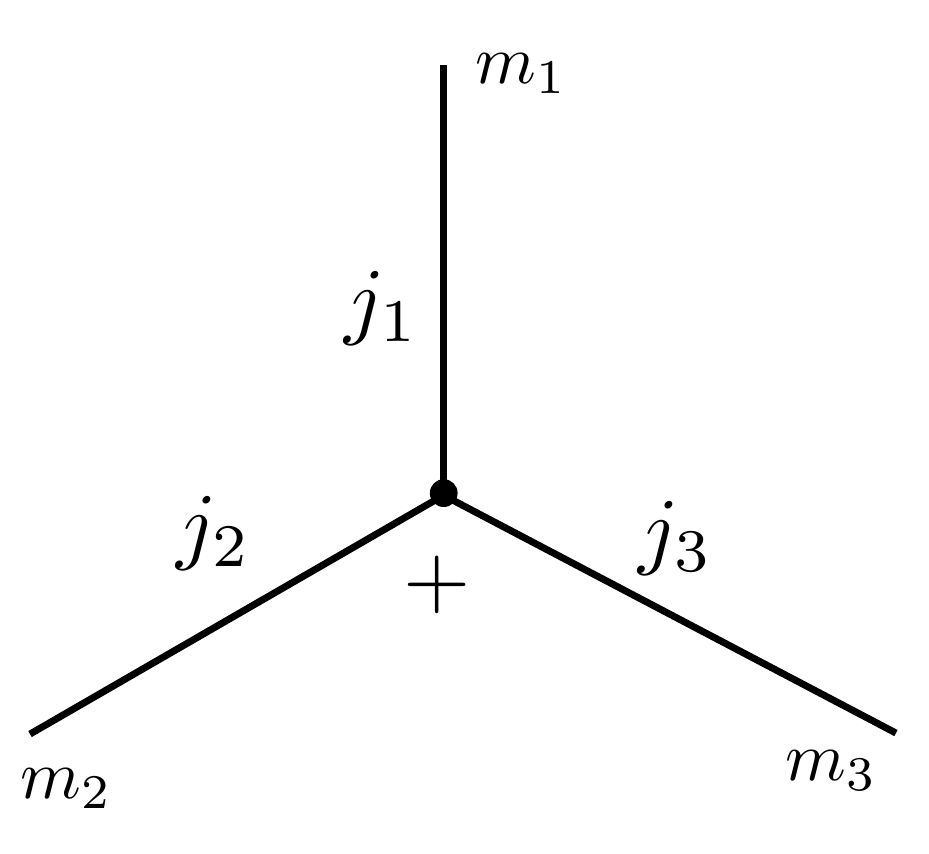}}.
\end{equation}
The Wigner-D matrix $D^j_{mn}(h)$, as a tensor $h\in \mathcal H_j\otimes\mathcal H_j^*$, is 
\begin{equation}
D^j_{mn}(h)=\langle jm|h|jn\rangle=\makeSymbol{\includegraphics[width=0.2\textwidth]{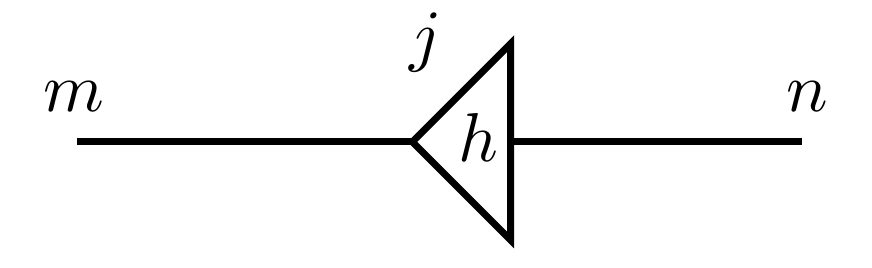}}.
\end{equation}

For the multiplication operator  $D^\iota_{ab}(h_e)$,  its action on $D^j_{mn}(h_e)$ reads
\begin{equation}\label{eq:coupleholonomy}
\begin{aligned}
D^{\iota}_{ab}(h_e)D^{j}_{mn}(h_e)=\sum_{J={j-\iota}}^{j+\iota}d_J(-1)^{M-N}\left(
\begin{array}{ccc}
\iota&j&J\\
a&m&-M
\end{array}
\right)\left(
\begin{array}{ccc}
\iota&j&J\\
b&n&-N
\end{array}
\right)D^J_{MN}(h),
\end{aligned}
\end{equation}
where $\left(\begin{array}{ccc}
\iota&j&J\\
a&m&-M
\end{array}
\right)$ denotes the Wigner $3j$-symbol. This graphically corresponds to
\begin{equation}\label{eq:coupleholonomyp}
\makeSymbol{\includegraphics[width=0.2\textwidth]{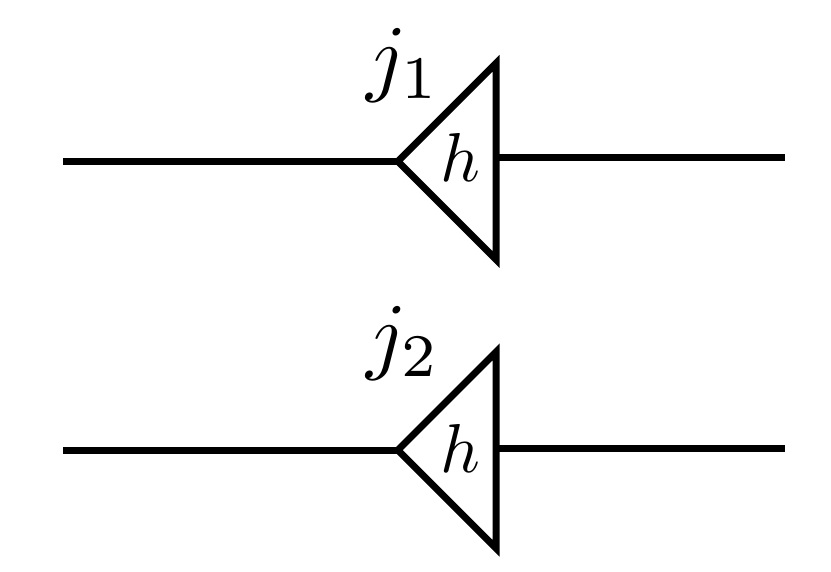}}=\sum_{J=|j_1-j_2|}^{j_1+j_2}d_J\makeSymbol{\includegraphics[width=0.2\textwidth]{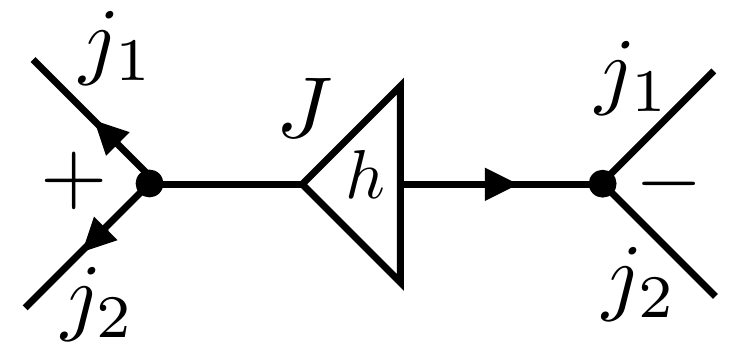}}.
\end{equation}
For the operators $\hat p_s^i(e)$ and  $\hat p_t^j(e)$, by \eqref{eq:ptps} we have
\begin{equation}\label{eq:flux12}
\begin{aligned}
\hat p_s^\alpha(e) D^j_{mn}(h_e)&=-t w_j\epsilon^j_{\tilde m m}\left(\begin{array}{ccc}
j&j&1\\
\tilde m&m'&\alpha
\end{array}
\right)D^j_{m' n}(h_e),\\
\hat p_t^\alpha(e) D^j_{mn}(h_e)&=t w_j\epsilon^j_{\tilde n n'} \left(\begin{array}{ccc}
j&j&1\\
\tilde n &n&\alpha
\end{array}
\right)D^j_{m n'}(h_e),
\end{aligned}
\end{equation}
where we used \eqref{eq:tauj}. Thus , one has
\begin{equation}\label{eq:flux1}
p^\alpha_s(e)\makeSymbol{\includegraphics[width=0.2\textwidth]{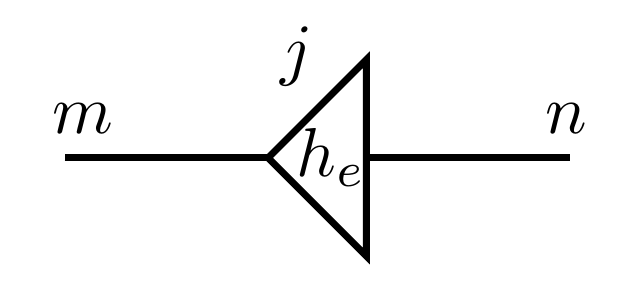}}=-tw_j\makeSymbol{\includegraphics[width=0.2\textwidth]{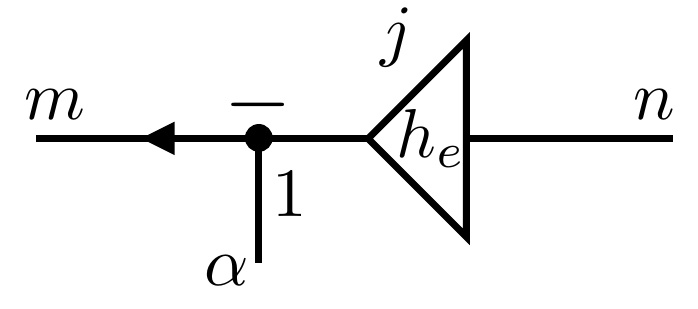}}
\end{equation}
\begin{equation}\label{eq:flux2}
p^\alpha_t(e)\makeSymbol{\includegraphics[width=0.2\textwidth]{dhe}}=tw_j\makeSymbol{\includegraphics[width=0.2\textwidth]{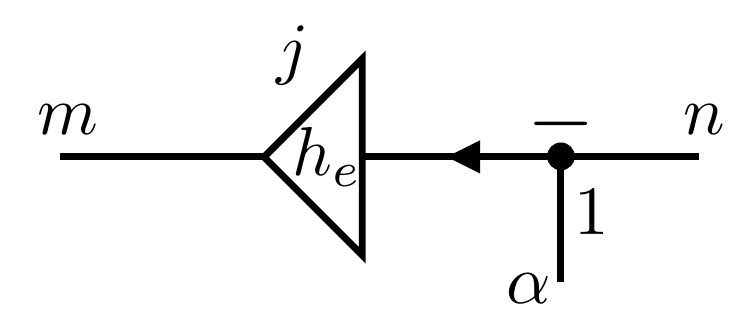}}
\end{equation}
%Because
%\begin{equation}
%\begin{aligned}
%&\int \dd h D^j_{m_1n_1}(h)D^j_{m_2n_2}(h)=(-1)^{m_1-n_1}\int \dd h \overline{D^j_{-m_1,-n_1}(h)}D^j_{m_2n_2}(h)\\
%=&\frac{1}{d_j}(-1)^{m_1-n_1}\delta(m_2,-m_1)\delta(n_2,-n_1)=\frac{1}{d_j}\epsilon^j_{m_2m_1}\epsilon^j_{n_2n_1},
%\end{aligned}
%\end{equation}
%we have
%\begin{equation}
%\int\dd h\makeSymbol{\includegraphics[width=0.2\textwidth]{dhdh}}=\frac{\delta(j_1,j_2)}{d_{j_1}}\makeSymbol{\includegraphics[width=0.1\textwidth]{ep2}}
%\end{equation}
%In the graphical method, the cohere state $\psi_{g_e}$ with $g_e=n_e^s e^{iz_e\tau_3}(n_e^t)^{-1}$ can be represented as 
%\begin{equation}
%\psi^ {t}_{g_e}(h_e)=\sum_j d_j e^{-\frac{t}{2} j(j+1)}\chi_j(g_e h_e^{-1})=\sum_j d_j e^{-\frac{t}{2} j(j+1)}\makeSymbol{\includegraphics[width=0.2\textwidth]{chij}}
%\end{equation}
%where we introduce  an additional graph 
%\begin{equation}
%\makeSymbol{\includegraphics[width=0.2\textwidth]{ztau3}}:=D^{j}_{mn}(e^{iz\tau_3})=e^{zm}\delta^m_{n}. 
%\end{equation}
%Moreover, because of the equation
%\begin{equation}
%\begin{aligned}
%&\overline{\chi_j(g_eh_e^{-1})}=\sum_{l}e^{\overline{z_e}l}D^j_{mn}(n_e^s)\epsilon^j_{ln}\epsilon^j_{ln'}D^j_{n'k}((n_e^t)^{-1})D^j_{km}(h_e^{-1}),
%\end{aligned}
%\end{equation}
%one can obtain that
%\begin{equation}
%\overline{\psi^t_{g_e}}(h_e)=\sum_j d_j e^{- \frac{t}{2} j(j+1)}\makeSymbol{\includegraphics[width=0.2\textwidth]{conchij}}
%\end{equation}
%

Given a monomial of holonomies and fluxes $\hat F_e$. According to  Eqs. \eqref{eq:coupleholonomyp}, \eqref{eq:flux1} and \eqref{eq:flux2}, we draw $\hat F_e$ graphically as
\begin{equation}
\int \dd\mu_h\overline{D^{j'}_{m'n'}(h)}\hat F_eD^j_{mn}(h)=\frac{1}{d_{j'}}\makeSymbol{\includegraphics[width=0.15\textwidth]{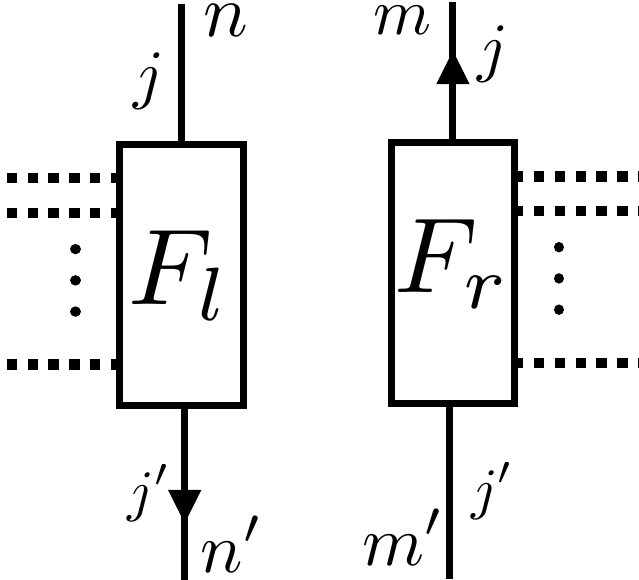}},
\end{equation}
where the dashed lines represents the indices possessed by $\hat F_e$. Some dashed lines may carry arrows depending on the form of $\hat F_e$. With this formula,  
%the action of $\hat F_e$ on $\chi_j(g_e h_e^{-1})$ reads 
%\begin{equation}\label{eq:defineF}
%\hat F_e\makeSymbol{\includegraphics[width=0.2\textwidth]{chij}}=\sum_J\makeSymbol{\includegraphics[width=0.2\textwidth]{fchij}}.
%\end{equation}
%Then, we have 
%\begin{equation}\label{eq:fchi}
%\begin{aligned}
%&\int \dd h_e\overline{\chi_{j'}(g_eh_e^{-1})}\hat F_e\chi_{j}(g_eh_e^{-1})=\frac{1}{d_{j'}}\makeSymbol{\includegraphics[width=0.25\textwidth]{expectF2}}=\frac{1}{d_{j'}}\makeSymbol{\includegraphics[width=0.25\textwidth]{expectF4}}.
%\end{aligned}
%\end{equation}
% In the last graph, the connections of the triangular boxes depend on whether the dashed lines possess arrows or not. To explain it more precisely, let us take an example. If the block $F_r$ is
%\begin{equation}
%\makeSymbol{\raisebox{0.0\height}{\includegraphics[width=0.12\textwidth]{blockF}}}=\makeSymbol{\raisebox{0.0\height}{\includegraphics[width=0.1\textwidth]{blockF2}}}, 
%\end{equation}
%then
%\begin{equation}
%\makeSymbol{\raisebox{0.0\height}{\includegraphics[width=0.13\textwidth]{blockF4}}}=\makeSymbol{\raisebox{0.0\height}{\includegraphics[width=0.1\textwidth]{blockF3}}}.
%\end{equation}
%For other cases with different $F_r$ (or $F_l$ ), \eqref{eq:fchi} can be obtained analogously.
%Finally, 
the expectation value of $\langle\hat F_e\rangle_{z_e}$ 
%with respect to the coherent state, according to \eqref{eq:fchi}, 
is
\begin{equation}
\begin{aligned}
&\langle \hat F_e\rangle_{z_e}=\sum_{jj'}d_je^{- \frac{t}{2}(j(j+1)+j'(j'+1))}\makeSymbol{\includegraphics[width=0.25\textwidth]{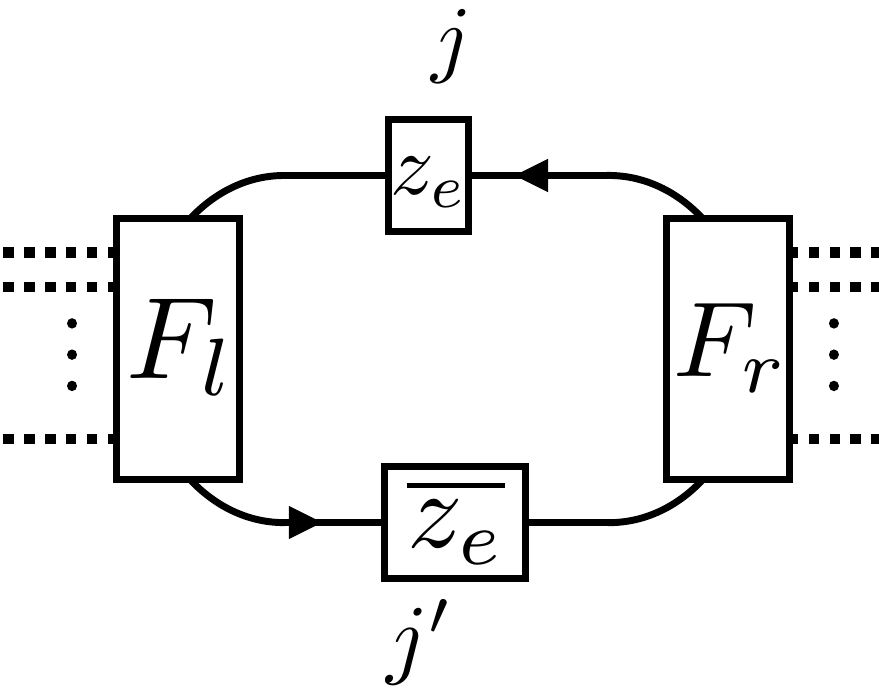}}.
\end{aligned}
\end{equation}
%which tells us that $n_e^s$ and $n_e^t$, whose information is entirely encoded within the triangular boxes in the last graph in equation \eqref{eq:basicformula}, can be “tore” apart from the rest of the portion. According to the graphical dictionary, the rest of the portion can be translated to the expectation value of the same operator with respect to the coherent state $\psi^t_{e^{iz_e\tau_3}}=:\psi^t_{z_e}$. These graphical features implies that one can always do the calculation with respect to $\psi^t_{z_e}$, then restore the information of $n_e^s$ and $n_e^t$ afterwards, which can be formulated according to \eqref{eq:basicformula} as
%\begin{equation}
%\begin{aligned}
%&\langle \psi^t_{g_e}|P(\{\hat p_s^{\alpha_i}(e)\},\{\hat p_t^{\alpha_j}(e)\},\{D^{\iota_k}_{a_k b_k}(h_e)\})|\psi^t_{g_e}\rangle\\
%=&\langle \psi^t_{z_e}|P(\{\hat p_s^{\beta_i}(e)D^1_{\beta_i\alpha_i}((n_e^s)^{-1})\},\{\hat p_t^{\beta_j}(e)D^1_{\beta_j\alpha_j}((n_e^t)^{-1})\},\{D^{\iota_k}_{a_k c_k}(n_e^s)D^{\iota_k}_{c_k d_k}(h_e)D^{\iota_k}_{d_k b_k}((n_e^t)^{-1})\})|\psi^t_{z_e}\rangle,
%\end{aligned} 
%\end{equation}
%with $P(\{\hat p_s^{\alpha_i}(e)\},\{\hat p_t^{\alpha_j}\},\{D^{\iota_k}_{a_k b_k}(h_e)\})$ be a monomial of the flux and holonomy operators.
%In the following context, we will denote 
%\begin{equation}
%\langle \psi^t_{z_e}|\hat F_e|\psi^t_{z_e}\rangle=:\langle \hat F_e\rangle_{z_e}. 
%\end{equation}
%
To prove \eqref{eq:pspth}, we claim that
\begin{equation}
\begin{aligned}
\makeSymbol{\raisebox{0.1\height}{\includegraphics[width=0.3\textwidth]{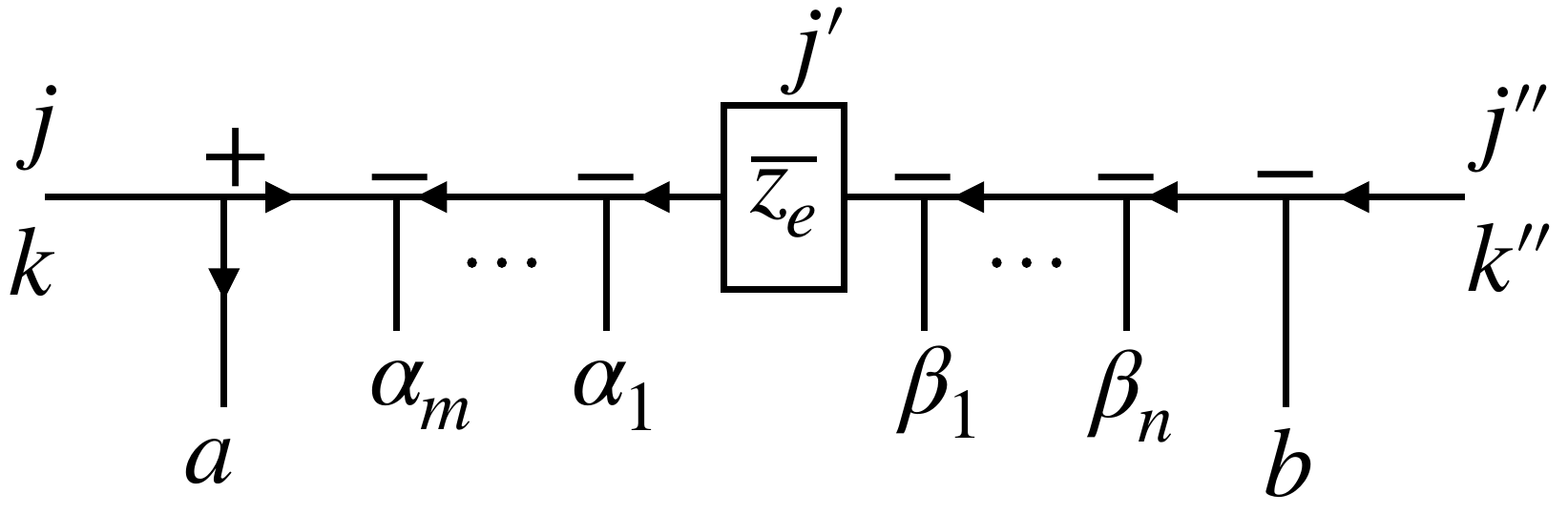}}}=e^{(k-a+\alpha_1+\cdots+\alpha_m)\overline{z_e}}\makeSymbol{\raisebox{0\height}{\includegraphics[width=0.3\textwidth]{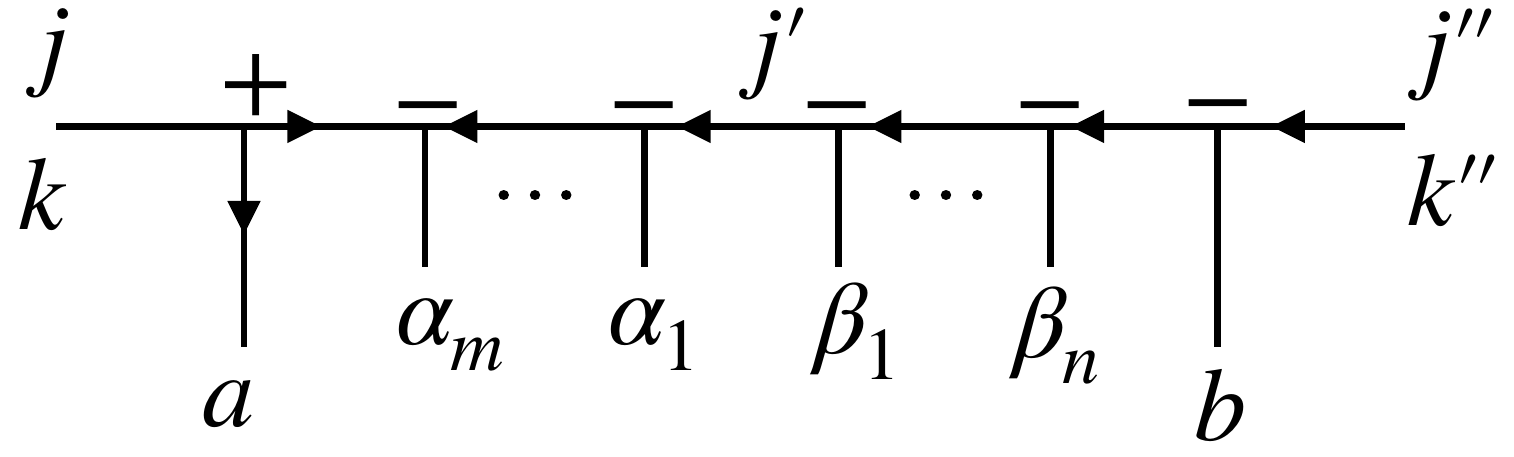}}},
\end{aligned}
\end{equation}
which can be verified easily with the fact that
$
D^{j}_{mn}(e^{iz\tau_3})=e^{z m}\delta_{mn},
$
and the non-vanishing condition of the 3$j$-symbol that
$
\left(
\begin{array}{ccc}
j_1&j_2&j_3\\
m_1&m_2&m_3
\end{array}
\right)\propto\delta_{m_1+m_2+m_3,0}$. 

Consequently, it has
\begin{equation}
\begin{aligned}
\makeSymbol{\raisebox{0.0\height}{\includegraphics[width=0.3\textwidth]{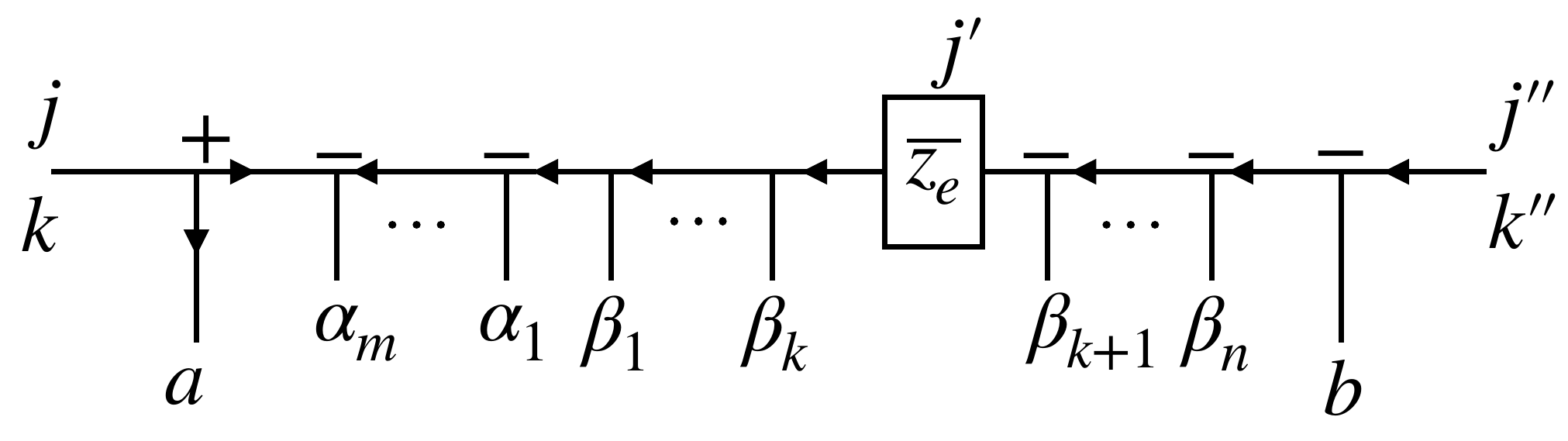}}}=e^{(\beta_1+\cdots\beta_k)\overline{z_e}}\makeSymbol{\raisebox{0.0\height}{\includegraphics[width=0.3\textwidth]{pspth1}}},
\end{aligned}
\end{equation}
which leads to
\begin{equation}\label{eq:pspt}
\begin{aligned}
\makeSymbol{\raisebox{0.0\height}{\includegraphics[width=0.3\textwidth]{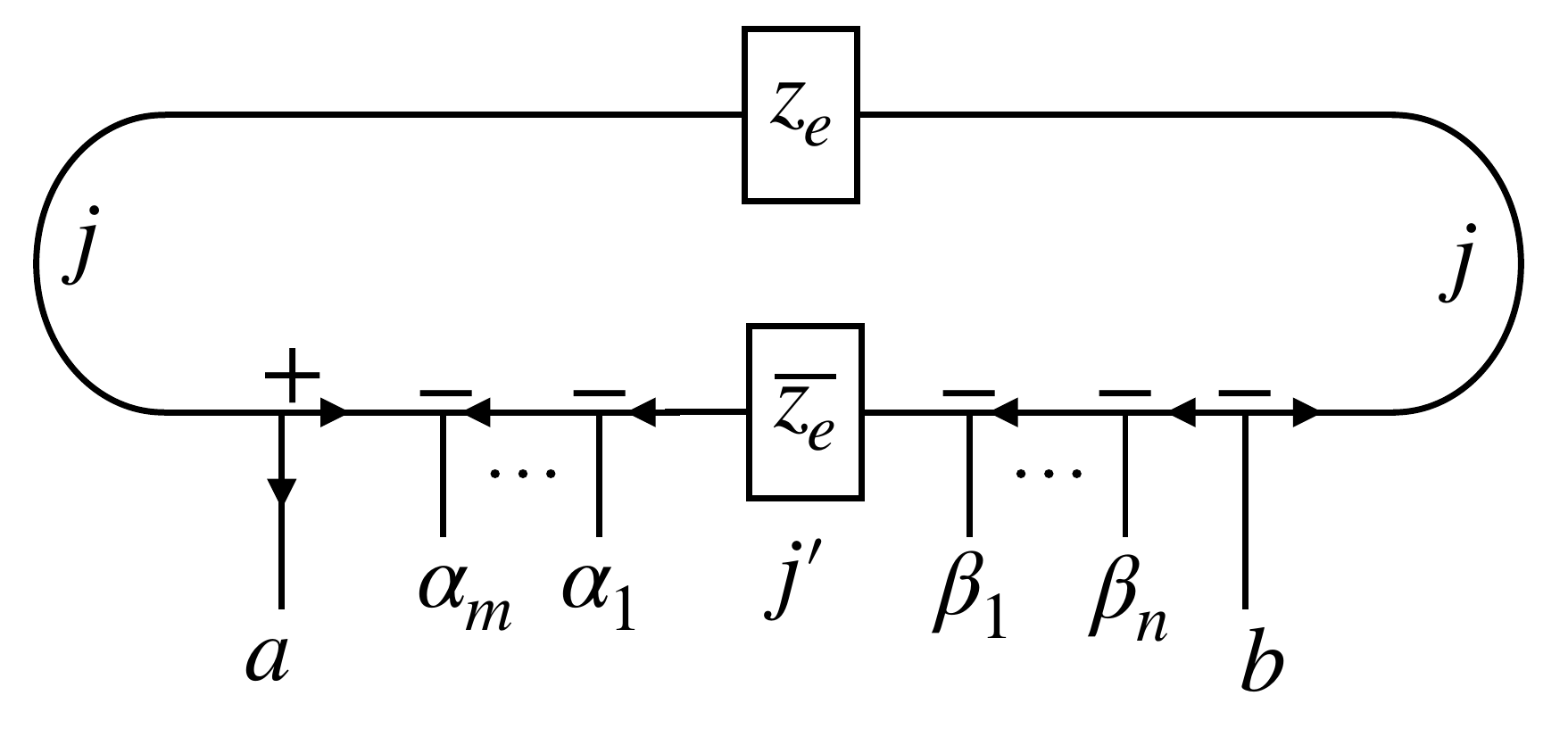}}}=e^{-(\beta_1+\cdots\beta_n)\overline{z_e}}\makeSymbol{\raisebox{0.0\height}{\includegraphics[width=0.3\textwidth]{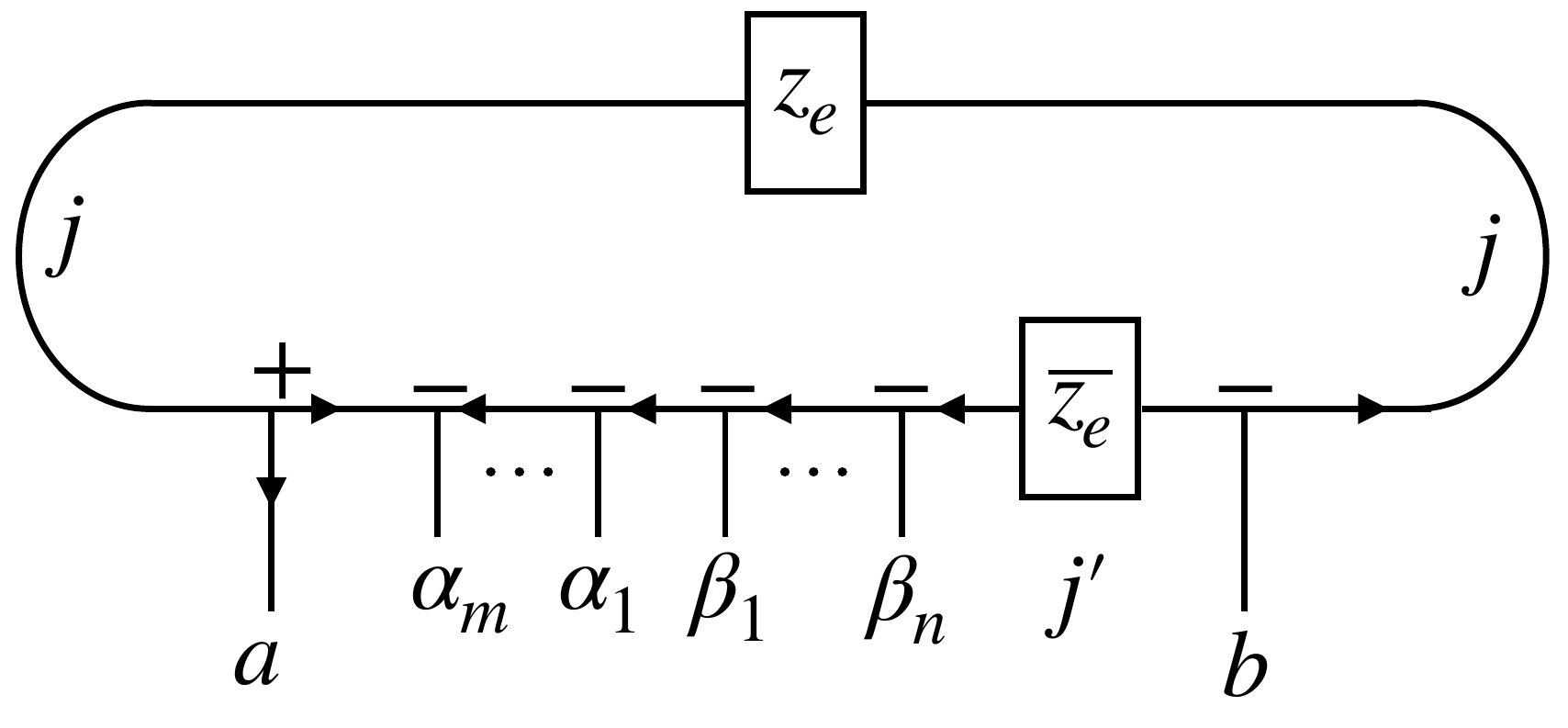}}}.
\end{aligned}
\end{equation}
This graphical equation can be decoded as  \eqref{eq:pspth}, where the factor $(-1)^n$ in the right hand side is because of the minus sign in the definition of $\hat p_t^\alpha(e)$. 

\section{Expectation value of $\hat F^{\alpha_1\cdots \alpha_m}_{\iota a b}$}\label{app:shabiref}
\subsection{ for the  case with $\iota=0$}
Let us consider 
\begin{equation}
\hat F^{\alpha_1\cdots \alpha_m}_{000}\equiv \hat F^{\alpha_1\cdots \alpha_m}=\hat p_s^{\alpha_1}(e)\cdots \hat p_s^{\alpha_m}(e).
\end{equation}
 Eq. \eqref{eq:flux12} gives us that
%\begin{equation}
%\begin{aligned}
%\hat F^{\alpha_1\cdots \alpha_m}\makeSymbol{\includegraphics[width=0.2\textwidth]{chij}}=(- t w_j)^m\makeSymbol{\includegraphics[width=0.2\textwidth]{pjchij}}.
%\end{aligned}
%\end{equation}
%Therefore, the expectation value is
\begin{equation}
\begin{aligned}
&\langle \hat F^{\alpha_1\cdots \alpha_m}\rangle_{z_e}=t^m\sum_{j\geq 1/2}d_je^{- tj(j+1)} F_0(j,\frac{\partial_\eta}{2})\frac{\sinh((2j+1)\eta)}{\sinh(\eta)} 
\end{aligned}
\end{equation}
where the function $F_0$ is given by
\begin{equation}
\begin{aligned}
F_0(j,k)=&(-w_j)^m\epsilon^j_{kk_1}\left(
\begin{array}{ccc}
j&1&j\\
k_1&\alpha_1&n_1
\end{array}
\right)\epsilon^j_{n_1k_2}\left(
\begin{array}{ccc}
j&1&j\\
k_2&\alpha_2&n_2
\end{array}
\right)\cdots \\
&\epsilon_{n_{m-2}k_{m-1}}^j\left(
\begin{array}{ccc}
j&1&j\\
k_{m-1}&\alpha_{m-1}&n_{m-1}
\end{array}
\right) 
\epsilon_{n_{m-1}k_{m}}^j\left(
\begin{array}{ccc}
j&1&j\\
k_{m}&\alpha_{m}&k
\end{array}
\right)
\end{aligned}
\end{equation}
%with $\left(\begin{array}{ccc} j_1&j_2&j_3\\k_1&k_2&k_3\end{array}\right)$ being the 3-$j$ symbol, and $\epsilon^j_{mn}$ being the 2-$j$ symbol (cf. Appendix \ref{app:sl2csu2}). 
By  Eqs. \eqref{eq:CG3j} and \eqref{eq:negativeCG}, we achieve
\begin{equation}\label{eq:j1j}
\begin{aligned}
\left.w_j\epsilon^{j}_{\tilde nk}\left(
\begin{array}{ccc}
j&1&j\\
k&\alpha&n
\end{array}
\right)\right|_{j\to-(j+1)}=(-1)^{\alpha}w_j\epsilon^j_{\tilde n k}\left(
\begin{array}{ccc}
j&1&j\\
k&\alpha&n
\end{array}
\right),
\end{aligned}
\end{equation}
which results in
\begin{equation}
F_0(j,k)=F_0(-j-1,k).
\end{equation}
Therefore
\begin{equation}\label{eq:FwithF0}
\langle\hat F^{\alpha_1\cdots\alpha_m}\rangle_{z_e}=t^m\sum_{n\in\mathbb Z} e^{-\frac{t}{4}n^2+\frac{t}{4}} n F_0(\frac{n-1}{2},\frac{\partial_\eta}{2})\frac{e^{n\eta}}{2\sinh(\eta)}
\end{equation}
where $n\equiv 2j+1$. 
Substituting the values of the 3-$j$ and 2-$j$ symbols, we have that
\begin{equation}\label{eq:multiflux}
\begin{aligned}
F_0(j,\frac{\partial_\eta}{2})=\delta(\sum_{i=1}^m\alpha_i,0)\left(\prod_{i=1}^m\frac{1}{(1+|\alpha_i|)^{1/2}}\right)\prod_{n=1}^m(\alpha_nj-\frac{\partial_\eta}{2}+\sum_{k=1}^n\alpha_k)
\end{aligned}
\end{equation} 
Applying the Poisson summation formula to Eq. \eqref{eq:FwithF0}, we get
\begin{equation}\label{eq:qm0}
\begin{aligned}
&\langle\hat F^{\alpha_1\cdots \alpha_m}\rangle_{z_e}\\
=&\delta\left(\sum_{i}\alpha_i,0\right)t^m\prod_{i=1}^m\frac{1}{(1+|\alpha_i|)^{1/2}}e^{t/4}\int_{-\infty}^\infty\dd x\, e^{-\frac{t}{4}x^2}x\prod_{n=1}^m(\alpha_n\frac{x-1}{2}-\frac{\partial_\eta}{2}+\sum_{k=1}^n\alpha_k)\frac{e^{x\eta}}{2\sinh(\eta)}+O(t^\infty). 
\end{aligned}
\end{equation}
By the trick
\begin{equation}
\partial_\eta^n\frac{e^{x\eta}}{\sinh(\eta)}=\left.\left((x+\partial_y)^n\frac{e^{x\eta}}{\sinh(y)}\right)\right|_{y\to\eta},
\end{equation}
Eq. \eqref{eq:qm0} can be finally simplified as
\begin{equation}\label{eq:qm}
\begin{aligned}
&\langle\hat F^{\alpha_1\cdots \alpha_m}\rangle_{z_e}\\
=&\delta\left(\sum_{i}\alpha_i,0\right)t^m\prod_{i=1}^m\frac{1}{(1+|\alpha_i|)^{1/2}}e^{t/4}\int_{-\infty}^\infty\dd x x\prod_{k=1}^m\left(\frac{\alpha_k-1}{2}x-\frac{\partial_y}{2}+\sum_{i=1}^k\alpha_i-\frac{\alpha_k}{2}\right)\frac{e^{-\frac{t}{4}x^2+x\eta}}{2\sinh(y)}\Big|_{y\to\eta}+O(t^\infty)\\
\end{aligned}
\end{equation}
The integrand is a  polynomial of $x$ multiplied by a Gaussian function $e^{-t x^2/4+x\eta}$. Thus let us integrate $x^ne^{-t x^2/2+x\eta}$
\begin{equation}\label{eq:intxn}
\begin{aligned}
&\int_{-\infty}^\infty \dd x x^n e^{-\frac{t}{4}x^2+\eta x}=\frac{t}{ e^{t/4}}\frac{\sinh(\eta)}{\eta }\langle 1\rangle \sum_{k=0}^{\floor{n/2}} \binom{n}{2k}\eta^{n-2k}\left(\frac{t}{2}\right)^{k-n}(2k-1)!!\end{aligned}
\end{equation}
where $\floor{y}$ denotes the greatest integer less than or equal to $y$ {and Eq. \eqref{eq:norm} is substituted because we are concerned on the expectation value with respect to the normalized coherent states, i.e. $\langle\hat F^{\alpha_1\cdots \alpha_m}\rangle_{z_e}/\langle 1\rangle$}. The leading-order term of the RHS of  Eq. \eqref{eq:intxn} is $O(t^{-n+1})$ with $n$ the power of $x$ in the integrand. Thus it is the highest-power term of $x$ in the integrand of \eqref{eq:qm} that leads to the leading-order term of $\langle\hat F^{\alpha_1\cdots \alpha_m}\rangle_{z_e}$, and the highest-power and second-highest-power terms of $x$ lead to the next-to-leading order term of $\langle\hat F^{\alpha_1\cdots \alpha_m}\rangle_{z_e}$, and so on. 
Therefore, we finally get
\begin{equation}
\begin{aligned}
&\langle\hat F^{\alpha_1\cdots \alpha_m}\rangle_{z_e}\\
=&\delta\left(\sum_{i}\alpha_i,0\right)\langle 1\rangle\prod_{i=1}^m\frac{1}{(1+|\alpha_i|)^{1/2}}\Bigg\{\left(\prod_{k=1}^m(\alpha_k-1)\right)\eta^{m}+\left(\prod_{k=1}^m(\alpha_k-1)\right)\frac{(m+1)m}{2}\eta^{m-2}\frac{t}{2}+\\
&+\sum_{k=1}^m\left(\prod_{i\neq k}(\alpha_i-1)\right)\left(\frac{\coth(\eta)+\alpha_k}{2}+\alpha_1+\cdots+\alpha_{k-1}\right)\eta^{m-1} t+\cdots\Bigg\}.
\end{aligned}
\end{equation} 
\subsection{for the case with $\iota\neq 0$}
Now let us consider the operator 
\begin{equation}
\hat F_{\iota a b}^{\alpha_1\cdots \alpha_n}=\hat p_s^{\alpha_1}(e)\cdots \hat p_s^{\alpha_m}(e)D^\iota_{ab}(h_e),\  \iota\neq 0.
\end{equation}
%Graphically we have
%\begin{equation}
%\begin{aligned}
%&\hat F_{\iota a b}^{\alpha_1\cdots \alpha_n}\makeSymbol{\raisebox{0.0\height}{\includegraphics[width=0.2\textwidth]{chij}}}=\sum_{j'=|j-\iota|}^{j+\iota}d_{j'}(-tw_{j'})^m\makeSymbol{\raisebox{0.0\height}{\includegraphics[width=0.25\textwidth]{phchijv1}}}.
%\end{aligned}
%\end{equation}
By using \eqref{eq:flux12} and \eqref{eq:coupleholonomy}, it can be obtained that
\begin{equation}\label{eq:oneholonomy}
\begin{aligned}
&\langle \hat F_{\iota a b}^{\alpha_1\cdots\alpha_m}\rangle_{z_e}=\sum_{jj'}d_je^{-\frac{t}{2}(j(j+1)+j'(j'+1))}d_{j'}(-tw_{j'})^m\sum_{kk'}e^{kz_e+k'\overline{z_e}}\epsilon^{j'}_{k'k_1}\left(
\begin{array}{ccc}
j'&1&j'\\
k_1&\alpha_1&n_1
\end{array}
\right)\epsilon^{j'}_{n_1k_2}\left(
\begin{array}{ccc}
j'&1&j'\\
k_2&\alpha_2&n_2
\end{array}
\right)\cdots\\
& \left(
\begin{array}{ccc}
j'&1&j'\\
k_{m-1}&\alpha_{m-1}&n_{m-1}
\end{array}
\right)\epsilon^{j'}_{n_{m-1}k_m}\left(
\begin{array}{ccc}
j'&1&j'\\
k_m&\alpha_m&n_m
\end{array}
\right)\epsilon^{j'}_{n_mk_{m+1}}\left(
\begin{array}{ccc}
j'&\iota&j\\
k_{m+1}&a'&k
\end{array}
\right)\epsilon^\iota_{aa'}\left(
\begin{array}{ccc}
j'&\iota&j\\
k'&b&k''
\end{array}
\right)\epsilon^j_{k''k}
\end{aligned}
\end{equation}
Because the products of these $3j$-symbols vanish unless $k=k'+b$, the last equation can be simplified as
\begin{equation}\label{eq:qmhgeneral}
\begin{aligned}
\langle \hat F^{\alpha_1\cdots\alpha_m}_{\iota a b}\rangle_{z_e}=t^me^{bz_e}\sum_{j,j'}e^{-\frac{t}{2}(j(j+1)+j'(j'+1))}F_\iota(j,j',\frac{\partial_\eta}{2})\frac{\sinh((2j'+1)\eta)}{\sinh(\eta)}
\end{aligned}
\end{equation} 
where $F_\iota$ is given by
\begin{equation}\label{eq:3jF}
\begin{aligned}
F_\iota(j,j',k'):=&d_j d_{j'}(-w_{j'})^m\epsilon^{j'}_{k'k_1}\left(
\begin{array}{ccc}
j'&1&j'\\
k_1&\alpha_1&n_1
\end{array}
\right)\epsilon^{j'}_{n_1k_2}\left(
\begin{array}{ccc}
j'&1&j'\\
k_2&\alpha_2&n_2
\end{array}
\right)\cdots \left(
\begin{array}{ccc}
j'&1&j'\\
k_{m-1}&\alpha_{m-1}&n_{m-1}
\end{array}
\right)\epsilon^{j'}_{n_{m-1}k_m}\\
&\left(
\begin{array}{ccc}
j'&1&j'\\
k_m&\alpha_m&n_m
\end{array}
\right)\epsilon^{j'}_{n_mk_{m+1}}\left(
\begin{array}{ccc}
j'&\iota&j\\
k_{m+1}&a'&k'+b
\end{array}
\right)\epsilon^\iota_{aa'}\left(
\begin{array}{ccc}
j'&\iota&j\\
k'&b&k''
\end{array}
\right)\epsilon^j_{k''\,k'+b}.
\end{aligned}
\end{equation}
The possible values of $j$ are $|j'-\iota|,\ |j'-\iota|+1,\ \cdots,\ j'+\iota$.   By applying Eqs. \eqref{eq:CG3j} and \eqref{eq:negativeCG} again, we have
\begin{equation}
\begin{aligned}
&\epsilon^{j'}_{n_mk_{m+1}}\left(
\begin{array}{ccc}
j'&\iota&j\\
k_{m+1}&a'&k'+b
\end{array}
\right)\left(
\begin{array}{ccc}
j'&\iota&j\\
k'&b&k''
\end{array}
\right)\epsilon^j_{k''\,k'+b}\Bigg|_{j=j'-\Delta}\\
=&(-1)^{b-a+1}\left[\epsilon^{j'}_{n_mk_{m+1}}\left(
\begin{array}{ccc}
j'&\iota&j\\
k_{m+1}&a'&k'+b
\end{array}
\right)\left(
\begin{array}{ccc}
j'&\iota&j\\
k'&b&k''
\end{array}
\right)\epsilon^j_{k''\,k'+b}\Bigg|_{j=j'+\Delta}\right]_{j'\to -j'-1}.
\end{aligned}
\end{equation}
Together with Eq. \eqref{eq:j1j}, it leads to
\begin{equation}\label{eq:Fjpmd}
F_{\iota}(j'+\Delta,j',k')=-F_{\iota}(j'-\Delta,j',k')\Big|_{j'\to -j'-1}
\end{equation}
where we use $F_\iota(j,j',k) \propto \delta(\sum_{i=1}^m\alpha_i-a+b,0)$.
Thus for $j'=j\pm d$ with $d>0$, one has
\begin{equation*}
\begin{aligned}
&\sum_{2j'+d\geq  \iota}f(d_{j'}^2) F_\iota(j'+d,j',\frac{\partial\eta}{2})\frac{\sinh(d_{j'}\eta)}{\sinh(\eta)}+\sum_{2j'-d\geq \iota}f(d_{j'}^2)F_\iota(j'-d,j',\frac{\partial\eta}{2})\frac{\sinh(d_{j'}\eta)}{\sinh(\eta)}\\
=&\sum_{-2j'-2+d\geq  \iota}f(d_{j'}^2)F_\iota(-j'-1+d,-j'-1,\frac{\partial\eta}{2})\frac{\sinh(-d_{j'}\eta)}{\sinh(\eta)}+\sum_{2j'-d\geq \iota}f(d_{j'}^2)F_\iota(j'-d,j',\frac{\partial\eta}{2})\frac{\sinh(d_{j'}\eta)}{\sinh(\eta)}\\
=&\sum_{\substack{n\in \mathbb Z\\ n\notin [d-\iota,d+\iota] }}f(d_{j'}^2)F_\iota(\frac{n-1}{2}-d,\frac{n-1}{2},\frac{\partial_\eta}{2})\frac{\sinh(n\eta)}{\sinh(\eta)}.
\end{aligned}
\end{equation*}
Here the conditions $2j'\pm d\geq \iota$ are from the the triangle condition i.e. $j'+j\geq \iota$, of the $3j$-symbol. The second equality is because of the replacement $j'\to -j'-1$ in the first summation,  $d_{j'}=2j'+1\equiv n$ and $f$ is an arbitrary function. 
Therefore,
\begin{equation}\label{eq:mergedqmh}
\langle \hat F^{\alpha_1\cdots\alpha_m}_{\iota a b}\rangle_{z_e}=t^m e^{bz_e}\sum_{\substack{0\leq d\leq \iota\\ d+\iota\in \mathbb Z}}\frac{2-\delta(d,0)}{2} e^{-\frac{t}{4} \left(2 d ^2-1\right)}\sum_{\substack{ n\in\mathbb Z\\ n\notin[d-\iota,d+\iota]}}e^{-\frac{1}{4} t \left(n^2-2 d n\right)}F_\iota(\frac{n-1}{2}-d,\frac{n-1}{2},\frac{\partial_\eta}{2})\frac{\sinh(n\eta)}{\sinh(\eta)}
\end{equation}

To calculate $F_\iota(j,j',k')$, we notice that
\begin{equation}\label{eq:factor1}
\begin{aligned}
-w_{j'}\epsilon^{j'}_{n_{i-1} k_i}\left(
\begin{array}{ccc}
j'&1&j'\\
k_i&\alpha_i&n_i
\end{array}
\right)=\delta(k_i+\alpha_i+n_i,0)\left(\alpha _i j'-n_i\right) \sqrt{\frac{\left(j'-n_{i-1}\right)! \left(n_{i-1}+j'\right)!}{\left(1-\alpha _i\right)! \left(\alpha _i+1\right)! \left(j'-n_i\right)! \left(n_i+j'\right)!}}.
\end{aligned}
\end{equation}
Moreover, by applying Eq. \eqref{eq:CG3j} and the expression of the Clebsch-Gordan coefficients (Eq. (5), Section 8.2.1 in \cite{varshalovich1988quantum})
, we have
\begin{equation}\label{eq:factor2}
\begin{aligned}
&\epsilon^{j'}_{n_m k_{m+1}}\left(
\begin{array}{ccc}
j'&\iota&j\\
k_{m+1}&a'&k
\end{array}
\right)\left(
\begin{array}{ccc}
j'&\iota&j\\
k'&b&k''
\end{array}
\right)\epsilon^j_{k''\,k}=(-1)^{j'-\iota-j+a'+b}\delta_{a'+k,n_m}\delta_{k'+b,k}\times\\
&\frac{(j+\iota-j')!(j-\iota+j')!(-j+\iota+j')!}{(j+\iota+j'+1)!}\left(\frac{(j'+n_m)!(j'-n_m)!}{(j'+k')!(j'-k')!(\iota+a')!(\iota-a')!(\iota+b)!(\iota-b)!}\right)^{1/2}\\
&\sum_z\frac{(-1)^{z}(j'+\iota+k-z)!(j-k+z)!}{z!(j'-j+\iota-z)!(j'+k+a'-z)!(j-\iota-k-a'+z)!}\sum_z\frac{(-1)^{z}(j+\iota+k'-z)!(j'-k'+z)!}{z!(j-j'+\iota-z)!(j+k-z)!(j'-\iota-k+z)!}.
\end{aligned}
\end{equation}
 Eqs. \eqref{eq:factor1} and \eqref{eq:factor2} lead to
\begin{equation}\label{eq:algebraicF}
\begin{aligned}
F_\iota(j'-d,j',k')=&\delta(\sum_{i=1}^m\alpha_i-a+b,0)\prod_{i=1}^m\frac{1}{(1+|\alpha_i|)^{1/2}}\left(\frac{1}{(\iota+a)!(\iota-a)!(\iota+b)!(\iota-b)!}\right)^{1/2}\\
&\frac{(-1)^{d-2a+b}(2j'+1)(2j'-2d+1)}{(2j'-d+\iota+1)_{2\iota+1}}\prod_{n=1}^m(\alpha_i j'-k'+\sum_{i=1}^n\alpha_i)\\
&\sum_z\frac{(-1)^{z}(d+\iota)_{z}(j'+k'+b-z+\iota)_{\iota+a}(j'-d-k'-b+z)_{\iota-a}}{z!}\\
&\sum_z\frac{(-1)^{z}(\iota-d)_z(j'-d+k'-z+\iota)_{\iota-b}(j'-k'+z)_{\iota+b}}{z!}
\end{aligned}
\end{equation}
where $(x)_n:=x(x-1)\cdots(x-n+1)$ is the falling factorial. 

There is subtle issue on the summation over $k$ and $k'$ in Eq. \eqref{eq:oneholonomy}. In order to obtain the factor $\sinh((2j'+1)\eta)/\sinh(\eta)$ in \eqref{eq:qmhgeneral}, one has to consider a summation of $e^{2k'\eta}$ over all values of $k'$ in $[-j',j']$.
However, except the constraint of $|k'|\leq j'$ for $k'$, there exists another condition of $|k'+b|=|k|\leq j$ implied by $k=k'+b$, which seemingly narrows the range of $k'$.
In order to preserve the range of $k'$, one can always ``artificially engineer"  $F_\iota(j,j',k')$ such that it vanishes for $|k'+b|>j$, which has actually been done by the triangle inequality of $3j$-symbols in the RHS Eq. \eqref{eq:3jF}. Thus the subtlety is now encoded in the condition that $F_\iota(j,j',k')$ vanishes for $|k'+b|>j$. We need to verify this for the algebraic expression of  $F_\iota(j,j',k')$ in the RHS of Eq. \eqref{eq:algebraicF}. To check it, assume $k'+b=j+\delta=j'-d+\delta$ with $0<\delta\leq d+b$. Then $(j'-k'+z)_{\iota+b}$, in the second summand over $z$, vanishes because
\begin{equation}
j'-k'+z\geq d+b-\delta\geq 0,\ j'-k'+z-(\iota+b)+1\leq \delta<0
\end{equation}
where $0\leq z\leq \iota-d$ is applied.
Therefore, we conclude that the RHS of Eq.\eqref{eq:algebraicF} vanishes for $j-b<k'\leq j'$.
Similarly, it can be checked that $(j'-d+k'-z+\iota)_{\iota-b}$ vanishes, indicating the vanishing of the RHS of Eq.\eqref{eq:algebraicF} for $-j'\leq k'+b<-j-b$.
We thus claim that the RHS of Eq.\eqref{eq:algebraicF} vanishes for $|k'+b|>j$, which indicates that the expression of the RHS of Eq.\eqref{eq:algebraicF} can be analytically extended to give $F_\iota(j,j',k')$ for the full range of $k'$, i.e. $|k'|\leq j$.

Replacing $j'$ and $k'$ in Eq. \eqref{eq:algebraicF} by $\frac{n-1}{2}$ and $\frac{\partial_\eta}{2}$ respectively, we have 
\begin{equation}\label{eq:algebraicF1}
\begin{aligned}
F_\iota(\frac{n-1}{2}-d,\frac{n-1}{2},\frac{\partial_\eta}{2})=&\delta(\sum_{i=1}^m\alpha_i-a+b,0)\prod_{i=1}^m\frac{1}{(1+|\alpha_i|)^{1/2}}\left(\frac{1}{(\iota+a)!(\iota-a)!(\iota+b)!(\iota-b)!}\right)^{1/2}\\
&\frac{(-1)^{d-2a+b}n(n-2d)}{(n-d+\iota)_{2\iota+1}}\prod_{k=1}^m(\alpha_i \frac{n-1}{2}-\frac{\partial_\eta}{2}+\sum_{i=1}^k\alpha_i)\\
&\sum_{z=0}^{\iota+d}\frac{(-1)^{z}(\iota+d)_{z}(\frac{n-1}{2}+\frac{\partial_\eta}{2}+b-z+\iota)_{\iota+a}(\frac{n-1}{2}-\frac{\partial_\eta}{2}-d-b+z)_{\iota-a}}{z!}\\
&\sum_{z=0}^{\iota-d}\frac{(-1)^{z}(\iota-d)_z(\frac{n-1}{2}+\frac{\partial_\eta}{2}-d-z+\iota)_{\iota-b}(\frac{n-1}{2}-\frac{\partial_\eta}{2}+z)_{\iota+b}}{z!}.
\end{aligned}
\end{equation}
In order to apply the Poisson summation formula to \eqref{eq:mergedqmh} to calculate the summation over $n$, we need to extend the summation to entire $\mathbb Z$, while Eq. \eqref{eq:mergedqmh} sums $n$ over $\mathbb Z\setminus [d-\iota,d+\iota]$.  However, because of the term $(n-d+\iota)_{2\iota+1}$ in the denominator of $F_\iota$, this extension is not trivial. We need to prove that  $F_\iota(\frac{n-1}{2}-d,\frac{n-1}{2},\frac{\partial_\eta}{2})\frac{\sinh(n\eta)}{\sinh(\eta)}$ is well-defined at the points where $F_\iota$ itself does not. 
Once it is proven, we have 
  \begin{equation}
  \begin{aligned}
  &\langle \hat F^{\alpha_1\cdots\alpha_m}_{\iota a b}\rangle_{z_e}=t^me^{bz_e}\sum_{\substack{0\leq d\leq \iota\\ d+\iota\in \mathbb Z}}\frac{2-\delta(d,0)}{2} e^{-\frac{t}{4} \left(2 d ^2-1\right)}\int\dd x\, e^{-\frac{1}{4} t \left(x^2-2 d x\right)}F_\iota(\frac{x-1}{2}-d,\frac{x-1}{2},\frac{\partial_\eta}{2})\frac{\sinh(x\eta)}{\sinh(\eta)}\\
&-\sum_{n\in [d-\iota,d+\iota]\cap \mathbb Z}e^{-\frac{1}{4} t \left(n^2-2 d n\right)}F_\iota(\frac{n-1}{2}-d,\frac{n-1}{2},\frac{\partial_\eta}{2})\frac{\sinh(n\eta)}{\sinh(\eta)}+O(t^\infty).
  \end{aligned}
\end{equation}
Moreover, the integrals in the last equation usually produce a factor $e^{\eta^2/t}$.  Thus the second term given by the summation over $n\in [d-\iota,d+\iota]\cap \mathbb Z$ decays exponentially as $t\to 0$ after normalization. Therefore, we finally obtain
 \begin{equation}\label{eq:expectedgeneral}
  \begin{aligned}
  &\langle \hat F^{\alpha_1\cdots\alpha_m}_{\iota a b}\rangle_{z_e}=t^me^{bz_e}\sum_{\substack{0\leq d\leq \iota\\ d+\iota\in \mathbb Z}}\frac{2-\delta(d,0)}{2} e^{-\frac{t}{4} \left(2 d ^2-1\right)}\int\dd x\, e^{-\frac{1}{4} t \left(x^2-2 d x\right)}F_\iota(\frac{x-1}{2}-d,\frac{x-1}{2},\frac{\partial_\eta}{2})\frac{\sinh(x\eta)}{\sinh(\eta)}+O(t^\infty).
  \end{aligned}
\end{equation}

In order to show that $F_\iota(\frac{n-1}{2}-d,\frac{n-1}{2},\frac{\partial_\eta}{2})\frac{\sinh(n\eta)}{\sinh(\eta)}$ is  well-defined at the poles of $F_\iota$, we have checked that  $F_\iota(\frac{n-1}{2}-d,\frac{n-1}{2},\frac{\partial_\eta}{2})$ for at least $\iota\leq 20$ can be simplified as summation of terms taking the form\footnote{Here, we conjecture that this statement is true for all $\iota$.} 
$$\frac{1}{x-n}f(x,\partial_\eta)\frac{\sinh( x\eta)}{\sinh(\eta)}$$
where $f(x,\partial_\eta)$, a polynomial of $x$ and $\partial_\eta$, satisfies that $f(n,\cdot)$ takes $\partial_\eta=2k$ for all $k\in\mathbb Z\cap[-\frac{n-1}{2},\frac{n-1}{2}]$ as its zeros.  Then, the following theorem will be helpful. 
\begin{thm}\label{thm:extendsummation}
Let $f(x,k)$  be a polynomials of $x$ and $k$ and $n\in \mathbb Z$ be some integer. Let 
\begin{equation}\label{eq:preconditionD1}
h(x):=\frac{1}{x-n}f(x,\partial_\eta)\frac{\sinh( x\eta)}{\sinh(\eta)}.
\end{equation}
Then $h(n):=\lim_{x\to n} h(x)$ is well-defined, i.e., $n$ is a removable singularity of $h(x)$, provided that for $x=n$ $f(x,\partial_\eta)$ satisfies 
\begin{equation}\label{eq:precondictionD1}
f(n,\partial_\eta)=g(n,\partial_\eta)\left(\prod_{k=-\frac{n-1}{2}}^{\frac{n-1}{2}}\left(\partial_\eta+2k\right)\right),
\end{equation}
with some polynomial $g(n,\partial_\eta)$  of $n$ and $\partial_\eta$. 
Moreover, if \eqref{eq:precondictionD1} holds, $h(x)$ will take the form
\begin{equation}\label{eq:finalresultsD1}
h(x)=g(x,\partial_\eta)\left(\tilde g_0(x,\eta)\frac{\sinh((x-n)\eta)}{x-n}+\sum_{l\geq 1} g_l(x,\eta)(x-n)^l\right).
\end{equation}
\end{thm}
\begin{proof}
We will prove that 
\begin{equation}\label{eq:tobeproven}
\left(\prod_{k=-\frac{n-1}{2}}^{\frac{n-1}{2}}\left(\partial_\eta+2k\right)\right)\frac{\sinh( x\eta)}{\sinh(\eta)}=\sum_l g_l(x,\eta)(x-n)^l
\end{equation}
with $g_l$ some function of $x$ and $\eta$, and $g_0$ taking the form
\begin{equation}\label{eq:tobeproven1}
g_0(x,\eta)=\tilde g_0(x,\eta)\sinh((x-n)\eta).
\end{equation}

Let $y=x-n$. Then we have
\begin{equation}
\frac{\sinh(x\eta)}{\sinh(\eta)}=\frac{1}{\sinh(\eta)}\left(\cosh(y\eta)\sinh(n\eta)+\sinh(y\eta)\cosh(n\eta)\right).
\end{equation}
Substitute the above expression into the LHS of  Eq. \eqref{eq:tobeproven} and expand the result. One then expresses the LHS of  Eq. \eqref{eq:tobeproven} as a linear combination of terms taking the forms of $\partial_\eta^l(h_1(\eta)\sinh(y\eta))$ and $\partial_\eta^l(h_2(\eta)\cosh(y\eta))$, where $h_1$ and $h_2$ are some arbitrary functions. Expanding the action of the differential operator by Leibniz's rule, we obtain linear combinations of $q_1(\eta)\cosh(y\eta) y^m$ and $q_2(\eta)\sinh(y\eta) y^m$, where $q_1$ and $q_2$ are some arbitrary functions. Thus,  \eqref{eq:tobeproven} is obtained.

To get $g_0$, one just act $\partial_\eta$ on neither $\cosh(y\eta)$ nor $\sinh(y\eta)$. Thus  
\begin{equation}
g_0=\left(\prod_{k=-\frac{n-1}{2}}^{\frac{n-1}{2}}\left(\partial_\eta+2k\right)\frac{\sinh(n\eta)}{\sinh(\eta)}\right)\cosh(y \eta)+\left(\prod_{k=-\frac{n-1}{2}}^{\frac{n-1}{2}}\left(\partial_\eta+2k\right)\frac{\cosh(n\eta)}{\sinh(\eta)}\right)\sinh(y \eta)
\end{equation}
Since
$
\frac{\sinh(n\eta)}{\sinh(\eta)}=\sum_{k=-j}^{j}e^{-2k \eta}
$ with $2j+1=n$
and
$
(\partial_\eta+k)e^{-k\eta}=0
$, we have
\begin{equation}
\prod_{k=-\frac{n-1}{2}}^{\frac{n-1}{2}}\left(\partial_\eta+2k\right)\frac{\sinh(n\eta)}{\sinh(\eta)}=0.
\end{equation}
Therefore
\begin{equation}
g_0=\left(\prod_{k=-\frac{n-1}{2}}^{\frac{n-1}{2}}\left(\partial_\eta+2k\right)\frac{\cosh(n\eta)}{\sinh(\eta)}\right)\sinh(y \eta)
\end{equation}
\end{proof}

\section{$\langle\hat F^{\alpha_1\cdots\alpha_m}_{\iota a b}\rangle_{z_e}$ for the case of $\iota=1$}\label{app:gengjiashabiref}
By setting $\iota=1$ in \eqref{eq:mergedqmh}, we can obtain  
\begin{equation}\label{eq:stt}
\begin{aligned}
&\langle \hat F_{1 a b}^{\alpha_1\cdots\alpha_m}\rangle_{z_e}\\
=&-\delta\left(\sum_{i=1}^m\alpha_i-a+b\right)\frac{(-1)^a}{2} t^m e^{b z_e}\left(\prod_{i=1}^m\frac{1}{\sqrt{1+|\alpha_i|}}\right)\frac{1}{\sqrt{(1+|a|)(1+|b|)}} \sum_{\substack{n\in\mathbb Z\\ n\notin [-1,1]}}e^{-\frac{t}{4}(n^2-1)}\\
&\prod_{k=1}^m\left(\alpha_k \frac{n-1}{2}-\frac{\partial_\eta}{2}+\sum_{i=1}^k\alpha_i\right)\frac{n \left(a \frac{n-1}{2}+b+\frac{\partial_\eta}{2}\right) \left(b \frac{n-1}{2}-\frac{\partial_\eta}{2}\right) }{\frac{n^2-1}{4}}\frac{\sinh(n\eta)}{\sinh(\eta)}\\
%%%
+&\delta\left(\sum_{i=1}^m\alpha_i-a+b\right)(-1)^{b-a} t^me^{b z_e}\left(\prod_{i=1}^m\frac{1}{\sqrt{1+|\alpha_i|}}\right)\frac{1}{\sqrt{(1+|a|)(1+|b|)}}\sum_{\substack{n\in\mathbb Z\\n\notin [-2,0]}}e^{-\frac{t}{4}(n+1)^2}\\
&\prod_{k=1}^m\left(-\alpha_k\frac{n+1}{2}-\frac{\partial_\eta}{2}+\sum_{i=1}^k\alpha_i\right)\frac{ \left(\frac{n+1}{2}-(1-|b|-b)\frac{\partial_\eta}{2}\right)\left(\frac{n+1}{2}-(|b|-1-b)\frac{\partial_\eta}{2}+|b|\right) }{\frac{n+1}{2}}\frac{\sinh(n\eta)}{\sinh(\eta)}.
\end{aligned}
\end{equation}
where in the second term we replaced $n$ by $-n$ for the further convenience.  It is easy to check that Theorem \ref{thm:extendsummation} can be applied to extend the summation in the last equation to all $n\in\mathbb Z$. Therefore, we obtain that 
\begin{equation}\label{eq:combinediota1}
\begin{aligned}
&\langle \hat F_{1 a b}^{\alpha_1\cdots\alpha_m}\rangle_{z_e}\\
=&-\frac{(-1)^a}{2}\delta\left(\sum_{i=1}^m\alpha_i-a+b\right)t^me^{b z_e}\left(\prod_{i=1}^m\frac{1}{\sqrt{1+|\alpha_i|}}\frac{1}{\sqrt{1+|a|}\sqrt{1+|b|}}\right)\int_{-\infty}^\infty\dd x\, e^{-\frac{ t}{4}(x^2-1)}\\
&\prod_{k=1}^m\left(\alpha_k\frac{x-1}{2}-\frac{\partial_\eta}{2}+\sum_{i=1}^k\alpha_i\right)\left(a \frac{x-1}{2}+\frac{\partial_\eta}{2}+b\right)\frac{2x\left(b (x-1)-\partial_\eta\right) }{x^2-1}\frac{\sinh(x\eta)}{\sinh(\eta)}\\
%%%
+&(-1)^{b-a}\delta\left(\sum_{i=1}^m\alpha_i-a+b\right)t^me^{b z_e}\left(\prod_{i=1}^m\frac{1}{\sqrt{1+|\alpha_i|}}\frac{1}{\sqrt{1+|a|}\sqrt{1+|b|}}\right)\int_{-\infty}^\infty\dd x e^{-\frac{ t}{4}(x+1)^2}\\
&\prod_{k=1}^m\left(-\alpha_k\frac{n+1}{2}-\frac{\partial_\eta}{2}+\sum_{i=1}^k \alpha_i \right)\left(\frac{n+1}{2}-(|b|-1-b)\frac{\partial_\eta}{2}+|b|\right)\frac{ \left(x+1-(1-|b|-b)\partial_\eta\right)}{x+1}\frac{\sinh(n\eta)}{\sinh(\eta)}\\
&+O(t^\infty).
\end{aligned}
\end{equation}

To show the algorithms to compute Eq. \eqref{eq:combinediota1}, $m\neq 0$ will be assumed  without the loss of generality. For the case of $m=0$, the algorithm can be applied very similarly.

To begin with, we will rewrite Eq. \eqref{eq:combinediota1} as
\begin{equation}
\begin{aligned}
&\langle \hat F_{1 a b}^{\alpha_1\cdots\alpha_m}\rangle_{z_e}\\
=&-\frac{(-1)^a}{2}\delta\left(\sum_{i=1}^m\alpha_i-a+b\right)t^me^{b z_e}\left(\prod_{i=1}^m\frac{1}{\sqrt{1+|\alpha_i|}}\frac{1}{\sqrt{1+|a|}\sqrt{1+|b|}}\right)\int_{-\infty}^\infty\dd x\, e^{-\frac{ t}{4}(x^2-1)}\\
&\prod_{k=2}^m \left(\alpha_k\frac{x-1}{2}-\frac{\partial_\eta}{2}+\sum_{i=1}^k\alpha_i\right)\left(a \frac{x-1}{2}+\frac{\partial_\eta}{2}+b\right)x\left(b-\frac{\partial_\eta}{x-1}\right)  \left(\alpha_1-\frac{\partial_\eta}{x+1}\right) \frac{\sinh(n\eta)}{\sinh(\eta)}\\
%%%
+&(-1)^{b-a}\delta\left(\sum_{i=1}^m\alpha_i-a+b\right)t^me^{b z_e}\left(\prod_{i=1}^m\frac{1}{\sqrt{1+|\alpha_i|}}\frac{1}{\sqrt{1+|a|}\sqrt{1+|b|}}\right)\int_{-\infty}^\infty\dd x\, e^{-\frac{ t}{4}(x+1)^2}\\
&\prod_{k=1}^m\left(-\alpha_k\frac{x+1}{2}-\frac{\partial_\eta}{2}+\sum_{i=1}^k\alpha_k \right)\left(\frac{x+1}{2}-(|b|-1-b)\frac{\partial_\eta}{2}+|b|\right)\left(1-(1-|b|-b)\frac{\partial_\eta}{x+1}\right)\frac{\sinh(x\eta)}{\sinh(\eta)}.
\end{aligned}
\end{equation}
Then for the first integral, it is
\begin{equation}\label{eq:chooseAalpha1}
\begin{aligned}
I_1=&\int_{-\infty}^\infty\dd x\, e^{-\frac{t}{4}(x^2-1)}\prod_{k=2}^m\left(\alpha_k\frac{x-1}{2}-\frac{\partial_\eta}{2}+\sum_{i=1}^k\alpha_i\right)\left(a \frac{x-1}{2}+\frac{\partial_\eta}{2}+b\right)x\left(b-\frac{\partial_\eta}{x-1}\right)  \left(\alpha_1-\frac{\partial_\eta}{x+1}\right) \frac{\sinh(x\eta)}{\sinh(\eta)}.
\end{aligned}
\end{equation}
Because 
\begin{equation}\label{eq:key11}
\begin{aligned}
&\left(b-\frac{\partial_\eta}{x-1}\right)  \left(\alpha_1-\frac{\partial_\eta}{x+1}\right) \frac{\sinh(x\eta)}{\sinh(\eta)}\\
=&\Bigg(b\alpha_1\frac{\sinh(x\eta)}{\sinh(\eta)}-\alpha_1\left(\frac{\cosh (\eta  x)}{\sinh (\eta )}-\frac{\sinh (\eta  (x-1))}{\sinh ^2(\eta )(x-1)}\right)-b\left(\frac{\cosh (\eta  x)}{\sinh (\eta )}-\frac{\sinh (\eta  (x+1))}{\sinh ^2(\eta )(x+1)}\right)+\\
&\left( \frac{\sinh (\eta  x)}{\sinh(\eta)}+\frac{\cosh (\eta )}{\sinh(\eta)^3} \frac{ \sinh (\eta  (x-1))}{x-1}-\frac{\cosh(\eta)}{\sinh(\eta)^3}\frac{ \sinh ( (x+1)\eta)}{x+1}\right)\Bigg),
\end{aligned}
\end{equation}
$I_1$ can be deduced further as the following
\begin{equation}
\begin{aligned}
I_1=&\int_{-\infty}^\infty\dd x\, e^{-\frac{t}{4}(x^2-1)}\prod_{k=2}^m \left(\alpha_k\frac{x-1}{2}-\frac{\partial_\eta}{2}+\sum_{i=1}^k\alpha_i\right)\left(a \frac{x-1}{2}+\frac{\partial_\eta}{2}+b\right)x\frac{(b\alpha_1+1)\sinh(x\eta)-(\alpha_1+b)\cosh(\eta x)}{\sinh(\eta)}\\
%%%
+&\int_{-\infty}^\infty\dd y\, e^{-\frac{ t}{4}(y^2+2y)}\prod_{k=2}^m \left(\alpha_k\frac{y}{2}-\frac{\partial_\eta}{2}+\sum_{i=1}^k\alpha_i\right)\left(a \frac{y}{2}+\frac{\partial_\eta}{2}+b\right)(y+1)\frac{\alpha_1\sinh(\eta)+\cosh (\eta )}{\sinh(\eta)^3} \frac{ \sinh (\eta  y)}{y}\\
%%%
-&\int_{-\infty}^\infty\dd y\, e^{-\frac{t}{4}(y^2-2y)}\prod_{k=2}^m \left(\alpha_k\frac{y}{2}-\frac{\partial_\eta}{2}+\sum_{i=1}^{k-1}\alpha_i\right)\left(a \frac{y}{2}+\frac{\partial_\eta}{2}+b-a\right)(y-1)\frac{\cosh(\eta)-b\sinh(\eta)}{\sinh(\eta)^3}\frac{ \sinh ( y\eta)}{y}.
\end{aligned}
\end{equation}
Then by using the trick
\begin{equation}\label{eq:trick11}
\partial_\eta^n e^{\pm x\eta} f(\eta)=\left.(\pm x+\partial_z)^ne^{\pm x\eta} f(z)\right |_{z\to \eta},
\end{equation}
we have
\begin{equation}
\begin{aligned}
I_1=&\sum_{s=\pm 1}\int_{-\infty}^\infty\dd x\, e^{-\frac{ t}{4}(x^2-1)}\prod_{k=2}^m\left(\frac{\alpha_k-s}{2}x-\frac{\alpha_k+\partial_z}{2}+\sum_{i=1}^k\alpha_i\right)\left(\frac{a+s}{2} x +\frac{-a+\partial_z}{2}+b\right)\\
&x\frac{(-(\alpha_1+b)+s(b\alpha_1+1))e^{sx}}{2\sinh(z)}\Big|_{z\to\eta}\\
%%%
+&\sum_{s=\pm 1}\int_{-\infty}^\infty\dd y\, e^{-\frac{ t}{4}(y^2+2y)}\prod_{k=2}^m\left(\frac{\alpha_k-s}{2}y-\frac{\partial_z}{2}+\sum_{i=1}^k\alpha_i\right)\left(\frac{a+s}{2}y+\frac{\partial_z}{2}+b\right)\\
&(y+1)\frac{ s e^{s \eta y}}{2y}\frac{\alpha_1\sinh(z)+\cosh (z )}{\sinh(z)^3}\Big|_{z\to \eta}\\
%%%
-&\sum_{s=\pm 1}\int_{-\infty}^\infty\dd y\, e^{-\frac{ t}{4}(y^2-2y)}\prod_{k=2}^m \left(\frac{\alpha_k-s}{2} y-\frac{\partial_z}{2}+\sum_{i=1}^{k-1}\alpha_i\right)\left(\frac{a+s}{2}y+\frac{\partial_z}{2}+b-a\right)\\
&(y-1)\frac{ s e^{s\eta y}}{y}\frac{\cosh(z)-b\sinh(z)}{\sinh(z)^3}\Big|_{z\to \eta}
\end{aligned}
\end{equation}
By defining
\begin{equation}
\begin{aligned}
C_1=&\prod_{k=2}^m\left(-\frac{\partial_\eta}{2}+\sum_{i=1}^k\alpha_i\right)\left(\frac{\partial_z}{2}+b\right)\frac{\alpha_1\sinh(\eta)+\cosh (\eta )}{\sinh(\eta)^3}\\
C_2=&-\prod_{k=2}^m \left(-\frac{\partial_\eta}{2}+\sum_{i=1}^{k-1}\alpha_k\right)\left(\frac{\partial_\eta}{2}+b-a\right)\frac{\cosh(\eta)-b\sinh(\eta)}{\sinh(\eta)^3}
\end{aligned}
\end{equation}
$I_1$ can be finally simplified as
\begin{equation}
\begin{aligned}
I_1=&\sum_{s=\pm 1}\int_{-\infty}^\infty\dd x\, e^{-\frac{ t}{4}(x^2-1)}\prod_{k=2}^m \left(\frac{\alpha_k-s}{2}x-\frac{\alpha_k+\partial_z}{2}+\sum_{i=1}^k\alpha_i\right)\left(\frac{a+s}{2} x +\frac{-a+\partial_z}{2}+b\right)\\
&x\frac{(-(\alpha_1+b)+s(b\alpha_1+1))e^{sx}}{2\sinh(z)}\Big|_{z\to\eta}\\
%%%
+&\sum_{s=\pm 1}\int_{-\infty}^\infty\dd y\, e^{-\frac{ t}{4}(y^2+2y)}\Bigg(\prod_{k=2}^m\left(\frac{\alpha_k-s}{2}y-\frac{\partial_z}{2}+\sum_{i=1}^k\alpha_i\right)\left(\frac{a+s}{2}y+\frac{\partial_z}{2}+b\right)\\
&(y+1)\frac{\alpha_1\sinh(z)+\cosh (z )}{\sinh(z)^3}\Big|_{z\to \eta}-C_1\Bigg)\frac{ s e^{s \eta y}}{2y}\\
%%%
-&\sum_{s=\pm 1}\int_{-\infty}^\infty\dd y\, e^{-\frac{ t}{4}(y^2-2y)}\Bigg(\prod_{k=2}^m \left(\frac{\alpha_k-s}{2} y-\frac{\partial_z}{2}+\sum_{i=1}^{k-1}\alpha_i\right)\left(\frac{a+s}{2}y+\frac{\partial_z}{2}+b-a\right)\\
&(y-1)\frac{\cosh(z)-b\sinh(z)}{\sinh(z)^3}\Big|_{z\to \eta}-C_2\Bigg)\frac{ s e^{s\eta y}}{y}\\
&+C_1\int_{-\infty}^\infty e^{-\frac{ t}{4}(y^2+2y)} \frac{\sinh(\eta y)}{y}\dd y-C_2\int_{-\infty}^\infty e^{-\frac{t}{4}(y^2-2y)} \frac{\sinh(\eta y)}{y}\dd y.
\end{aligned}
\end{equation}
In this expression, the integrands of first three terms involve only polynomials of the integral variables, which can be computed by using the same procedure as that to compute $\langle \hat F^{\alpha_1\cdots\alpha_m}_{\frac12 ab}\rangle$. For the last two terms, by using the formula
\begin{equation}
\begin{aligned}
\int_{-\infty}^\infty\dd y e^{-\frac{t}{4}(y^2+b y)}\frac{\sinh(y\eta)}{y}=\frac{1}{2} \pi  \left(\text{erfi}\left(\frac{1+\frac{ t b}{4\eta}}{2 \sqrt{\frac{ t }{4\eta^2}}}\right)+\text{erfi}\left(\frac{1-\frac{t b}{4\eta}}{2 \sqrt{\frac{t }{4\eta^2}}}\right)\right),
\end{aligned}
\end{equation}
and
\begin{equation}
\begin{aligned}
\text{erfi}\left(\frac{1}{x}\right)=e^{\frac{1}{x^2}}\frac{x}{\pi}\sum_{n=0}^\infty \Gamma(\frac{1}{2}+n)x^{2n},
\end{aligned}
\end{equation}
a straightforward calculation can be performed as the following
\begin{equation}\label{eq:extra1}
\begin{aligned}
&\int_{-\infty}^\infty\dd y\, e^{-\frac{ t}{4}(y^2\pm 2y)}\frac{ \sinh ( y\eta)}{y}=\langle 1\rangle\frac{1}{ \sqrt{ \pi }} \frac{\sinh(\eta)}{\eta^2 }\sum_{m=0}^\infty\sum_{n=0}^m\Gamma(\frac{1}{2}+n) \frac{2^n }{\eta^{n+m} }\binom{n+m}{m-n} \Bigg\{e^{ - \eta} +e^{ \eta}\left(-1\right)^{m-n}\Bigg\} \left(\frac{t}{2}\right)^{m+2},
\end{aligned}
\end{equation}
{where Eq. \eqref{eq:norm} is substitute because we are concerned about the expectation value with respect to the normalized coherent state, i.e. $\langle \hat F^{\alpha_1\cdots\alpha_m}_{\iota a b}/\langle 1\rangle$.}
This complete our computation for $I_1$. 

Secondly, $I_2$ can be simplified as the following
\begin{equation}
\begin{aligned}
I_2=&\int_{-\infty}^\infty\dd x e^{-\frac{ t}{4}(x+1)^2}\prod_{k=1}^m \left(-\alpha_k\frac{x+1}{2}-\frac{\partial_\eta}{2}+\sum_{i=1}^k\alpha_i \right)\left(\frac{x+1}{2}-(|b|-1-b)\frac{\partial_\eta}{2}+|b|\right)\frac{\sinh(x\eta)}{\sinh(\eta)}\\
%%%
-&\int_{-\infty}^\infty\dd x e^{-\frac{ t}{4}(x+1)^2}\prod_{k=1}^m\left(-\alpha_k\frac{x+1}{2}-\frac{\partial_\eta}{2}+\sum_{i=1}^k\alpha_i \right)\left(\frac{x+1}{2}+(|b|-1-b)\frac{\partial_\eta}{2}+|b|\right)\frac{(1-|b|-b)\partial_\eta}{x+1}\frac{\sinh(x\eta)}{\sinh(\eta)}.
\end{aligned}
\end{equation}
Because 
\begin{equation}\label{eq:key12}
\frac{\partial_\eta}{x+1}\frac{\sinh(x\eta)}{\sinh(\eta)}=
\frac{1}{\sinh(\eta)^2}
\left(\sinh(\eta)\cosh(x\eta)-\frac{\sinh((x+1)\eta)}{x+1}\right)
\end{equation}
we have
\begin{equation}
\begin{aligned}
I_2
=&\int_{-\infty}^\infty\dd x e^{-\frac{t}{4}(x+1)^2}\prod_{k=1}^m \left[-\alpha_k\frac{x+1}{2}-\frac{\partial_\eta}{2}+\sum_{i=1}^k\alpha_i \right]\left[\frac{x+1}{2}-(|b|-1-b)\frac{\partial_\eta}{2}+|b|\right]\times\\
&\frac{\sinh(x\eta)-(1-|b|-b)\cosh(x \eta)}{\sinh(\eta)}\\
%%%
+&\int_{-\infty}^\infty\dd y e^{-\frac{t}{4}y^2}\prod_{k=1}^m\left[-\alpha_k\frac{y}{2}-\frac{\partial_\eta}{2}+\sum_{i=1}^k\alpha_i \right]\left[\frac{y}{2}-(|b|-1-b)\frac{\partial_\eta}{2}+|b|\right]\frac{(1-|b|-b)}{\sinh(\eta)^2}\frac{\sinh(y\eta)}{y}.
\end{aligned}
\end{equation}
By defining 
\begin{equation}
C=\left(\alpha_1-\frac{\partial_\eta}{2}\right)\left(\alpha_1+\alpha_2-\frac{\partial_\eta}{2}\right)\cdots \left((\alpha_1+\cdots+\alpha_m-\frac{\partial_\eta}{2}\right)\left(|b|-(|b|-1-b)\frac{\partial_\eta}{2}\right)\frac{1-|b|-b}{\sinh(\eta)^2},
\end{equation}
we obtain 
\begin{equation}
\begin{aligned}
I_2
=&\int_{-\infty}^\infty\dd x e^{-\frac{t}{4}(x+1)^2}\prod_{k=1}^m\left[-\alpha_k\frac{x+1}{2}-\frac{\partial_\eta}{2}+\sum_{i=1}^k\alpha_i \right]\times\\
&\left[\frac{x+1}{2}-(|b|-1-b)\frac{\partial_\eta}{2}+|b|\right]\frac{\sinh(x\eta)-(1-|b|-b)\cosh(x \eta)}{\sinh(\eta)}\\
%%%
+&\int_{-\infty}^\infty\dd y e^{-\frac{t}{4}y^2}\Bigg\{\prod_{k=1}^m\left[-\alpha_k\frac{y}{2}-\frac{\partial_\eta}{2}+\sum_{i=1}^k\alpha_k \right]\left[\frac{y}{2}-(|b|-1-b)\frac{\partial_\eta}{2}+|b|\right]-C\Bigg\}\frac{(1-|b|-b)}{\sinh(\eta)^2}\frac{\sinh(y\eta)}{y}\\
+&C\int_{-\infty}^\infty\dd y e^{-\frac{t}{4}y^2}\frac{\sinh(y\eta)}{y}.
\end{aligned}
\end{equation}
Similar as $I_1$, the first two integrals are computable because the integrands therein are just polynomials multiplied by the Gaussian functions.  For the last term, we have
\begin{equation}\label{eq:extra2}
\begin{aligned}
&C\int_{-\infty}^\infty\dd y e^{-\frac{ t}{4}y^2}\frac{\sinh(y\eta)}{y}=C\sgn(\eta)\pi \text{erfi}\left(\frac{1}{2 \sqrt{\frac{ t}{4\eta^2}}}\right)\\
=&\langle 1\rangle  C \frac{ t^2}{ 2\sqrt{ \pi } e^{t/4}}\frac{\sinh(\eta)}{\eta^2 } \sum_{n=0}^\infty \Gamma(\frac{1}{2}+n)\left(\frac{ t}{\eta^2}\right)^n.
\end{aligned}
\end{equation}
{where Eq. \eqref{eq:norm} is substitute because we are concerned about the expectation value with respect to the normalized coherent state, i.e. $\langle \hat F^{\alpha_1\cdots\alpha_m}_{\iota a b}\rangle/\langle 1\rangle$.}
Therefore, $I_2$ is computable. 

In summary, we obtain 
\begin{equation}\label{eq:f1}
\begin{aligned}
&\langle \hat F_{1 a b}^{\alpha_1\cdots\alpha_m}\rangle_{z_e}=-\frac{(-1)^a}{2}\delta\left(\sum_{i=1}^m\alpha_i-a+b\right)t^me^{b z_e}\left(\prod_{i=1}^m\frac{1}{\sqrt{1+|\alpha_i|}}\frac{1}{\sqrt{1+|a|}\sqrt{1+|b|}}\right)I_1\\
%%%
+&(-1)^{b-a}\delta\left(\sum_{i=1}^m\alpha_i-a+b\right)t^me^{b z_e}\left(\prod_{i=1}^m\frac{1}{\sqrt{1+|\alpha_i|}}\frac{1}{\sqrt{1+|a|}\sqrt{1+|b|}}\right)I_2+O(t^\infty)
\end{aligned}
\end{equation}
where $I_1$ and $I_2$ can be computed by using the algorithm introduced above. Taking the operator $\hat F_{1 a b}=D^1_{ab}(h_e)$ as an example, we can have 
\begin{equation}\label{eq:alphais0}
\langle \hat F_{1ab}\rangle_{z_e}/\langle 1\rangle=\left\{
\begin{array}{cc}
e^{i \xi }-\left(\frac{  e^{i \xi }  \tanh \left(\frac{\eta }{2}\right)}{2\eta }+\frac{1}{4}   e^{i \xi } \right)t+O(t^2),&\ a=b=1,\\
1-\frac{ \tanh \left(\frac{\eta }{2}\right)}{\eta }t+O(t^2),&\ a=b=0,\\
e^{-i \xi }-\left(\frac{e^{-i \xi } \tanh \left(\frac{\eta }{2}\right)}{2\eta }+\frac{1}{4}  e^{-i \xi }\right) t+O(t^2),&\ a=b=-1,
\end{array}
\right.
\end{equation}
which is compatible with the corresponding result given in \cite{Dapor:2017gdk}.

\section{Mathematical supports for Sec. \ref{se:leading}}\label{sec:matrixelementhp}
To begin with, we need to study the matrix element of the holonomy and flux. According to the results in \cite{thiemann2001gaugeIII}, we have
\begin{equation}\label{eq:fluxmatrixelements}
\begin{aligned}
\langle\psi_{g_e}|\hat p_s^\alpha|\psi_{g_e'}\rangle=i\langle\psi_{g_e}|\psi_{{g_e'}}\rangle  \frac{\tr(\tau^\alpha  g_e'g_e^\dagger )}{\sinh(\zeta_e^{(0)})}\left(- \zeta_e^{(0)} +(   \coth (\zeta_e^{(0)} )-\frac{1}{\zeta_e^{(0)}})\frac{t}{2}\right)+O(t^\infty).
\end{aligned}
\end{equation}
where $\zeta_e^{(0)}$ is given by $2\cosh(\zeta_e^{(0)})=\tr(g_e' g_e^\dagger)$. For $\hat p_t^\alpha(e)$, one has
\begin{equation}\label{eq:fluxmatrixelementt}
\begin{aligned}
\langle\psi_{g_e}|\hat p_t^\alpha|\psi_{g_e'}\rangle=i\langle\psi_{g_e}|\psi_{{g_e'}}\rangle  \frac{\tr(  g_e^\dagger g_e'\tau^\alpha)}{\sinh(\zeta_e^{(0)})}\left( \zeta_e^{(0)} -(   \coth (\zeta_e^{(0)} )-\frac{1}{\zeta_e^{(0)}})\frac{t}{2}\right)+O(t^\infty).
\end{aligned}
\end{equation}
For the holonomy operator $D^\iota_{ab}(h_e)$, 
we have \cite{thiemann2001gaugeIII}
\begin{equation}\label{eq:holonomymatrixelement}
\begin{aligned}
&\langle\psi_{g_e}|D^{\frac12}_{ab}(h_e)|\psi_{g_e'}\rangle\\
=&\langle\psi_{g_e}|\psi_{g_e'}\rangle\left( \frac{2\tr(\tau^kg_e'g_e^\dagger)
D^{\frac12}_{ab}(\tau^k g_e')}{\left(e^{\zeta_e^{(0)}}+1\right)^2\zeta_e^{(0)}}e^{\frac{\zeta_e^{(0)}}{2}} + \frac{D^{\frac{1}{2}}_{ab}(g_e')}{2\zeta_e^{(0)}}e^{-\frac{\zeta_e^{(0)}}{2}}
\right)e^{\frac{- t}{16}}\left( \left(e^{\zeta_e^{(0)}}+1\right) \zeta_e^{(0)}+(-e^{\zeta_e^{(0)}}+1)\frac{t}{4}\right)+O(t^\infty).
\end{aligned}
\end{equation}

According  to Eqs. \eqref{eq:fluxmatrixelements}, \eqref{eq:fluxmatrixelementt} and \eqref{eq:holonomymatrixelement}, the matrix elements of the fluxes and holomomies are of a form described below 
\begin{equation}
\langle\psi_{g_e}|\hat O_i|\psi_{g_e'}\rangle=\langle\psi_{g_e}|\psi_{g_e'}\rangle\left(E_0(g_e,g_e')+tE_1(g_e,g_e')+O(t^\infty)\right).
\end{equation}
Recalling  Eq. \eqref{eq:insertidentity}, we are going to consider integrals containing $\frac{\langle\psi_{g_e}|\psi_{g_e'}\rangle}{\|\psi_{g_e}\|\|\psi_{g_e'}\|}$. These integrals can be analyzed with the generalized stationary phase approximation guided by H\"ormander’s theorem 7.7.5 in \cite{hormander2015analysis}.
\begin{thm}\label{hormanderThm}
Let $K$ be a compact subset in $\mathbb{R}^n$, X be an open neighborhood of K, and k be a positive integer.  If (1) the complex functions $u\in C^{2k}_{0}(K)$, $f\in C^{3k+1}(X)$ and  $\Im(f)\geq 0$ in X, with $\Im(f)$ being the imaginary part of $f$; (2) there is a unique point $x_0\in K$ satisfying $\Im(f(x_0))=0$, $f'(x_0)=0$, and $\det (f''(x_0))\neq 0$ ($f''$ denotes the Hessian matrix), $f'\neq 0$ in $K \backslash\left\{x_{0}\right\}$ then we have the following estimation:\\
\begin{equation}\label{eq:theorem1}
\left|\int_{K} u(x) e^{i\lambda f(x)} dx-e^{i\lambda f(x_0)}\left[\mathrm{det}\left(\dfrac{\lambda f''(x_0)}{2\pi i}\right)\right]^{-\frac{1}{2}}\sum_{s=0}^{k-1}\left(\dfrac{1}{\lambda}\right)^s L_s u(x_0)\right|\leq C\left(\dfrac{1}{\lambda}\right)^{k}
\sum_{|\alpha|\leq 2k}\mathrm{sup}\left|D^{\alpha}u\right|.
\end{equation}
Here the constant C is bounded when f stays in a bounded set in $C^{3k+1}(X)$. We have used the standard multi-index notation $\alpha=\langle \alpha_1,...,\alpha_n\rangle$ and
\begin{equation}
D^{\alpha}=(-i)^\alpha\frac{\partial^{|\alpha|}}{\partial x_1^{\alpha_1}...\partial x_n^{\alpha_n}}, \quad \text{where}\quad |\alpha|=\sum_{i=1}^{n}\alpha_i
\end{equation}
$L_s u(x_0)$ denotes the following operation on u:
\begin{equation}\label{eq:theorem2}
L_s u(x_0)= i^{-s} \sum_{l-m=s}\sum_{2l\geq 3m}\frac{(-1)^l 2^{-l}}{l!m!}\left[\sum_{a,b=1}^{n}H_{ab}^{-1}(x_0)\dfrac{\partial^2}{\partial x_a\partial x_b}\right]^l \left(g_{x_0}^m u\right)\left(x_0\right),
\end{equation}
where $H(x)=f''(x)$ denotes the Hessian matrix and the function $g_{x_0}(x)$ is given by
\begin{equation}\label{eq:factorg}
g_{x_0}(x)=f(x)-f(x_0)-\frac{1}{2}H^{ab}\left(x_0\right)\left(x-x_0\right)_a\left(x-x_0\right)_b
\end{equation}
such that $g_{x_0}\left(x_0\right)=g'_{x_0}\left(x_0\right)=g''_{x_0}\left(x_0\right)=0$. For each s, $L_s$ is a differential operator of order 2s acting on $u\left(x\right)$.
\end{thm}
\subsection{Analysis of integrals containing $\frac{\langle\psi_{g_e}|\psi_{g_e'}\rangle}{\|\psi_{g_e}\|\|\psi_{g_e'}\|}$}\label{sec:hessian}
Define
\begin{equation}
G(\vec p^{(1)},\vec\theta^{(1)},\vec p^{(2)},\vec\theta^{(2)}):=\frac{\langle\psi_{g^{(1)}}|\psi_{g^{(2)}}\rangle}{\|\psi_{g^{(1)}}\|\|\psi_{g^{(2)}}\|}
\end{equation}
with $g^{(k)}$ parameterized by
\begin{equation}
g^{(k)}=e^{i\vec p^{(k)}\cdot \vec\tau}e^{\vec \theta^{(k)}\cdot \vec\tau},\ k=1,2.
\end{equation}
According to Eq. \eqref{eq:innerproduct} and Eq. \eqref{eq:norm}, $G(\vec p^{(1)},\vec\theta^{(1)},\vec p^{(2)},\vec\theta^{(2)})$ reads
\begin{equation}
G(\vec p^{(1)},\vec\theta^{(1)},\vec p^{(2)},\vec\theta^{(2)})=\frac{\zeta \sqrt{\sinh(p^{(1)})\sinh(p^{(2)})}}{\sqrt{p^{(1)}p^{(2)}}\sinh(\zeta)}e^{-\frac{-2\zeta^2+(p^{(1)})^2+(p^{(2)})^2}{2t}}
\end{equation}
where $p^{(i)}=\sqrt{\vec p^{(i)}\cdot\vec p^{(i)}}$ and $\zeta$ is given by 
\begin{equation}\label{eq:todefinezeta}
2\cosh(\zeta)=\tr(g_1^\dagger g_2).
\end{equation}
Denote $$S(\vec p^{(1)},\vec\theta^{(1)},\vec p^{(2)},\vec\theta^{(2)}):=-2\zeta^2+(p^{(1)})^2+(p^{(2)})^2.$$
We first claim that
\begin{lmm}\label{lmm:saddlecondition1}
$\Re(S(\vec p^{(1)},\vec\theta^{(1)},\vec p^{(2)},\vec\theta^{(2)}))$, the real part of $S(\vec p^{(1)},\vec\theta^{(1)},\vec p^{(2)},\vec\theta^{(2)})$, is non-negative and vanishes iff $\vec p^{(2)}=\vec p^{(1)}$ and $\vec \theta^{(2)}=\vec \theta^{(1)}$. 
\end{lmm}
To prove this lemma, let us introduce the following proposition which is given in \cite{thiemann2001gaugeII}.% and also proven in Appendix \ref{app:deltaqg0}.
\begin{pro}\label{pro:deltaqg0}
Let $g=e^{i\vec p\cdot\vec\tau}e^{\vec \theta\cdot\vec\tau}$. Considering $\zeta=s+i\phi\in\mathbb C$ with $s\in\mathbb R$ and $\phi\in[0,\pi]$ determined by $\cosh(\zeta)=\tr(g)$, we have, with denoting $p:=\sqrt{\vec p\cdot\vec p}$, 
\begin{equation}
\delta=\frac{p^2}{4}-s^2+\phi^2\geq0
\end{equation}
where the equality occurs iff $\theta:=\sqrt{\vec\theta\cdot\vec\theta}=0$. 
\end{pro}

Thanks to this proposition, we prove Lemma \ref{lmm:saddlecondition1} as follows.
\begin{proof}[Proof of Lemma \ref{lmm:saddlecondition1}]
$g_1^\dagger g_2$ can be decomposed as $$g_1^\dagger g_2=e^{i\vec x\cdot\vec\tau}e^{\vec y\cdot\vec\tau}.$$
Denote $x=\sqrt{\vec x\cdot\vec x}$ and $y=\sqrt{\vec y\cdot\vec y}$. Then
\begin{equation}
\Re(S(\vec p,\vec\theta,\vec p,\vec\theta))=2 \delta-\frac{x^2}{2}+(p^{(1)})^2+(p^{(2)})^2.
\end{equation}
where $\delta:=x^2/4-\Re(\zeta)^2$ is non-negative according to the proposition \ref{pro:deltaqg0}. Thus we only need to prove that $-x^2/2+(p^{(1)})^2+(p^{(2)})^2\geq 0$. 

By definition, we have
$2\cosh(x)= \tr(g_1^\dagger g_2g_2^\dagger g_1)$,
which leads to
\begin{equation}
2\cosh(x)=\tr(e^{2i\vec p^{(1)}\cdot\vec\tau} e^{2i\vec p^{(2)}\cdot\vec\tau}).
\end{equation}
Since $e^{i\vec \mu\cdot\vec\tau}=\cosh(\frac{\mu}{2})\mathbb I+2i\frac{\vec\mu\cdot\vec\tau}{\mu}\sinh(\frac{\mu}{2})$ with $\mu=\sqrt{\vec\mu\cdot\vec\mu}$, we have
\begin{equation}\label{eq:inequality1}
\begin{aligned}
\cosh(x)=\frac{1- \beta}{2}\cosh(p^{(1)}-p^{(2)})+\frac{1+\beta}{2}\cosh(p^{(1)}+p^{(2)})\leq \cosh(p^{(1)}+p^{(2)})
\end{aligned}
\end{equation}
where $\beta=\frac{\vec p^{(1)}\cdot\vec p^{(2)}}{p^{(1)}p^{(2)}}\in [-1,1]$ and $\cosh(p^{(1)}+p^{(2)})\geq \cosh(p^{(1)}-p^{(2)})$ is used. Moreover, because of $\sqrt{2(p^{(1)})^2+2(p^{(2)})^2}\geq p^{(1)}+p^{(2)}\geq 0$, it has
\begin{equation}\label{eq:inequality2}
\cosh(\sqrt{2(p^{(1)})^2+2(p^{(2)})^2}\,)\geq\cosh(p^{(1)}+p^{(2)}).
\end{equation}
Combining the results of \eqref{eq:inequality1} and \eqref{eq:inequality2}, one finally have
\begin{equation}
-\frac{x^2}{2}+(p^{(1)})^2+(p^{(2)})^2\geq 0
\end{equation}
where the equality occurs only if $\vec p^{(1)}=\vec p^{(2)}$. 

In summary we have
\begin{equation}
\Re(S(\vec p,\vec\theta,\vec p,\vec\theta))\geq 0
\end{equation}
and $\Re(S(\vec p,\vec\theta,\vec p,\vec\theta))=0$ only if $\vec p^{(1)}=\vec p^{(2)}$ and $\delta=0$ which means $\vec\theta^{(1)}=\vec\theta^{(2)}$. 
\end{proof}

It turns out below that the integrand of Eq. \eqref{eq:insertidentity} consists of a Gaussian-like function
\begin{equation}
e^{-\frac{ 1}{2t}\left(S(\vec p,\vec\theta,\vec p^{(1)},\vec\theta^{(1)})+S(\vec p^{(1)},\vec\theta^{(1)},\vec p^{(2)},\vec\theta^{(2)})+\cdots+S(\vec p^{(k)},\vec\theta^{(k)},\vec p,\vec\theta)\right)}.
\end{equation}
Lemma \ref{lmm:saddlecondition1} suggests us to do the stationary phase approximation analysis at $\vec p^{(i)}=\vec p$ and $\vec\theta^{(i)}=\vec\theta$. Notice that $\vec p$ and $\vec\theta$ are given to parameterize $g$ as $g:=e^{i\vec p\cdot\vec\tau}e^{\vec \theta\cdot\vec\tau}$, 
and $g$ labels the coherent state $|\psi_g\rangle$ with respect to which the expectation value of $\hat O$ is computed. Thus, rather than considering all values of $\vec p$ and $\vec\theta$, it is sufficient to set 
\begin{equation}
\begin{aligned}
\vec p_o=(0,0,p),\ \vec \theta_o=(0,0,\theta)
\end{aligned}
\end{equation}
according to \eqref{eq:basicformula}. Denote
\begin{equation}\label{eq:definef}
f_{k;\vec p,\vec\theta}(\vec p^{(1)},\vec\theta^{(1)},\cdots,\vec p^{(k)},\vec\theta^{(k)})\equiv S(\vec p,\vec\theta,\vec p^{(1)},\vec\theta^{(1)})+S(\vec p^{(1)},\vec\theta^{(1)},\vec p^{(2)},\vec\theta^{(2)})+\cdots+S(\vec p^{(k)},\vec\theta^{(k)},\vec p,\vec\theta),
\end{equation}
we have the following result:
\begin{thm}\label{thm:saddlepoint2}
\begin{itemize}
\item[(i)] $\Re(f_{k;\vec p,\vec \theta})\geq 0$ and the equality occurs only when all $g^{(i)}$ coincide, namely $\vec p^{\, (i)}=\vec p$ and $\vec\theta^{(i)}=\vec\theta$. 
\item[(ii)] At $\vec p^{(i)}=\vec p_o=(0,0,p)$ and $\vec\theta^{(i)}=\vec\theta_o=(0,0,\theta)$, it has
\begin{equation}
\nabla_{\vec p^{(i)}}f_{k;\vec p_o,\vec \theta_o}=0=\nabla_{\vec \theta^{(i)}}f_{k;\vec p_o,\vec \theta_o},\ \forall i=1,\cdots,k.
\end{equation}
\item[(iii)] The Hessian matrix $f''_{k;\vec p_o,\vec \theta_o}$ of $f_{k;\vec p_o,\vec \theta_o}$ at $\vec p^{(i)}=\vec p_o=(0,0,p)$ and $\vec\theta^{(i)}=\vec\theta_o=(0,0,\theta)$ is non-degenerate with the determinant 
\begin{equation}\label{eq:determinedhessian}
\det\left(f''_{k;\vec p_o,\vec \theta_o}\right)\Big|_{\vec p^{(i)}=\vec p_o,\vec\theta^{(i)}=\vec\theta_o}=\left(\frac{1024 \sin^4\left(\frac{\theta }{2}\right)}{\theta^4}\right)^{k}
\end{equation}
\end{itemize}
\end{thm}
\begin{proof}
The first statement is true by using Lemma \ref{lmm:saddlecondition1}. For the second statement, let us consider 
\begin{equation}
S(\vec p^{(1)},\vec\theta^{(1)},\vec p^{(2)},\vec\theta^{(2)})=-2\zeta^2+(\vec p^{(1)})^2+(\vec p^{(2)})^2
\end{equation}
where $\zeta$ is given by
\begin{equation}
\cosh(\zeta)=\frac{1}{2}\tr(e^{\vec\theta^{(2)}\cdot\vec\tau}e^{-\vec\theta^{(1)}\cdot\vec\tau}e^{i\vec p^{(1)}\cdot\vec\tau}e^{i\vec p^{(2)}\cdot\vec\tau}).
\end{equation}
Then it has that
\begin{equation}\label{eq:firstD}
\begin{aligned}
\frac{\partial\cosh(\zeta)}{\partial p_j^{(1)}}\Bigg|_{\vec p_o,\vec \theta_o}=\frac{\partial\cosh(\zeta)}{\partial p_j^{(2)}}\Bigg|_{\vec p_o,\vec \theta_o}&=\frac{\delta_{j,3}}{2}\sinh(p)\\
\frac{\partial\cosh(\zeta)}{\partial \theta_j^{(1)}}\Bigg|_{\vec p_o,\vec \theta_o}=-\frac{\partial\cosh(\zeta)}{\partial \theta_j^{(2)}}\Bigg|_{\vec p_o,\vec \theta_o}&=\frac{i\delta_{j,3}}{2}\sinh(p)\\
\end{aligned}
\end{equation}
where the subscript $\vec p_o,\vec\theta_o$ indicates to take values at $\vec p^{(1)}=\vec p_o=\vec p^{(2)}$ and $\vec \theta^{(1)}=\vec \theta_o=\vec \theta^{(2)}$. According to Eq. \eqref{eq:firstD}, we have
\begin{equation}
\begin{aligned}
\nabla_{\vec p^{(1)}}S(\vec p^{(1)},\vec\theta^{(1)},\vec p^{(2)},\vec\theta^{(2)})\Big|_{\vec p_o,\vec\theta_o}=0=\nabla_{\vec p^{(2)}}S(\vec p^{(1)},\vec\theta^{(1)},\vec p^{(2)},\vec\theta^{(2)})\Big|_{\vec p_o,\vec\theta_o}
\end{aligned}
\end{equation}
and, thus, $\nabla_{\vec p^{(i)}} f_{k;\vec p_o,\vec \theta_o}\big|_{\vec p_o,\vec\theta_o}=0$ for all $i=1,2,\cdots, k$. For $\nabla_{\vec \theta^{(i)}} f_{k;\vec p_o,\vec \theta_o}\big|_{\vec p_o,\vec\theta_o}$, the second equation in Eq. \eqref{eq:firstD} gives us
\begin{equation}
\nabla_{\vec \theta^{(i)}} f_{k;\vec p_o ,\vec \theta_o}\big|_{\vec p_o,\vec\theta_o}=\nabla_{\vec \theta^{(i)}} S(\vec p^{(i-1)},\vec\theta^{(i-1)},\vec p^{(i)},\vec\theta^{(i)})\Big|_{\vec p_o,\vec\theta_o}+\nabla_{\vec \theta^{(i)}} S(\vec p^{(i)},\vec\theta^{(i)},\vec p^{(i+1)},\vec\theta^{(i+1)})\Big|_{\vec p_o,\vec\theta_o}=0,
\end{equation}
which completes the proof of the second statement. 

For the last statement, using the conclusion of the second statement, we can immediately get that
\begin{equation}
\frac{\partial f_{k;\vec p_o,\vec \theta_o}}{\partial x_m^{(i)}\partial y_n^{(j)}}\Bigg|_{\vec p_o,\vec \theta_o}=0,\ |i-j|>1
\end{equation}
where $x_m^{(i)}$ and $y_n^{(j)}$ represent $\vec\theta_m^{(i)}$ or $\vec p_m^{(i)}$. 
Therefore, if we order the arguments $\vec p^{(i)},\vec\theta^{(i)}$ as
\begin{equation}
p^{(1)}_1,p^{(1)}_2,p^{(1)}_3,\theta^{(1)}_1,\theta^{(1)}_2,\theta^{(1)}_3,p^{(2)}_1,p^{(2)}_2,p^{(2)}_3,\theta^{(2)}_1,\theta^{(2)}_2,\theta^{(2)}_3,\cdots. 
\end{equation}
to arrange the matrix elements of $ f''_{k;\vec p_o,\vec \theta_o}(\vec p_o,\vec\theta_o)\equiv f''_{k;\vec p_o,\vec \theta_o}\Big|_{\vec p_o,\vec\theta_o}$, the resulting matrix is block-tridiagonal matrix. Moreover, since all of the $p$-arguments, as well as the $\theta$-arguments, in $f_{k;\vec p_o,\vec\theta_o}$ are symmetric, we conclude that
\begin{equation}
\frac{\partial f_{k;\vec p_o,\vec \theta_o}}{\partial x_m^{(i)}\partial y_l^{(j)}}\Bigg|_{\vec p_o,\vec \theta_o}=\frac{\partial f_{k;\vec p_o,\vec \theta_o}}{\partial x_m^{(i')}\partial y_l^{(j')}}\Bigg|_{\vec p_o,\vec \theta_o}.
\end{equation} 
Consequently $f''_{k;\vec p_o,\vec \theta_o}(\vec p_o,\vec\theta_o)$ takes the form
\begin{equation}\label{eq:hessian}
f''_{k;\vec p_o,\vec \theta_o}(\vec p_o,\vec\theta_o)=\left(
\begin{array}{cccccc}
  A & B & 0&0&\cdots & 0 \\
  B^T & A & B&0& \cdots & 0 \\
  0&B^T & A & B& \cdots & 0 \\
  \vdots    &\vdots      & \vdots          & \ddots & \vdots&    \vdots      \\
  0 & 0& \cdots & B^T & A&B\\
  0 & 0& \cdots & 0&B^T & A
\end{array}
\right)
\end{equation}
where $A$ and $B$ are $6\times 6$ matrix. The matrix $A$ and $B$ can be calculated by considering the case with $k=2$, which gives us
\begin{equation*}
A=\left(
\begin{array}{cccccc}
 \frac{4 \tanh \left(\frac{p}{2}\right)}{p} & 0 & 0 & -\frac{4 \sin ^2\left(\frac{\theta }{2}\right) \tanh \left(\frac{p}{2}\right)}{\theta } & -\frac{2 \sin (\theta ) \tanh \left(\frac{p}{2}\right)}{\theta } & 0 \\
 0 & \frac{4 \tanh \left(\frac{p}{2}\right)}{p} & 0 & \frac{2 \sin (\theta ) \tanh \left(\frac{p}{2}\right)}{\theta } & -\frac{4 \sin ^2\left(\frac{\theta }{2}\right) \tanh \left(\frac{p}{2}\right)}{\theta } & 0 \\
 0 & 0 & 2 & 0 & 0 & 0 \\
 -\frac{4 \sin ^2\left(\frac{\theta }{2}\right) \tanh \left(\frac{p}{2}\right)}{\theta } & \frac{2 \sin (\theta ) \tanh \left(\frac{p}{2}\right)}{\theta } & 0 & -\frac{4 p (\cos (\theta )-1) \coth (p)}{\theta ^2} & 0 & 0 \\
 -\frac{2 \sin (\theta ) \tanh \left(\frac{p}{2}\right)}{\theta } & -\frac{4 \sin ^2\left(\frac{\theta }{2}\right) \tanh \left(\frac{p}{2}\right)}{\theta } & 0 & 0 & -\frac{4 p (\cos (\theta )-1) \coth (p)}{\theta ^2} & 0 \\
 0 & 0 & 0 & 0 & 0 & 2 \\
\end{array}
\right)
\end{equation*}
and
\begin{equation*}
B=\left(
\begin{array}{cccccc}
 -\frac{2 \tanh \left(\frac{p}{2}\right)}{p} & 0 & 0 & \frac{2 \sin ^2\left(\frac{\theta }{2}\right) \tanh \left(\frac{p}{2}\right)+i \sin (\theta )}{\theta } & \frac{\sin (\theta ) \tanh \left(\frac{p}{2}\right)+i (\cos (\theta )-1)}{\theta } & 0 \\
 0 & -\frac{2 \tanh \left(\frac{p}{2}\right)}{p} & 0 & \frac{-i \cos (\theta )-\sin (\theta ) \tanh \left(\frac{p}{2}\right)+i}{\theta } & \frac{2 \sin ^2\left(\frac{\theta }{2}\right) \tanh \left(\frac{p}{2}\right)+i \sin (\theta )}{\theta } & 0 \\
 0 & 0 & -1 & 0 & 0 & i \\
 \frac{2 \sin ^2\left(\frac{\theta }{2}\right) \tanh \left(\frac{p}{2}\right)-i \sin (\theta )}{\theta } & \frac{i \left(\cos (\theta )+i \sin (\theta ) \tanh \left(\frac{p}{2}\right)-1\right)}{\theta } & 0 & \frac{2 p (\cos (\theta )-1) \coth (p)}{\theta ^2} & -\frac{2 i p (\cos (\theta )-1)}{\theta ^2} & 0 \\
 \frac{-i \cos (\theta )+\sin (\theta ) \tanh \left(\frac{p}{2}\right)+i}{\theta } & \frac{2 \sin ^2\left(\frac{\theta }{2}\right) \tanh \left(\frac{p}{2}\right)-i \sin (\theta )}{\theta } & 0 & \frac{2 i p (\cos (\theta )-1)}{\theta ^2} & \frac{2 p (\cos (\theta )-1) \coth (p)}{\theta ^2} & 0 \\
 0 & 0 & -i & 0 & 0 & -1 \\
\end{array}
\right).
\end{equation*}
With the expression of $A$ and $B$, it can be verified that 
\begin{equation}\label{eq:ABA}
BA^{-1}B^T=B^T A^{-1}B=0.
\end{equation}
To calculate $\det(f''_{k;\vec p_o,\vec \theta_o}(\vec p_o,\vec\theta_o))$, we define matrices of $\tilde B$, $\tilde C$ and D of dimensions $6\times 6(k-1)$, $6(k-1)\times 6$ and $6(k-1)\times 6(k-1)$ respectively such that 
\begin{equation}\label{eq:block2}
f''_{k;\vec p_o,\vec \theta_o}(\vec p_o,\vec\theta_o)=\left(
\begin{array}{cc}
A&\tilde B\\
\tilde C&D
\end{array}
\right).
\end{equation}
Then by using the property of the Schur complement, it has
\begin{equation}
\det(f''_{k;\vec p_o,\vec \theta_o}(\vec p_o,\vec\theta_o))=\det(A)\det(D-\tilde CA^{-1}\tilde B)=\det(A)\det(D)
\end{equation}
where we used $\tilde CA^{-1}\tilde B=0$ because of Eq. \eqref{eq:ABA} and that $f''_{k;\vec p_o,\vec \theta_o}$ is block-tridiagonal matrix. Because $D$ is the Hessian matrix $f''_{k-1;\vec p_o,\vec \theta_o}(\vec p_o,\vec\theta_o)$, we finally have
\begin{equation}
\det(f''_{k;\vec p_o,\vec \theta_o}(\vec p_o,\vec\theta_o))=\det(A)^k=\left(\frac{1024\sin^4(\frac{\theta}{2})}{\theta^4}\right)^k.
\end{equation}
\end{proof}
By these results, the stationary phase approximation method introduced in Theorem \ref{hormanderThm} can be applied to calculate the integral \eqref{eq:insertidentity}. 

Taking advantage of the above results, we now can come to the proof of Theorem \ref{thm:leadingordergeneral}.
\subsection{proof of Theorem \ref{thm:leadingordergeneral}}
As in Eq. \eqref{eq:insertidentity}, it has
\begin{equation}\label{eq:insertidentity2}
\frac{\langle\Psi_{\vec w}|\hat O|\Psi_{\vec w}\rangle}{\langle\Psi_{\vec w}|\Psi_{\vec w}\rangle}=\int\prod_{j=1}^{k-1}\prod_{e\in E(\gamma)}\frac{2}{\pi t^3}\dd\mu_H(u^{(j)}_e)\dd^3\vec p_e^{(j)}\prod_{i=1}^k\frac{\langle\Psi_{\vec g^{(i-1)}}|\hat O_i|\Psi_{\vec g^{(i)}}\rangle}{\|\Psi_{\vec g^{(i-1)}}\|\|\Psi_{\vec g^{(i)}}\|}
\end{equation}
where we denoted $|\Psi_{\vec g^{(0)}}\rangle=|\Psi_{\vec g^{(k)}}\rangle:=|\Psi_{\vec w}\rangle$, applied Eq. \eqref{eq:measure} and used the decomposition
\begin{equation}
g_e^{(i)}
=e^{i\vec p^{(i)}_e\cdot\vec\tau}e^{\vec\theta^{(i)}_e\cdot\vec\tau}=e^{i\vec p^{(i)}_e\cdot\vec\tau} u_e^{(i)}.
\end{equation}
By the assumption, we have
\begin{equation}
\frac{\langle\Psi_{\vec g^{(i-1)}}|\hat O_i|\Psi_{\vec g^{(i)}}\rangle}{\|\Psi_{\vec g^{(i-1)}}\|\|\Psi_{\vec g^{(i)}}\|}=\frac{\langle \Psi_{\vec g^{(i-1)}}|\Psi_{\vec g^{(i)}}\rangle}{\|\Psi_{\vec g^{(i-1)}}\|\|\Psi_{\vec g^{(i)}}\|}\left(E_0^{(i)}(\vec g^{(i-1)},\vec g^{(i)})+t E_1^{(i)}(\vec g^{(i-1)},\vec g^{(i)})+O(t^\infty)\right).
\end{equation}
Thus Eq. \eqref{eq:insertidentity2} takes the form
\begin{equation}\label{eq:integralwithmatrixelement}
\begin{aligned}
\frac{\langle\Psi_{\vec w}|\hat O|\Psi_{\vec w}\rangle}{\langle\Psi_{\vec w}|\Psi_{\vec w}\rangle}=\int\prod_{j=1}^{k-1}\prod_{e\in E(\gamma)}\frac{2}{\pi t^3}\dd\mu_H(u^{(j)}_e)\dd^3\vec p_e^{(j)}\prod_{i=1}^k\frac{\langle\Psi_{\vec g^{(i-1)}}|\Psi_{\vec g^{(i)}}\rangle}{\|\Psi_{\vec g^{(i-1)}}\|\|\Psi_{\vec g^{(i)}}\|}\prod_{l=1}^k E^{(l)}(\vec g^{(l-1)},\vec g^{(l)})
\end{aligned}
\end{equation}
with $P(\{\vec g^{(i)}\}_{i=1}^k)$ denoting the function
\begin{equation}
E^{(l)}(\vec g^{(l-1)},\vec g^{(l)}):=\left(E_0^{(l)}(\vec g^{(l-1)},\vec g^{(l)})+t E_1^{(l)}(\vec g^{(l-1)},\vec g^{(l)})+O(t^\infty)\right).
\end{equation}
Eq. \eqref{eq:integralwithmatrixelement} can be analyzed with the stationary phase approximation 
according to Theorem \ref{thm:saddlepoint2}. 
It should be noticed that the Haar measure $\dd\mu_H$ can be expressed as 
\begin{equation}
\dd\mu_H(u)=\frac{\sin ^2\left(\frac{1}{2} \sqrt{\vec\theta \cdot\vec\theta }\right)}{4 \pi ^2 \,(\vec\theta \cdot\vec\theta)}\dd^3\vec\theta
\end{equation}
where $u\in\mathrm{SU}(2)$ is coordinatized as $u=e^{\vec\theta\cdot\vec\tau}$ and $\dd^3\vec\theta$ is the Lebesgue measure on $\mathbb R^3$. Substituting the last equation into Eq. \eqref{eq:integralwithmatrixelement} and applying Eq. \eqref{eq:theorem1} as well as Eq. \eqref{eq:determinedhessian},  we finally obtain
\begin{equation}\label{eq:expectationaftersaddle}
\begin{aligned}
&\frac{\langle\Psi_{\vec w}|\hat O|\Psi_{\vec w}\rangle}{\langle\Psi_{\vec w}|\Psi_{\vec w}\rangle}=\prod_{e\in E(\gamma)}\left(\frac{\theta_e^2}{\sin^2(\theta_e/2)}\right)^{k-1}\,\sum_{s=0}^{l-1}\left(2t\right)^s \mathfrak D^{(s)}\Big|_{g_e^{(k)}=e^{i w_e\tau_3},\forall k,e}+O(t^l)
\end{aligned}
\end{equation}
where $\mathfrak D^{(s)}$ takes the form
\begin{equation}\label{eq:Ds}
\mathfrak D^{(s)}= (-1)^s \sum_{l-j=s}\sum_{2l\geq 3j}\frac{(-1)^l 2^{-l}}{l!j!}\left[\sum_{a,b=1}^{n}H_{ab}^{-1}(x_0)\dfrac{\partial^2}{\partial x_a\partial x_b}\right]^l \left(G^{(j)}\prod_{i=1}^k E^{(i)}\right)\Bigg|_{g_e^{(l)}=e^{i w_e\tau_3},\forall l,e}
\end{equation}
with $G^{(j)}$ being defined as the following,
\begin{equation}
G^{(j)}\left(\{\vec g^{(i)}\}_{i=k}^{k-1}\right)=\left(\prod_{e\in E(\gamma)}\prod_{n=1}^{k-1}\frac{\sin^2(|\vec\theta_e^{(n)}|/2)}{|\vec\theta_e^{(n)}|^2}\right)\left(\prod_{e\in E(\gamma)}\prod_{n=0}^{k-1}\frac{\zeta_e^{(n,n+1)} \sqrt{\sinh(p_e^{(n)})\sinh(p_e^{(n+1)})}}{\sqrt{p_e^{(n)}p_e^{(n+1)}}\sinh(\zeta_e^{(n,n+1)})}\right)g_{x_0}(\{\vec g^{(i)}\}_{i=1}^{k-1})^j.
\end{equation}  
Here $g_{x_0}$ is some function defined by applying \eqref{eq:factorg} to the current case.

If the leading order term of $\mathfrak D^{(s)}$ is claimed to be $O(t^{d_s})$, then each term in the summation over $s$ of Eq. \eqref{eq:expectationaftersaddle} is $O(t^{s+d_s})$.
Moreover, for each $l$ in Eq. \eqref{eq:Ds}, the derivative acting on $G^{(j)}\prod E^{(i)}$ is of order $2l$  in total. Because of the properties given by Eq.\eqref{eq:factorg}, the non-vanishing result appears when there are at least $3j$ derivatives acting on $G^{(j)}$, which indicates that the order of derivative that acts on the term $\prod E^{(i)}$ in Eq. \eqref{eq:Ds} is $2l-3j$. Because $2l-3j=s$ and $l=s+j$, there are at most $2s-j$ derivatives acting on $\prod E^{(i)}$. Further, since $j\geq 0$, the maximum order of derivative acting on the term $\prod E^{(i)}$ is $2s$, which only occurs for $l=s$.
Then, let us count the leading order of $\mathfrak D^{(s)}$ for a given $s$. 
According to the expression of $\mathfrak D^{(s)}$, once it is evaluated at the critical point given by $g_e^{(k)}=e^{i\omega_e\tau_3}$, those $E^{(m)}$ contributed by operators $\hat O_m$ satisfying \eqref{eq:operatorom} will inevitably increase the power of its leading order term if they are not acted by any derivative operators.
For a fixed $s$ there are at least $n_s$ of these ``non-acted" terms with
\begin{equation}
n_s=\frac{(N_0-2s)+|N_0-2s|}{2}.
\end{equation}
Therefore, $\dfrac{\langle\Psi_{\vec w}|\hat O|\Psi_{\vec w}\rangle}{\langle\Psi_{\vec w}|\Psi_{\vec w}\rangle}$ is of order of $t^{n}$ with
\begin{equation}
n\geq \min\limits_{s\in\mathbb Z^+}(s+n_s)=\floor{\frac{N_0+1}{2}}.
\end{equation}

\subsection{proof of Theorem \ref{thm:leadingordermultiplyPH}}\label{app:proofofthem44}
 At first, 
 denote $\hat p^{0}(e)$, $D^{\frac12}_{\frac12\frac12}(h_e)$ and $D^{\frac12}_{-\frac12-\frac12}(h_e)$ by $\hat O_a^{\rm d}$ with $a=1,2,3$ respectively, and $p^{\pm 1}(e)$, $D^{\frac12}_{\frac12-\frac12}(h_e)$ or $D^{\frac12}_{-\frac12\frac12}(h_e)$ by $\hat O_a^{\rm nd}$ with $a=1,2,3,4$ respectively. According to \eqref{eq:fluxmatrixelements}, \eqref{eq:fluxmatrixelementt} and \eqref{eq:holonomymatrixelement}, the matrix elements of $\hat O_a^{i}$ take the form
 \begin{equation}
 \langle\psi_{g^{(1)}}|\hat O_a^{i}|\psi_{g^{(2)}}\rangle=\langle\psi_{g^{(1)}}|\psi_{g^{(2)}}\rangle E^{i}_{a;0}(g^{(1)},g^{(2)})+O(t),\ \forall i={\rm d,nd }.
 \end{equation}
Then by formulae listed in Sec. \ref{app:derivatives}, we obtain:  
 \begin{lmm}\label{lmm:derivativeE}
 Given $g^{(i)}=e^{i\vec p^{(i)}\cdot\vec\tau}e^{i\vec \theta^{(i)}\cdot\vec\tau}$, we have
  \begin{equation}
 \begin{aligned}
 (\partial_{x^{(j)}_k}E_{a;0}^{\rm d})(e^{iw\tau_3},e^{iw\tau_3})&=0,\  j=1,2,\text{ and }k=1,2\\
 (\partial_{x^{(j)}_3}E^{\rm nd}_{a;0})(e^{iw\tau_3},e^{iw\tau_3})&=0,\  j=1,2
 \end{aligned}
 \end{equation}
for all $w=p-i\theta\in\mathbb C$, where $x^{(i)}_j$ denotes $p^{(i)}_j$ or $\theta^{(i)}_j$. Moreover, consider the matrix $A$ and $B$ in Eq. \eqref{eq:hessian}. $E^{\rm nd}_{a;0}$ satisfies that 
\begin{equation}
\begin{aligned}
\nabla_{\vec x^{(1)}}^TE^{\rm nd}_{a;0}(e^{iw\tau_k},e^{iw\tau_3})A^{-1}B&=0=\nabla_{\vec x^{(2)}}^TE^{\rm nd}_{a;0}(e^{iw\tau_k},e^{iw\tau_3})A^{-1}B^T,\\
\nabla_{\vec x^{(1)}}^TE^{\rm nd}_{a;0}(e^{iw\tau_k},e^{iw\tau_3})A^{-1}B^T&=-\nabla_{\vec x^{(1)}}^TE^{\rm nd}_{a;0}(e^{iw\tau_k},e^{iw\tau_3}),\\
\nabla_{\vec x^{(2)}}^TE^{\rm nd}_{a;0}(e^{iw\tau_k},e^{iw\tau_3})A^{-1}B&=-\nabla_{\vec x^{(2)}}^TE^{\rm nd}_{a;0}(e^{iw\tau_k},e^{iw\tau_3}),\\
B^TA^{-1}\nabla_{\vec x^{(1)}}E^{\rm nd}_{a;0}(e^{iw\tau_k},e^{iw\tau_3})&=0=BA^{-1}\nabla_{\vec x^{(2)}}E^{\rm nd}_{a;0}(e^{iw\tau_k},e^{iw\tau_3}),\\
BA^{-1}\nabla_{\vec x^{(1)}}E^{\rm nd}_{a;0}(e^{iw\tau_k},e^{iw\tau_3})&=-\nabla_{\vec x^{(1)}}E^{\rm nd}_{a;0}(e^{iw\tau_k},e^{iw\tau_3}),\\
B^TA^{-1}\nabla_{\vec x^{(2)}}E^{\rm nd}_{a;0}(e^{iw\tau_k},e^{iw\tau_3})&=-\nabla_{\vec x^{(2)}}E^{\rm nd}_{a;0}(e^{iw\tau_k},e^{iw\tau_3}),
\end{aligned}
\end{equation}
where $\nabla_{\vec x^{(i)}}=(\partial_{\vec p_1^{(i)}},\partial_{\vec p_2^{(i)}},\partial_{\vec p_3^{(i)}},\partial_{\vec \theta_1^{(i)}},\partial_{\vec \theta_2^{(i)}},\partial_{\vec \theta_3^{(i)}})^T$ and $\nabla_{\vec x^{(i)}}^T$ is its transpose.  
 \end{lmm}

 The second lemma is on the inverse of the Hessian matrix \eqref{eq:hessian}. Recalling Eq. \eqref{eq:block2}, one has that the inverse of $f''_{k;\vec p_o,\vec \theta_o}$, denoted as $H_{k,\vec p_o,\vec \theta_o}^{-1}$, is
 \begin{equation}
 H_{k,\vec p_o,\vec \theta_o}^{-1}=\left(
\begin{array}{cc}
A^{-1}+A^{-1}\tilde B(D-\tilde C A^{-1}\tilde B)^{-1}\tilde C A^{-1}&-A^{-1}\tilde B (D-\tilde C A^{-1}\tilde B)^{-1}\\
-(D-\tilde C A^{-1}\tilde B)^{-1}\tilde C A^{-1}&(D-\tilde C A^{-1}\tilde B)^{-1}
\end{array} 
 \right).
 \end{equation}
Since $\tilde C A^{-1} \tilde B=0$  and $D=f''_{k-1;\vec p_o,\vec \theta_o}$, we have
\begin{equation}
 H_{k,\vec p_o,\vec \theta_o}^{-1}=\left(
\begin{array}{cc}
A^{-1}+A^{-1}\tilde BH_{k-1,\vec p_o,\vec \theta_o}^{-1} \tilde C A^{-1}&-A^{-1}\tilde B H_{k-1,\vec p_o,\vec \theta_o}^{-1}\\
-H_{k-1,\vec p_o,\vec \theta_o}^{-1}\tilde C A^{-1}&H_{k-1,\vec p_o,\vec \theta_o}^{-1}
\end{array} 
 \right).
 \end{equation}
 For $k=1$, $H_{1,\vec p_o,\vec \theta_o}^{-1}=A^{-1}$. Thus one has that $\tilde BH_{1,\vec p_o,\vec \theta_o}^{-1} \tilde C=0$ and 
 \begin{equation*}
 H_{2,\vec p_o,\vec \theta_o}^{-1} =\left(
 \begin{array}{cc}
A^{-1}&-A^{-1}\tilde B H_{1,\vec p_o,\vec \theta_o}^{-1}\\
-H_{1,\vec p_o,\vec \theta_o}^{-1}\tilde C A^{-1}&H_{1,\vec p_o,\vec \theta_o}^{-1}
\end{array} 
 \right).
 \end{equation*}
  Doing this successively,
one has $\tilde BH_{k-1,\vec p_o,\vec \theta_o}^{-1} \tilde C=0$,
 and
 \begin{equation}
 H_{k,\vec p_o,\vec \theta_o}^{-1}=\left(
\begin{array}{cc}
A^{-1}&-A^{-1}\tilde B H_{k-1,\vec p_o,\vec \theta_o}^{-1}\\
-H_{k-1,\vec p_o,\vec \theta_o}^{-1}\tilde C A^{-1}&H_{k-1,\vec p_o,\vec \theta_o}^{-1}
\end{array} 
 \right).
 \end{equation}
 Finally, $H_{k;\vec p_o,\vec \theta_o}^{-1}$ can be obtained with this recurrence relation and the initial data $H_{1,\vec p_o,\vec \theta_o}^{-1}=A^{-1}$. The result is as follows: 
\begin{lmm}\label{lmm:inversehessian}
$H_{k;\vec p_o,\vec \theta_o}^{-1}$ satisfies that
\begin{equation}
(H_{k;\vec p_o,\vec \theta_o}^{-1})_{mn}=\left\{
\begin{array}{cc}
(-1)^{|m-n|}A^{-1}(BA^{-1})^{|m-n|},& m< n\\
A^{-1},&m=n\\
(-1)^{|m-n|}A^{-1}(B^TA^{-1})^{|m-n|},& m> n\\
\end{array}
\right.
\end{equation} 
where $H_{k,\vec p_o,\vec \theta_o}^{-1}$ is arranged  as a block matrix as $f''_{k;\vec p_o,\vec \theta_o}$ in \eqref{eq:hessian}], with $(H_{k;\vec p_o,\vec \theta_o}^{-1})_{mn}$ as a block. 
\end{lmm}

Now the theorem \eqref{thm:leadingordermultiplyPH} can be proven.
\begin{proof}[Proof of Theorem \ref{thm:leadingordermultiplyPH}]
For convenience, we define $s_o=M_++N_+$.
By Theorem \ref{thm:leadingordergeneral}, $\langle \mathcal M\rangle_{z_e}$ is of order $t^{s_o}$ or higher.
Adopting the result from equation \eqref{eq:expectationaftersaddle}, $O(t^{s_o})$ only occurs when $s= s_o$ and $\mathfrak D^{(s_o)}$ is of $O(t^0)$.  

Set $s=s_o$ in the definition \eqref{eq:Ds} of $\mathfrak D^{(s)}$. 
For given $l$ and $j$, if there is one $E^{(m)}$ taking the form of $E^{\rm nd}_{a;0}+O(t)$ not being acted by derivatives, then the eventual evaluation at the critical point will vanish.
Moreover, in $\mathfrak D^{(s_o)}$, on one hand 
it contains $2s_o$ of $E^{\rm nd}_{a;0}+O(t)$, 
and on the other hand the maximum order of derivative that can act on $\prod E^{(i)}$ is  $2s_o$, which only occurs when $l=s_o$. 
Therefore, 
only when $l=s_o$ and $m=0$, all $E^{\rm nd}_{a;0}+O(t)$ are acted by derivatives.
Finally, 
\begin{equation}\label{eq:Dso}
\begin{aligned}
\mathfrak D^{(s_o)}=\left(\frac{\sin^2(\frac{\theta}{2})}{\theta^2}\right)^{|\mathcal M|-1}\frac{2^{-s_o}}{s_o!}\left[\sum_{a,b=1}^{6(|\mathcal M|-1)}H^{-1}_{ab}\frac{\partial^2}{\partial x_a\partial x_b}\right]^{s_o}\prod_{i=1}^{|\mathcal M|-1}E^{(i)}\Bigg|_{g^{(l)}=e^{iw\tau_3},\forall l},
\end{aligned}
\end{equation}
where $|\mathcal M|$ denotes the number of factors in the monomial $\mathcal M$. 

To calculate \eqref{eq:Dso}, we employ the notion introduced in Lemma \ref{lmm:derivativeE} and \ref{lmm:inversehessian} to treat $H^{-1}$ as a block matrix. Then $\sum_{a,b=1}^{6(|\mathcal M|-1)}H^{-1}_{ab}\frac{\partial^2}{\partial x_a\partial x_b}$ is rewritten as
\begin{equation}
\sum_{a,b=1}^{6(|\mathcal M|-1)}H^{-1}_{ab}\frac{\partial^2}{\partial x_a\partial x_b}=\sum_{m,n=1}^{|\mathcal M|-1} \nabla_{\vec x^{(m)}}^T\left(\hat H^{-1}_{|\mathcal M|-1;\vec p_o,\vec \theta _o}\right)_{mn}\nabla_{\vec x^{(n)}}.
\end{equation}
We then expand $\left[\sum_{m,n=1}^{|\mathcal M|-1} \nabla_{\vec x^{(m)}}^T\left(\hat H^{-1}_{|\mathcal M|-1;\vec p_o,\vec \theta _o}\right)_{mn}\nabla_{\vec x^{(n)}}\right]^{s_o}$ and let each individual term of the expansion act on $\prod_{i=1}^{|\mathcal M|-1}E^{(i)}$.
In each individual term of the expansion, it contains derivative with respect to certain $\vec x^{(q)}$. Because each $E^{\rm nd}_{a;0}+O(t)$ only depends on certain $\vec x^{(q)}$,
 we only consider the case when all derivatives are paired with all $E^{\rm nd}_{a;0}+O(t)$ with the same argument. For the other terms of the expansion, they give vanishing results because of the evaluation at the critical point. 

The procedure mentioned above is equivalent to the follows.  We first partition these $E^{\rm nd}_{a;0}+O(t)$ into ordering  pairs. Denote all possibilities of the partition as $\mathcal P$. Given a pair $(E^{\rm nd}_{a_m;0}(\vec x^{(m-1)},\vec x^{(m)})+O(t),E^{\rm nd}_{a_n;0}(\vec x^{(n-1)},\vec x^{(n)})+O(t))$ in a partition $p\in\mathcal P$.  It can be acted by 
\begin{align}
&\nabla_{\vec x^{(m)}}^T(H_{\mathcal M-1;\vec p_o,\vec \theta_o}^{-1})_{mn}\nabla_{\vec x^{(n)}},\label{eq:case1}\\
&\nabla_{\vec x^{(m-1)}}^T(H_{\mathcal M-1;\vec p_o,\vec \theta_o}^{-1})_{m-1,n}\nabla_{\vec x^{(n)}},\label{eq:case2}\\
&\nabla_{\vec x^{(m)}}^T(H_{\mathcal M-1;\vec p_o,\vec \theta_o}^{-1})_{m,n-1}\nabla_{\vec x^{(n-1)}},\label{eq:case3}\\
&\nabla_{\vec x^{(m-1)}}^T(H_{\mathcal M-1;\vec p_o,\vec \theta_o}^{-1})_{m-1,n-1}\nabla_{\vec x^{(n-1)}}\label{eq:case4}
\end{align}
According to Lemmas \ref{lmm:inversehessian} and \ref{lmm:derivativeE},
\begin{itemize}
\item[(1)] If $m< n$,  only the operator \eqref{eq:case3}  gives non-vanishing results which reads
\begin{equation}\label{eq:result1}
\begin{aligned}
&\nabla_{\vec x^{(m)}}^TE^{\rm nd}_{a_m;0}(\vec x^{(m-1)},\vec x^{(m)})(H_{\mathcal M-1;\vec p_o,\vec \theta_o}^{-1})_{m,n-1}\nabla_{\vec x^{(n-1)}}E^{\rm nd}_{a_n;0}(\vec x^{(n-1)},\vec x^{(n)})\\
=&\nabla_{\vec x^{(m)}}^TE^{\rm nd}_{a_m;0}(\vec x^{(m-1)},\vec x^{(m)})A^{-1}\nabla_{\vec x^{(n-1)}}E^{\rm nd}_{a_n;0}(\vec x^{(n-1)},\vec x^{(n)})
\end{aligned}
\end{equation}
%%%%%%
\item[(2)] If $m> n$, only the operator \eqref{eq:case2}  gives non-vanishing results which reads
\begin{equation}\label{eq:result2}
\begin{aligned}
&\nabla_{\vec x^{(m-1)}}^TE^{\rm nd}_{a_m;0}(\vec x^{(m-1)},\vec x^{(m)})(H_{\mathcal M-1;\vec p_o,\vec \theta_o}^{-1})_{m-1,n}\nabla_{\vec x^{(n)}}E^{\rm nd}_{a_n;0}(\vec x^{(n-1)},\vec x^{(n)})\\
=&\nabla_{\vec x^{(m-1)}}^TE^{\rm nd}_{a_m;0}(\vec x^{(m-1)},\vec x^{(m)})A^{-1}\nabla_{\vec x^{(n)}}E^{\rm nd}_{a_n}(\vec x^{(n-1)},\vec x^{(n)})\\
=&\nabla_{\vec x^{(n)}}^TE^{\rm nd}_{a_n;0}(\vec x^{(n-1)},\vec x^{(n)})A^{-1}\nabla_{\vec x^{(m-1)}}E^{\rm nd}_{a_m;0}(\vec x^{(m-1)},\vec x^{(m)}),
\end{aligned}
\end{equation}
where in the last step we used $A=A^T$.
\end{itemize}
It should be reminded that an evaluation at the critical point has been done in Eqs.\eqref{eq:result1} and \eqref{eq:result2}.

According to the results in Eqs. \eqref{eq:result1} and \eqref{eq:result2}, rather than partitioning the $E^{\rm nd}_{a;0}+O(t)$ into {\it ordering  }pairs, we can identify the partitions $p_1$ and $p_2$ if $p_1$ can be the same as $p_2$ by reordering each of its pairs. The set with this identification will be denoted by $\tilde{\mathcal P}$. Then we finally have
\begin{equation}
\begin{aligned}
&\left[\sum_{m,n=1}^{|\mathcal M|-1} \nabla_{\vec x^{(m)}}^T\left(\hat H^{-1}_{|\mathcal M|-1;\vec p_o,\vec \theta _o}\right)_{mn}\nabla_{\vec x^{(n)}}\right]^{s_o}\prod_{i=1}^{|\mathcal M|-1}E^{(i)}
\Bigg|_{g^{(l)}=e^{iw\tau_3},\forall l}
\\
=&2^{s_o}\sum_{p\in\tilde{\mathcal P}}\prod_{\substack{(m,n)\in p\\\text{with }m<n}}\nabla_{\vec x^{(m)}}^TE^{\rm nd}_{a_m;0}(\vec x^{(m-1)},\vec x^{(m)})A^{-1}\nabla_{\vec x^{(n-1)}}E^{\rm nd}_{a_n;0}(\vec x^{(n-1)},\vec x^{(n)})\prod_{l\in \mathcal P^{\mathbb C}}E^{\rm d}_{a_l;0}(e^{iw\tau_3},e^{iw\tau_3})+O(t)
\end{aligned}
\end{equation}
where $\mathcal P^{\mathbb C}$ denotes the set of those $E^{(i)}$ of the form $E^{\rm d}_{a;0}$. Substituting the results, we finally have
\begin{equation}
\begin{aligned}
\frac{\langle\psi_{z_e}|\mathcal M|\psi_{z_e}\rangle}{\langle\psi_{z_e}|\psi_{z_e}\rangle}=&\left(\left(2t\right)^{s_o} \frac{1}{s_o!}\sum_{p\in\tilde{\mathcal P}}\prod_{\substack{(m,n)\in p\\\text{with }m<n}}\nabla_{\vec x^{(m)}}^TE^{\rm nd}_{a_m;0}(\vec x^{(m-1)},\vec x^{(m)})A^{-1}\nabla_{\vec x^{(n-1)}}E^{\rm nd}_{a_n;0}(\vec x^{(n-1)},\vec x^{(n)})\right)\times\\
&\prod_{l\in \mathcal P^{\mathbb C}}E^{\rm d}_{a_l;0}(e^{iw\tau_3},e^{iw\tau_3})+O(t^{s_o+1})
\end{aligned}
\end{equation}
Note that in 
$$\prod_{\substack{(m,n)\in p\\\text{with }m<n}}\nabla_{\vec x^{(m)}}^TE^{\rm nd}_{a_m;0}(\vec x^{(m-1)},\vec x^{(m)})A^{-1}\nabla_{\vec x^{(n-1)}}E^{\rm nd}_{a_n;0}(\vec x^{(n-1)},\vec x^{(n)}),$$
the first argument in the most left factor and the second argument in the most right factor cannot be acted by the derivative operator. Therefore, by the similar discussion as above, one verifies that 
\begin{equation}
\begin{aligned}
\frac{\langle\psi_{z_e}|\mathcal M'|\psi_{z_e}\rangle}{\langle\psi_{z_e}|\psi_{z_e}\rangle}=&\left(\left(2t\right)^{s_o} \frac{1}{s_o!}\sum_{p\in\tilde{\mathcal P}}\prod_{\substack{(m,n)\in p\\\text{with }m<n}}\nabla_{\vec x^{(m)}}^TE^{\rm nd}_{a_m;0}(\vec x^{(m-1)},\vec x^{(m)})A^{-1}\nabla_{\vec x^{(n-1)}}E^{\rm nd}_{a_n;0}(\vec x^{(n-1)},\vec x^{(n)})\right)+O(t^{s_o+1}).
\end{aligned}
\end{equation}
Moreover, $E^{\rm d}_{a_l;(0)}(e^{i\omega \tau_3},e^{i\omega \tau_3})$ represents the leading order of the expectation value of its corresponding operator. Since these $E^{\rm d}_{a_l;(0)}(e^{i\omega \tau_3},e^{i\omega \tau_3})$ do not vanish, we conclude that $\langle \mathcal M\rangle_{z_e}$ is of order $t^{s_o}$ if and only if $\langle \mathcal M'\rangle_{z_e}$ is. Thus Eq. \eqref{eq:leadingequalmultileading} is true for the case when both sides are of order $t^{s_o}$. 
Moreover, both sides of Eq. \eqref{eq:leadingequalmultileading} are of $O(t^{M_++N_+})$ or higher, and if the leading order terms of both sides are not $O(t^{M_++N_+})$, then they both vanish at $O(t^{M_++N_+})$.
    
\end{proof}

\subsection{derivative of the matrix element of $p^{i}(e)$ and $D^{\frac 12}_{ab}(h_e)$}\label{app:derivatives}

Note that the results show below ignore the terms of order $t$ and higher. We denote $\nabla_{\vec x^{(i)}}:=\left( \nabla_{\vec p^{(i)}},\nabla_{\vec \theta^{(i)}}\right)$ and $x^*$ is the complex conjugate of $x$. 
\begin{equation}
\begin{aligned}
&\nabla_{\vec x^{(1)}} \frac{\langle \psi_{g^{(1)}}|\hat p^1_s(e)|\psi_{g^{(2)}}\rangle}{\langle \psi_{g^{(1)}}|\psi_{g^{(2)}}\rangle}=\left(-\frac{1}{2},\frac{1}{2} i \tanh \left(\frac{p}{2}\right),0,-\frac{i p \sin (\theta ) \text{csch}(p)}{2 \theta },\frac{i p \sin ^2\left(\frac{\theta }{2}\right) \text{csch}(p)}{\theta },0\right)^T\\
=&\left(\nabla_{\vec x^{(2)}} \frac{\langle \psi_{g^{(1)}}|\hat p^1_s(e)|\psi_{g^{(2)}}\rangle}{\langle \psi_{g^{(1)}}|\psi_{g^{(2)}}\rangle}\right)^*,\\
%%%%%%%%
&\nabla_{\vec x^{(1)}} \frac{\langle \psi_{g^{(1)}}|\hat p^2_s(e)|\psi_{g^{(2)}}\rangle}{\langle \psi_{g^{(1)}}|\psi_{g^{(2)}}\rangle}=\left(-\frac{1}{2} i \tanh \left(\frac{p}{2}\right),-\frac{1}{2},0,-\frac{i p \sin ^2\left(\frac{\theta }{2}\right) \text{csch}(p)}{\theta },-\frac{i p \sin (\theta ) \text{csch}(p)}{2 \theta },0\right)^T\\
=&\left(\nabla_{\vec x^{(2)}} \frac{\langle \psi_{g^{(1)}}|\hat p^2_s(e)|\psi_{g^{(2)}}\rangle}{\langle \psi_{g^{(1)}}|\psi_{g^{(2)}}\rangle}\right)^*,\\
%%%%%%%%
&\nabla_{\vec x^{(1)}} \frac{\langle \psi_{g^{(1)}}|\hat p^3_s(e)|\psi_{g^{(2)}}\rangle}{\langle \psi_{g^{(1)}}|\psi_{g^{(2)}}\rangle}=\left(0,0,-\frac{1}{2},0,0,-\frac{i}{2}\right)^T=\left(\nabla_{\vec x^{(2)}} \frac{\langle \psi_{g^{(1)}}|\hat p^3_s(e)|\psi_{g^{(2)}}\rangle}{\langle \psi_{g^{(1)}}|\psi_{g^{(2)}}\rangle}\right)^*.
\end{aligned}
\end{equation}

%%%%%%%%%%%%%%%%%%%%%%%%%%%%%%%%%%%%%%%%%%

\begin{equation}
\begin{aligned}
&\nabla_{\vec x^{(1)}} \frac{\langle \psi_{g^{(1)}}|\hat p^1_t(e)|\psi_{g^{(2)}}\rangle}{\langle \psi_{g^{(1)}}|\psi_{g^{(2)}}\rangle}=\Bigg(\frac{e^{i \theta }+e^{p-i \theta }}{2 e^p+2},\frac{i e^{-i \theta } \left(e^p-e^{2 i \theta }\right)}{2 \left(e^p+1\right)},0,\frac{e^{-i \theta } \left(-1+e^{i \theta }\right) p \left(e^{2 p}+e^{i \theta }\right)}{2 \theta  \left(e^{2 p}-1\right)},\\
&\frac{i e^{-i \theta } \left(-1+e^{i \theta }\right) p \left(e^{2 p}-e^{i \theta }\right)}{2 \theta  \left(e^{2 p}-1\right)},0\Bigg)^T=\left(\nabla_{\vec x^{(2)}} \frac{\langle \psi_{g^{(1)}}|\hat p^1_t(e)|\psi_{g^{(2)}}\rangle}{\langle \psi_{g^{(1)}}|\psi_{g^{(2)}}\rangle}\right)^*,\\
%%%%%%%%
&\nabla_{\vec x^{(1)}} \frac{\langle \psi_{g^{(1)}}|\hat p^2_t(e)|\psi_{g^{(2)}}\rangle}{\langle \psi_{g^{(1)}}|\psi_{g^{(2)}}\rangle}=\Bigg(\frac{i e^{-i \theta } \left(-e^p+e^{2 i \theta }\right)}{2 \left(e^p+1\right)},\frac{e^{i \theta }+e^{p-i \theta }}{2 e^p+2},0,\frac{i e^{-i \theta } \left(-1+e^{i \theta }\right) p \left(-e^{2 p}+e^{i \theta }\right)}{2 \theta  \left(e^{2 p}-1\right)},\\
&\frac{e^{-i \theta } \left(-1+e^{i \theta }\right) p \left(e^{2 p}+e^{i \theta }\right)}{2 \theta  \left(e^{2 p}-1\right)},0\Bigg)^T=\left(\nabla_{\vec x^{(2)}} \frac{\langle \psi_{g^{(1)}}|\hat p^2_t(e)|\psi_{g^{(2)}}\rangle}{\langle \psi_{g^{(1)}}|\psi_{g^{(2)}}\rangle}\right)^*,\\
%%%%%%%%
&\nabla_{\vec x^{(1)}} \frac{\langle \psi_{g^{(1)}}|\hat p^3_t(e)|\psi_{g^{(2)}}\rangle}{\langle \psi_{g^{(1)}}|\psi_{g^{(2)}}\rangle}=\left(0,0,\frac{1}{2},0,0,\frac{i}{2}\right)^T=\left(\nabla_{\vec x^{(2)}} \frac{\langle \psi_{g^{(1)}}|\hat p^3_t(e)|\psi_{g^{(2)}}\rangle}{\langle \psi_{g^{(1)}}|\psi_{g^{(2)}}\rangle}\right)^*.
\end{aligned}
\end{equation}
%%%%%%
 \begin{equation}
 \begin{aligned}
 &\nabla_{\vec x^{(1)}}  \frac{\langle \psi_{g^{(1)}}|D^{\frac12}_{\frac12\frac12}(h_e) |\psi_{g^{(2)}}\rangle }{\langle \psi_{g^{(1)}}|\psi_{g^{(2)}}\rangle}=\left(0,0,-\frac{1}{4} e^{-\frac{i \theta}{2} },0,0,-\frac{1}{4} i e^{-\frac{i \theta}{2} }\right)^T,\\
 &\nabla_{\vec x^{(1)}}  \frac{\langle \psi_{g^{(1)}}|D^{\frac12}_{\frac12-\frac12}(h_e) |\psi_{g^{(2)}}\rangle }{\langle \psi_{g^{(1)}}|\psi_{g^{(2)}}\rangle}=\left(-\frac{e^{\frac{i \theta }{2}} \tanh \left(\frac{p}{2}\right)}{2 p},\frac{i e^{\frac{i \theta }{2}} \tanh \left(\frac{p}{2}\right)}{2 p},0,-\frac{i \sin \left(\frac{\theta }{2}\right)}{\theta +\theta  e^p},-\frac{\sin \left(\frac{\theta }{2}\right)}{\theta +\theta  e^p},0\right)^T,\\
 &\nabla_{\vec x^{(1)}}  \frac{\langle \psi_{g^{(1)}}|D^{\frac12}_{-\frac12\frac12}(h_e) |\psi_{g^{(2)}}\rangle }{\langle \psi_{g^{(1)}}|\psi_{g^{(2)}}\rangle}=\left(-\frac{e^{-\frac{i \theta}{2} } \tanh \left(\frac{p}{2}\right)}{2 p},-\frac{i e^{-\frac{i \theta }{2}} \tanh \left(\frac{p}{2}\right)}{2 p},0,-\frac{i e^p \sin \left(\frac{\theta }{2}\right)}{\theta +\theta  e^p},\frac{e^p \sin \left(\frac{\theta }{2}\right)}{\theta +\theta  e^p},0\right)^T,\\
 &\nabla_{\vec x^{(1)}}  \frac{\langle \psi_{g^{(1)}}|D^{\frac12}_{\frac12\frac12}(h_e) |\psi_{g^{(2)}}\rangle }{\langle \psi_{g^{(1)}}|\psi_{g^{(2)}}\rangle}=\left(0,0,\frac{1}{4} e^{\frac{i \theta }{2}},0,0,\frac{1}{4} i e^{\frac{i \theta }{2}}\right)^T,\\
 &\nabla_{\vec x^{(2)}}  \frac{\langle \psi_{g^{(1)}}|D^{\frac12}_{\frac12\frac12}(h_e) |\psi_{g^{(2)}}\rangle }{\langle \psi_{g^{(1)}}|\psi_{g^{(2)}}\rangle}=\left(0,0,\frac{1}{4} e^{-\frac{i \theta }{2} },0,0,-\frac{1}{4} i e^{-\frac{i \theta }{2}}\right)^T,\\
 &\nabla_{\vec x^{(2)}}  \frac{\langle \psi_{g^{(1)}}|D^{\frac12}_{\frac12-\frac12}(h_e) |\psi_{g^{(2)}}\rangle }{\langle \psi_{g^{(1)}}|\psi_{g^{(2)}}\rangle}=\left(\frac{e^{\frac{i \theta }{2}} \tanh \left(\frac{p}{2}\right)}{2 p},-\frac{i e^{\frac{i \theta }{2}} \tanh \left(\frac{p}{2}\right)}{2 p},0,-\frac{i e^p \sin \left(\frac{\theta }{2}\right)}{\theta +\theta  e^p},-\frac{e^p \sin \left(\frac{\theta }{2}\right)}{\theta +\theta  e^p},0\right)^T,\\
 &\nabla_{\vec x^{(2)}}  \frac{\langle \psi_{g^{(1)}}|D^{\frac12}_{-\frac12\frac12}(h_e) |\psi_{g^{(2)}}\rangle }{\langle \psi_{g^{(1)}}|\psi_{g^{(2)}}\rangle}=\left(\frac{e^{-\frac{i \theta }{2}} \tanh \left(\frac{p}{2}\right)}{2 p},\frac{i e^{-\frac{i \theta }{2}} \tanh \left(\frac{p}{2}\right)}{2 p},0,-\frac{i \sin \left(\frac{\theta }{2}\right)}{\theta +\theta  e^p},\frac{\sin \left(\frac{\theta }{2}\right)}{\theta +\theta  e^p},0\right)^T,\\
 &\nabla_{\vec x^{(2)}}  \frac{\langle \psi_{g^{(1)}}|D^{\frac12}_{-\frac12-\frac12}(h_e) |\psi_{g^{(2)}}\rangle }{\langle \psi_{g^{(1)}}|\psi_{g^{(2)}}\rangle}=\left(0,0,-\frac{1}{4} e^{\frac{i \theta }{2}},0,0,\frac{1}{4} i e^{\frac{i \theta }{2}}\right)^T.
 \end{aligned}
 \end{equation}